\documentclass[times,3p]{elsarticle}
\pdfoutput=1
\usepackage[utf8]{inputenc}
\usepackage{amsmath}
\usepackage{amssymb}
\usepackage[vlined,linesnumbered,ruled]{algorithm2e}
\SetEndCharOfAlgoLine{}
\usepackage{mathptmx}
\usepackage{mathtools}
\usepackage{microtype}
\usepackage{textcomp}
\usepackage{paralist}
\usepackage{eucal}
\newcommand{\DREICNF}{\textsc{3-Cnf}}
\usepackage{tikz}
\usepackage{pgfplots}
\pgfplotsset{compat=1.4}
\pgfplotsset{filter discard warning=false}
\usepgfplotslibrary{units}
\usepackage[textsize=scriptsize]{todonotes}
\definecolor{shadecolor}{named}{lightgray}
\usepackage{hyperref}
\hypersetup{colorlinks,allcolors=blue}
\newcommand{\kk}{\alpha}

 \newcommand{\decprob}[3]{%
 {\def\descriptionlabel##1{\hspace\labelsep\quad{}\it{}##1}%
 \par\vspace{\topsep}\noindent%
 \begin{compactdesc}
 \item[\textsc{#1}]
 \item[Input:] #2
 \item[Question:] #3
\end{compactdesc}
}\vspace{\topsep}}

\newcommand{\nifty}{\textsc{Nifty}}
\newcommand{\upath}{undirected path}
\newcommand{\dpath}{directed path}
\newcommand{\ucycle}{undirected cycle}

\newcounter{casecnt}  
\newenvironment{caselist}
 {\begin{list}{Case \arabic{casecnt})~~}{\usecounter{casecnt} \labelsep=0em \labelwidth=0em \leftmargin=0em \itemsep=0em \itemindent=\parindent \parindent=1em \parskip=0em \topsep=0em}}
 {\end{list}}

\def\bigO{O}
\def\bigo{o}

\renewcommand{\cite}{\citep}
\usetikzlibrary{arrows,decorations.pathmorphing,backgrounds,positioning,fit}
\usetikzlibrary{shapes}
\usetikzlibrary{arrows,calc}

\makeatletter
\def\NAT@spacechar{~}%
\makeatother

\newdefinition{definition}{Definition}
\newtheorem{theorem}{Theorem}
\newtheorem{lemma}{Lemma}
\newproof{proof}{Proof}

\newcommand{\Sol}{\ensuremath{S}}
\newcommand{\Sin}{\ensuremath{\mathcal S}}
\newcommand{\kred}{polynomial-time many-one reduction}
\newcommand{\wei}{\ensuremath{\omega}}
\newcommand{\tw}{\ensuremath{t}}
\newcommand{\fp}{fi\-xed-pa\-ra\-me\-ter}
\newcommand{\poly}{\ensuremath{\operatorname{poly}}}
\newtheorem{observation}{Observation}
\newtheorem{construction}{Construction}{\rm\bfseries}{}
\newtheorem{proc}{Procedure}{\rm\bfseries}{}
\newtheorem{reductionrule}{Reduction Rule}

\renewcommand{\algorithmautorefname}{Algorithm}

\date {January 22, 2016}
\newcommand{\DAGP}{\textsc{DAG Partitioning}}
\newcommand{\partSet}{partitioning set}
\newcommand{\mcset}{multiway cut}
\newcommand{\MC}{\textsc{Multiway Cut}}
\DeclareRobustCommand{\NoPolyKernelAssume}{$\text{NP} \subseteq \text{coNP/poly}$}
\DeclareRobustCommand{\NegNoPolyKernelAssume}{$\text{NP} \nsubseteq \text{coNP/poly}$}

\newcommand{\N}{\ensuremath{\mathbb{N}}}

\begin{document}

\def\sectionautorefname{Section} \def\subsectionautorefname{Section}

\begin{frontmatter}
  \title{Fixed-Parameter Algorithms for DAG Partitioning\tnoteref{t1}}
  \tnotetext[t1]{A preliminary version of this article appeared in the
    \emph{Proceedings of the 8th International Conference on Algorithms
    and Complexity} (CIAC'13)~\citep{BBC+13}.
    Besides providing full proof details, this
    revised and extended version improves our $\bigO(2^k\cdot n^2)$-time
    algorithm to run in $\bigO(2^k\cdot(n+m))$~time and provides
    linear-time executable data reduction rules. Moreover, we
    experimentally evaluated the algorithm and compared it to known
    heuristics by \citet{LBK09}. We prove the latter to work optimally
    on trees.}

  \author[nsu,tub]{René van Bevern\corref{cor1}} \ead{rvb@nsu.ru}

  \author[tub]{Robert Bredereck}
  \ead{robert.bredereck@tu-berlin.de}

  \author[ulm]{Morgan Chopin} 

  \author[tub]{Sepp Hartung} %

  \author[tub]{Falk Hüffner} 

  \author[tub]{André Nichterlein} \ead{andre.nichterlein@tu-berlin.de}

  \author[prg,tub]{Ondřej Suchý} \ead{ondrej.suchy@fit.cvut.cz}

  \cortext[cor1]{Corresponding author.}

  \address[nsu]{Novosibirsk State University, Novosibirsk, Russian Federation}

  \address[tub]{Institut für Softwaretechnik und Theoretische
    Informatik, TU Berlin, Germany}

  \address[ulm]{Institut für Optimierung und Operations Research,
    Universität Ulm, Germany}

  \address[prg]{Faculty of Information Technology, Czech Technical
    University in Prague, Czech Republic}

\begin{abstract}
  Finding the origin of short phrases propagating through the web has
  been formalized by \citeauthor{LBK09}~[ACM SIGKDD~2009] as \textsc{DAG
    Partitioning}: given an arc-weighted directed acyclic graph on
  $n$~vertices and $m$~arcs, delete arcs with total weight at most~$k$
  such that each resulting weakly-connected component contains exactly
  one sink---a~vertex without outgoing arcs. \DAGP{} is NP-hard.

  We show an algorithm to solve \DAGP{} in
  $\bigO(2^k \cdot (n+m))$~time, that is, in linear time for fixed~$k$.
  We complement it with linear-time executable data reduction rules. Our
  experiments show that, in combination, they can \emph{optimally} solve
  \DAGP{} on simulated citation networks within five minutes for
  $k\leq190$ and $m$ being~$10^7$ and larger.  We use our obtained
  optimal solutions to evaluate the solution quality of
  \citeauthor{LBK09}'s heuristic.

  We show that \citeauthor{LBK09}'s heuristic works optimally on trees
  and generalize this result by showing that \DAGP{} is solvable in
  $2^{\bigO(\tw^2)}\cdot n$ time if a width-$\tw{}$ tree decomposition
  of the input graph is given.  Thus, we improve an algorithm and answer
  an open question of \citeauthor{AM12}~[WAW 2012].

  We complement our algorithms by lower bounds on the running time of
  exact algorithms and on the effectivity of data reduction.
\end{abstract}

\end{frontmatter}

\section{Introduction}
\thispagestyle{empty}

The \DAGP{} problem was introduced by \citet{LBK09} in order to analyze
how short, distinctive phrases (typically, parts or mutations of
quotations, also called \emph{memes}) spread to various news sites and
blogs. To demonstrate their approach, \citet{LBK09} collected and
analyzed phrases from 90 million articles from the time around the
United States presidential elections in 2008; the results were featured
in the New York Times~\citep{Loh09}. Meanwhile, their approach has grown
up into \nifty{}, a system that allows near real-time observation of
the rise and fall of trends, ideas, and topics in the internet
\citep{SHE+13}.

A core component in the approach of \citet{LBK09} is a heuristic to
solve the NP-hard \DAGP{} problem. They use it to cluster short phrases,
which may undergo modifications while propagating through the web, with
respect to their origins.  To this end, they create an arc-weighted
directed acyclic graph with phrases as vertices and draw an arc from
phrase~$p$ to phrase~$q$ if $p$~presumably originates from~$q$, where
the weight of an arc represents the support for this hypothesis: the
weight assigned to an arc~$(p,q)$ is chosen inversely proportional to
the time difference between~$p$ and~$q$ and their Levenshtein distance
when using words as tokens, whereas it is proportional to the total
number of documents in the corpus that contain the
phrase~$q$.\footnote{Unfortunately, \citet{LBK09} neither give a more
  precise description nor the range of their weights. All our hardness results
  even hold for unit weights, while all our algorithms work whenever the
  weight of each arc is at least one; however, we choose the weights to
  be positive integers to avoid representation issues.}
A vertex without outgoing arcs is called a \emph{sink} and can be
interpreted as the origin of a phrase. If a phrase has \dpath{}s to more
than one sink, its ultimate origin is ambiguous and, in the model of
\citet{LBK09}, at least one of the ``$p$ originates from~$q$''
hypotheses is wrong.  \citet{LBK09} introduced \DAGP{} with the aim to
resolve these inconsistencies by removing a set of arcs (hypotheses)
with least support:

\decprob{\DAGP}
{A directed acyclic graph~$G = (V,A)$ with positive integer arc
  weights~$\wei{}\colon A \rightarrow \N$ and a positive integer~$k \in
  \N$.}
{Is there a set~$\Sol \subseteq A$ with $\sum_{a\in \Sol}
  \wei{}(a)\leq k$ such that each weakly-connected component in~$G' =
  (V,A \setminus \Sol)$ has exactly one sink?}

\noindent Herein, the model of \citet{LBK09} exploits that a
weakly-connected component of a directed acyclic graph contains exactly
one sink if and only if all its vertices have \dpath{}s to only one
sink. We call a set~$\Sol \subseteq A$ such that each weakly-connected 
component in~$G' = (V,A \setminus \Sol)$ has exactly one sink a~\emph{\partSet{}}.

\citet{LBK09} showed that \DAGP{} is NP-hard and presented a heuristic
to find \partSet{}s of small weight. \citet{AM12} showed that, for
fixed~$\varepsilon>0$, even approximating the minimum weight of
a~\partSet{} within a factor of~$\bigO(n^{1-\varepsilon})$ is NP-hard. In
the absence of approximation algorithms, exact solutions to \DAGP{}
become interesting for the reason of evaluating the quality of known
heuristics alone.

We aim for solving \DAGP{} exactly using \emph{fixed-parameter
  algorithms}---a~framework to obtain algorithms to optimally solve
NP-hard problems that run efficiently given that certain parameters of
the input data are small \cite{FG06, Nie06, DF13}. A natural parameter
to consider is the minimum weight~$k$ of the \partSet{} sought, since
one would expect that wrong hypotheses have little support.

\paragraph{Known results} To date, there are only few studies on \DAGP.
\Citet{LBK09} showed that \DAGP{} is NP-hard and present heuristics. \citet{AM12} showed that, on $n$-vertex graphs,
\DAGP{} is hard to approximate in the sense that if $\text{P} \neq
\text{NP}$, then there is no polynomial-time factor-$(n^{1-\varepsilon})$
approximation algorithm for any fixed $\varepsilon > 0$,
even if the input graph has unit-weight arcs, maximum outdegree three,
and only two sinks. Moreover, \citet{AM12} showed that \DAGP{} can be
solved in~$2^{\bigO(\tw{}^2)} \cdot n$ time if a width-$\tw{}$ path
decomposition of the input graph is given.

\DAGP{} is very similar to the well-known NP-hard \MC{} problem~\citep{DJP+94}: given an undirected edge-weighted graph and a
subset of the vertices called \emph{terminals}, delete edges of total
weight at most~$k$ such that each terminal is separated from all others.
\DAGP{} can be considered as \MC{} with the sinks being terminals and the
additional constraint that not all arcs outgoing from a vertex may be
deleted, since this would create a new sink. \Citet{Xia10} gave an
algorithm to solve \MC{} in $\bigO(2^k\cdot \min(n^{2/3},m^{1/2})\cdot
nm)$~time. Interestingly, in contrast to \DAGP{}, \MC{} is
constant-factor approximable (see, e.\,g., \citet{KKS+04}).
 
\paragraph{Our results} We provide algorithmic as well as intractability results.  On the algorithmic side, we present an $\bigO(2^k\cdot (n+m))$~time algorithm for \DAGP{} and complement it with linear-time executable data reduction rules. We experimentally evaluated both and, in combination, they solved instances with~$k\leq190$~optimally within five minutes, the number of input arcs being~$10^7$ and larger. Moreover, we use the optimal solutions found by our algorithm to evaluate the quality of \citet{LBK09}'s heuristic and find that it finds optimal solutions for most instances that our algorithm solves quickly, but performs worse by a factor of more than two on other instances.

Also, we give an algorithm that solves \DAGP{} in~$2^{\bigO(\tw{}^2)} \cdot n$
time if a width-$\tw{}$ tree decomposition of the input graph is
given. We thus answer an open question by \citet{AM12}. Since every
width-$\tw$ path decomposition is a width-$\tw$ tree decomposition but
not every graph allowing for a width-$\tw$ tree decomposition allows for
a width-$\tw$ path decomposition, our algorithm is an improvement over
the $2^{\bigO({\tw{}}^2)} \cdot n$-time algorithm of \citet{AM12}, which
requires a path decomposition as input.

On the side of intractability results, we strengthen the NP-hardness
results of \citet{LBK09} and \citet{AM12} to graphs of diameter two and
maximum degree three and we show that our $\bigO(2^k\cdot (n+m))$~time
algorithm cannot be improved to $\bigO(2^{\bigo(k)} \cdot \poly(n))$~time unless
the Exponential Time Hypothesis fails.  Moreover, we show that \DAGP{} does not admit polynomial-size problem kernels with respect to~$k$
unless \NoPolyKernelAssume.

\paragraph{Organization of this paper} %
In \autoref{sec:prelims}, we introduce necessary notation and two basic structural
observations for \DAGP{} that are important in our proofs.

In \autoref{sec:smallsol}, we present our $\bigO(2^k\cdot (n+m))$~time algorithm and
its experimental evaluation.
With the help of the optimal solutions computed by our algorithm we also evaluate the
quality of a heuristic presented by \citet{LBK09}.
Moreover, we discuss the limits of parameterized algorithms and problem kernelization
for \DAGP{} parameterized by~$k$.

\autoref{sec:tw} presents our $2^{\bigO(\tw^2)}\cdot n$~time algorithm. It follows that
\DAGP{} is linear-time solvable when at least one of the parameters~$k$ or~$\tw$ is fixed.
We further show that the heuristic presented by \citet{LBK09} works optimally on trees.

\autoref{sec:classical} then shows that other parameters are not as helpful in solving
\DAGP{}: \DAGP{} remains NP-hard even when graph parameters like the diameter or maximum
degree are constants.

\section{Preliminaries and basic observations}\label{sec:prelims}

We consider finite simple directed graphs~$G=(V,A)$ with \emph{vertex
  set}~$V(G):=V$ and \emph{arc set}~$A(G):=A\subseteq V\times V$, as
well as finite simple undirected graphs~$G=(V,E)$ with vertex set~$V$
and \emph{edge set}~$E(G):=E\subseteq\{\{u,v\} \mid u,v\in V\}$.
For a directed graph~$G$, the \emph{underlying undirected graph}~$G'$
is the graph that has undirected edges in the places where~$G$ has arcs,
that is, $G'= (V(G), \{\{v,w\} : (v,w) \in A(G)\})$.
We will use $n$~to denote the number of vertices and $m$~to denote
the number of arcs or edges of a graph.

For a (directed or undirected) graph~$G$, we denote by $G \setminus A'$
the subgraph obtained by removing from~$G$ the arcs or edges in~$A'$ and
by $G-V'$ the subgraph obtained by removing from~$G$ the vertices
in~$V'$. For $V'\subseteq V$, we denote by $G[V']:=G-(V\setminus V')$
the subgraph of~$G$ \emph{induced} by the vertex set~$V'$.

The set of \emph{out-neighbors} and \emph{in-neighbors} of a vertex~$v$
in a directed graph is~$N^+ (v) := \{u \mid (v,u) \in A\}$ and $N^- (v) :=
\{u \mid (u,v) \in A\}$, respectively. The \emph{outdegree}, the
\emph{indegree}, and the \emph{degree} of a vertex~$v \in V$ are $d^+
(v) := |N^+ (v)|$, $d^- (v) := |N^- (v)|$, and $d(v) := d^+ (v) + d^-
(v)$, respectively. A vertex~$v$ is a \emph{sink} if $d^+(v)=0$; it is
\emph{isolated} if $d(v)=0$.

A \emph{path} of length~$\ell-1$ from~$v_1$ to~$v_\ell$ in an undirected
graph~$G$ is a tuple~$(v_1,\dots,v_\ell)\in V^\ell$ such that
$\{v_i,v_{i+1}\}$ is an edge in~$G$ for~$1\leq i\leq\ell-1$. An
\emph{\upath{}} in a directed graph is a path in its underlying undirected
graph. A \emph{\dpath{}} of length~$\ell-1$ from~$v_1$ to~$v_\ell$
in a directed graph~$G$ is a tuple~$(v_1,\dots,v_\ell)\in V^\ell$ such
that $(v_i,v_{i+1})$ is an arc in~$G$ for~$1\leq i\leq\ell-1$.

We say that $u\in V$ \emph{can reach} $v\in V$ or that $v$ \emph{is
  reachable from}~$u$ in~$G$ if there is a \dpath{} from~$u$ to~$v$
in~$G$. We say that $u$~and~$v$ are \emph{connected} if there is a (not
necessarily directed) path from~$u$ to~$v$ in~$G$.  In particular,
$u$~is reachable from~$u$ and connected to~$u$. We use \emph{connected
  component} as an abbreviation for \emph{weakly connected component},
that is, a maximal set of pairwise connected vertices. The
\emph{diameter} of~$G$ is the maximum length of a shortest path between
two different vertices in the underlying undirected graph of~$G$.

\paragraph{Fixed-parameter algorithms} The main idea in \fp{} algorithms
is to accept the super-polynomial running time, which is seemingly
inevitable when optimally solving NP-hard problems, but to restrict it to
one aspect of the problem, the \emph{parameter}. More precisely, a
problem $\Pi$ is \emph{\fp{} tractable~(FPT)} with respect to a
parameter~$k$ if there is an algorithm solving any instance of $\Pi$ with
size~$n$ in $f(k) \cdot \poly(n)$~time for some computable function~$f$
\citep{FG06,Nie06,DF13,CFK+15}.
Such an algorithm is called \emph{fixed-parameter algorithm}.
Since \citet{SHE+13} point out that the input instances of \DAGP{} can
be so large that even running times quadratic in the input size are
prohibitively large, we focus on finding algorithms that run in
\emph{linear time} if the parameter~$k$ is a constant.
An important ingredient of our algorithms is linear-time data reduction,
which recently received increased interest since data reduction is
potentially applied to large input data \citep{PDS09,BHK+12,Hag12,Bev14b,
Kam15,FK15}.

\paragraph{Problem kernelization} One way of deriving \fp{} algorithms is \emph{(problem) kernelization} \citep{GN07,Bod09}. As a formal approach of describing efficient data reduction that preserves optimal solutions, problem kernelization is a powerful tool for attacking NP-hard problems.  A \emph{kernelization algorithm} consists of \emph{data reduction rules} that, applied to any instance $x$ with parameter~$k$, yield an instance $x'$ with parameter~$k'$ in time polynomial in $|x|+k$ such that $(x,k)$~is a yes-instance if and only if $(x',k')$~is a yes-instance, and if both $|x'|$ and $k'$ are bounded by some functions $g$ and $g'$ in~$k$, respectively. The function~$g$ is referred to as the \emph{size} of the \emph{problem kernel}~$(x',k')$.

Note that it is the parameter that allows us to measure the
effectiveness of polynomial-time executable data reduction, since a
statement like ``the data reduction shrinks the input by a
factor~$\alpha$'' would imply that we can solve NP-hard problems in
polynomial time.

From a practical point of view, problem kernelization is potentially applicable to speed up exact and heuristic algorithms to solve a problem.  Since kernelization is applied to shrink potentially \emph{large} input instances, recently the running time of kernelization has stepped into the focus and linear-time kernelization algorithms have been developed for various NP-hard problems \citep{PDS09,BHK+12,Hag12,Bev14b,Kam15,FK15}.

\paragraph{Two basic observations.} The following easy to prove structural observations will be exploited in many proofs of our work. 
The first observation states that a minimal \partSet{} does not introduce new sinks.
\begin{observation}\label{lem:no_new_sinks}
  Let $G$~be a directed acyclic graph and~$\Sol{}$ be a
  minimal \partSet{} for~$G$. Then, a vertex is a sink in
  $G\setminus\Sol{}$ if and only if it is a sink in~$G$.
\end{observation}
\begin{proof}
  Clearly, deleting arcs from a directed acyclic graph cannot turn a sink into a non-sink. Therefore, it remains to show that every sink in~$G\setminus\Sol{}$ was a sink already in~$G$. Towards a contradiction, assume that there is a vertex~$s$ that is a sink in~$G\setminus\Sol{}$ but not in~$G$. Then, there is an arc~$(s, v)$ in~$G$ for some vertex~$v$ of~$G$. Let $C_v$ and~$C_s$ be the connected components in~$G\setminus\Sol{}$ containing~$v$ and~$s$, respectively, and let~$s_v$ be the sink in~$C_v$. Then, for $\Sol{}':=\Sol{}\setminus\{(s,v)\}$, the connected component $C_v \cup C_s$ in~$G\setminus\Sol{}'$ has only one sink, namely $s_v$.  Thus, $\Sol{}'$~is also a \partSet{} for~$G$, but $\Sol{}' \subsetneq \Sol{}$, a contradiction to $\Sol{}$~being minimal.  \qed
\end{proof}

The second observation is that each vertex in a directed acyclic graph is connected to exactly one sink if and only if it can reach that sink.

\begin{observation}\label{lem:dirundirequiv}
  Let $G$~be a directed acyclic graph. An arc set~$\Sol{}$ is
  a \partSet{} for~$G$ if and only if each vertex in~$G$ can reach
  exactly one sink in~$G\setminus\Sol{}$.
\end{observation}
\newcommand{\mydag}{directed acyclic graph}
\begin{proof}
  If $\Sol{}$~is a \partSet{} for~$G$, then, by definition, each
  connected component of~$G\setminus\Sol{}$ contains exactly one
  sink. Therefore, each vertex in~$G\setminus\Sol{}$ can reach
  \emph{at most} one sink.  Moreover, since $G\setminus\Sol{}$ is a
  \mydag{} and each vertex in a \mydag{} can reach \emph{at least} one
  sink, it follows that each vertex in~$G\setminus\Sol{}$ can reach
  \emph{exactly} one sink.

  Now, assume that each vertex in~$G\setminus\Sol{}$ can
  reach exactly one sink.  We show that each connected component
  of~$G\setminus\Sol{}$ contains exactly one sink.  For the sake of
  contradiction, assume that a connected component~$C$
  of~$G\setminus\Sol{}$ contains multiple sinks~$s_1,\dots,s_t$.
  For~$i\in[t]$, let $A_i$~be the set of vertices that reach~$s_i$.
  These vertex sets are pairwise disjoint and, since every vertex in a
  directed acyclic graph reaches some sink, they partition the vertex
  set of~$C$.

  Since $C$~is a connected component, there are~$i,j\in[t]$ with~$i\ne
  j$ and some arc~$(v,w)$ in~$G\setminus\Sol{}$ from some vertex~$v\in
  A_i$ to some vertex~$w\in A_j$.  This is a contradiction, since
  $v$~can reach~$s_i$ as well as~$s_j$. \qed
\end{proof}

\section{Parameter weight of the \partSet{} sought}\label{sec:smallsol}

This section investigates the influence of the parameter ``weight~$k$ of
the \partSet{} sought'' on \DAGP{}. First, in \autoref{sec:klintime},
we present an $O(2^k\cdot(n+m))$-time algorithm and design linear-time
executable data reduction rules. 
In \autoref{sec:experiments}, we experimentally evaluate the algorithm
and data reduction rules and---with the help of the optimal solutions
computed by our algorithm---we also evaluate the quality of a heuristic
presented by \citet{LBK09}.
In \autoref{sec:kkern}, we  investigate the question whether the
provided data reduction rules might have a provable shrinking effect on
the input instance in form of a polynomial-size problem kernel.
We will see that, despite the fact that data reduction rules work very
effectively in experiments, polynomial-size problem kernels for \DAGP{}
do not exist under reasonable complexity-theoretic assumptions. 
Moreover, \autoref{sec:kkern} will also show that the algorithm presented
in \autoref{sec:klintime} is essentially optimal.

\subsection{Constant-weight \partSet{}s in linear
  time}\label{sec:klintime}

We now present an algorithm to compute \partSet{}s of
weight~$k$ in $O(2^k\cdot(n+m))$~time. Interestingly, although both problems are
NP-hard, it will turn out that \DAGP{} is substantially easier to solve  exactly
than the closely related \MC{} problem, for which a sophisticated
algorithm running in $\bigO(2^k \min(n^{2/3},m^{1/2})nm)$~time was given by
\citet{Xia10}. This is in contrast to \MC{} being constant-factor
approximable \citep{KKS+04}, while \DAGP{} is inapproximable unless
P${}={}$NP \citep{AM12}.

The main structural advantage of \DAGP{} over \MC{} is the alternative characterization of \partSet{}s given in \autoref{lem:dirundirequiv}: we only have to decide which sink each vertex~$v$ will reach in the optimally partitioned graph.  To this end, we first decide the question for all out-neighbors of~$v$.  This natural processing order allows us to solve \DAGP{} by the simple search tree algorithm shown in \autoref{alg:simple-st}, which we explain in the proof of the following~theorem.

\begin{theorem}
  \label{thm:search-tree-DAGP} \autoref{alg:simple-st} solves \DAGP{} in
  $\bigO(2^k \cdot (n+m))$ time.
\end{theorem}

\begin{proof}
  \autoref{alg:simple-st} is based on recursive branching and computes a \partSet{}~$\Sol{}$ of weight at most~$k$ for a directed acyclic graph~$G=(V,A)$.  It exploits the following structural properties of a minimal \partSet{}~\(\Sol\): by \autoref{lem:dirundirequiv}, a vertex~$v$ is connected to a sink~$s$ in~$G\setminus\Sol{}$ if and only if it can reach that sink~$s$ in~$G\setminus\Sol{}$.  Thus, consider a vertex~$v$ of~$G$ and assume that we know, for each out-neighbor~$w$ of~$v$ in~$G$, the sink~$s$ that $w$~can reach in~$G\setminus\Sol{}$. We call a sink~$s$ of~$G$ \emph{feasible} for~$v$ if an out-neighbor of~$v$ in~$G$ can reach~$s$ in~$G\setminus \Sol{}$. Let $D$~be the set of feasible sinks for~$v$. Since $v$~may be connected to only one sink in~$G\setminus\Sol{}$, at least~$|D|-1$ arcs outgoing from~$v$ are deleted by~$\Sol{}$. However, $\Sol{}$~does not disconnect $v$~from all sinks in~$D$, since then $\Sol{}$ would delete all arcs outgoing from~$v$, contradicting \autoref{lem:no_new_sinks}. Hence, exactly one sink~$s\in D$ must be reachable by~$v$ in~$G\setminus\Sol{}$. For each such sink~$s\in D$, the \partSet{}~$\Sol{}$ has to delete at least~$|D|-1$ arcs outgoing from~$v$.  We simply try out all these possibilities, which gives rise to the following search tree algorithm.
  \begin{algorithm}[t]
    \caption{Compute a \partSet{} of weight at most~$k$.}
    \label{alg:simple-st}
    \KwIn{A directed acyclic graph~$G=(V,A)$ with arc weights~$\wei{}$ and a positive integer~$k$.}
    \KwOut{A \partSet{}~$S$ of weight at most~$k$ if it exists; otherwise `no'.}

    \vspace{1em}
    $(v_1,v_2,\dots,v_n)\gets{}$reverse topological order of the vertices
    of~$G$\;
    $L\gets{}$array with $n$~entries\; %

    searchtree($1,\emptyset$)\tcp*{\textrm{start with vertex~$v_1$ and~$S=\emptyset$}}
    \textbf{output} `no'\tcp*{\textrm{there is no partitioning set of weight at most~$k$}}
    \vspace{1em}
    \SetKwProg{myproc}{Procedure}{}{}
    \myproc{searchtree($i$, $\Sol{}$)\tcp*[f]{\textrm{vertex counter~$i$; (partial) partitioning set~$S$}}}{
      \While(\nllabel{lin:skipverts}){$v_i$ is a sink~$s$ or there is a sink~$s$ such that $\forall w\in N^+(v_i): L[w]=s$}{
        $L[v_i]\gets s$\nllabel{lin:asself}\tcp*{\textrm{associate~$v_i$ with sink~$s$}}
        $i\gets i+1$\tcp*{\textrm{continue with next vertex}}
     }
     \If(\tcp*[f]{\textrm{all vertices have been handled, $\Sol{}$ is a \partSet{}}}){$i>n$}{
     \lIf(\nllabel{lin:outputS}){$\wei(\Sol{})\leq k$}{\textbf{output}~$\Sol{}$\tcp*[f]{\textrm{a partitioning set of weight at most~$k$ has been found}}}
}
     \Else{%
       $D\gets\{L[w]\mid w\in N^+(v_i)\}$\nllabel{lin:D}\tcp*{\textrm{the set of feasible sinks for~$v_i$}}
       \If(\nllabel{dsmallcheck}\tcp*[f]{\textrm{check whether we are allowed to delete $|D|-1$~arcs}}){$|D|-1\leq k-\wei(\Sol{})$}{
         \ForEach(\nllabel{lin:branch}\tcp*[f]{\textrm{try to associate~$v_i$ to each feasible sink~$s$}}){$s\in D$}{
           $L[v_i]\gets s$\nllabel{lin:asss}\;
           $S'\gets S\cup\{(v_i,w)\mid w\in N^+(v_i)\text{ and } L[w]\ne s\}$\nllabel{lin:Sprime}\;
           searchtree($i+1, S'$)\nllabel{lin:reccall}\;
       }
     }
   }
 }
  \end{algorithm}

  \autoref{alg:simple-st} starts with~$\Sol{}=\emptyset{}$ and processes
  the vertices of~$G$ in reverse topological order, that is, each vertex
  is processed by procedure~''searchtree'' after its out-neighbors. The
  procedure exploits the invariant that, when processing a vertex~$v$,
  each out-neighbor~$w$ of~$v$ has already been \emph{associated} with
  the sink~$s$ that it reaches in~$G\setminus\Sol{}$ (in terms of
  \autoref{alg:simple-st}, $L[w]=s$). It then tries all possibilities of
  associating $v$~with a feasible sink and augmenting~$\Sol{}$
  accordingly, so that the invariant also holds for the vertex processed after~$v$.

  Specifically, if $v$~is a sink, then \autoref{lin:asself} associates
  $v$~with itself. If all out-neighbors of~$v$ are associated with the
  same sink~$s$, then \autoref{lin:asself} associates
  $v$~with~$s$. Otherwise, \autoref{lin:D} computes the set~$D$ of
  feasible sinks for~$v$. In lines~\ref{lin:branch}--\ref{lin:reccall}, the algorithm
  branches into all possibilities of associating~$v$ with one of
  the~$|D|$ feasible sinks~$s\in D$ (by way of setting~$L[v]\gets s$ in
  \autoref{lin:asss}) and augmenting~$\Sol{}$ so
  that $v$~only reaches~$s$ in~$G\setminus\Sol{}$. That is, in each of
  the $|D|$~branches, it adds to~$\Sol{}$ the arcs
  outgoing from~$v$ to the out-neighbors of~$v$ that are associated with
  sinks different from~$s$ (the weight of~$\Sol{}$ increases by at
  least~$|D|-1$). Then, \autoref{lin:reccall} continues with the next
  vertex in the reverse topological order. After processing the last
  vertex, each vertex of~$G\setminus\Sol{}$ can reach exactly one sink,
  that is, $\Sol{}$~is a \partSet{}. If a branch finds a \partSet{} with
  weight at most~$k$, \autoref{lin:outputS} outputs it.

 \looseness=-1 We analyze the running time of this algorithm. To this end, we first bound the total number of times that procedure ``searchtree'' is called. To this end, we first analyze the number of \emph{terminal calls}, that is, calls that do not recursively call the procedure. Let $T(\alpha)$~denote the maximum possible number of terminal calls caused by the procedure ``searchtree'' when called with a set~$\Sol{}$ satisfying~$\wei(\Sol{})\geq \kk$, including itself if it is a terminal call.  Note that procedure ``searchtree'' calls itself only in \autoref{lin:reccall}, that is, for each sink~$s$ of some set~$D$ of feasible sinks with~$1\leq |D|-1\leq k-\kk$, it calls itself with a set~$S'$ of weight at least~$\alpha+|D|-1$. Thus, we have
  \[T(\kk)\leq |D|\cdot T(\kk+|D|-1).\] We now inductively show
  that~$T(\alpha)\leq 2^{k-\alpha}$ for~$0\leq\alpha\leq k$.  Then, it
  follows that there is a total number of $T(0)\leq 2^k$~terminal
  calls.  For the induction base case, observe that $T(k)=1$, since,
  if procedure ``searchtree'' is called with a set~$\Sol{}$ of weight
  at least~$k$, then any recursive call is prevented by the check in
  \autoref{dsmallcheck}.  Now, assume that $T(\kk')\leq
  2^{k-\kk'}$~holds for all~$\kk'$ with~$\alpha\leq \kk'\leq k$.  We show
  $T(\alpha-1)\leq 2^{k-(\alpha -1)}$ by exploiting $2\leq|D|\leq
  2^{|D|-1}$ as follows:
  \begin{align*}
    T(\alpha-1)&\leq |D|\cdot T(\alpha-1+|D|-1)
    \leq |D|\cdot 2^{k-(\alpha-1)-(|D|-1)} \leq |D|\cdot
    \frac{2^{k-(\alpha-1)}}{2^{|D|-1}} \leq 2^{k-(\alpha-1)}.
  \end{align*}
  It follows that there are at most~$T(0)=2^k$ terminal calls to
  procedure ``searchtree''. 

  In order to bound the total number of calls to procedure ``searchtree'', observe the following: if each inner node of a tree has at least two children, then the number of inner nodes in a tree is at most its number of leaves.  Now, since procedure ``searchtree'' calls itself only in \autoref{lin:reccall}, that is, for each sink~$s$ of some set~$D$ of feasible sinks with~$1\leq |D|-1\leq k-\kk$, each non-terminal call causes at least two new calls.  Thus, since there are $\bigO(2^k)$~terminal calls, there are also $\bigO(2^k)$~non-terminal calls. 

  It follows that there are $\bigO(2^k)$~total calls of procedure ``searchtree''.  For each such call, we iterate, in the worst case, over all out-neighbors of all vertices in the graph in lines~\ref{lin:skipverts}--\ref{lin:D}, which works in $\bigO(n+m)$~time. Moreover, for each call of procedure ``searchtree'', we compute a set~$S'$ in \autoref{lin:Sprime} in $\bigO(n+m)$~time. Hence, a total amount of~$\bigO(2^k\cdot (n+m))$~time is spent in procedure ``searchtree''. Initially, \autoref{alg:simple-st} uses $\bigO(n+m)$~time to compute a reverse topological ordering~\citep[Section~22.4]{CLRS01}.  \qed
\end{proof}

\noindent The experimental results in \autoref{sec:experiments} will show that
\autoref{alg:simple-st} alone cannot solve even moderately large
instances. Therefore, we complement it by linear-time executable data
reduction rules that will allow for a significant speedup. The following data
reduction rule is illustrated in \autoref{fig:halfblind}.

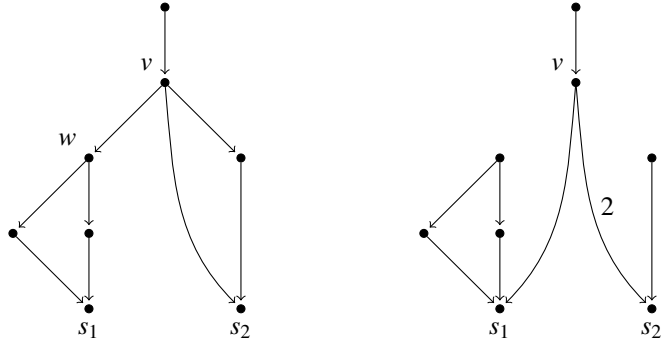
\begin{figure}
  \centering
  \begin{tikzpicture}[shorten >= 0.5mm]
    \tikzstyle{vertex}=[circle,draw,fill=black,minimum size=3pt,inner sep=0pt]

    \node[vertex,label=below:$s_1$] (s1) at (1,0) {};
    \node[vertex,label=below:$s_2$] (s2) at (3,0) {};

    \node[vertex] (a) at (0,1) {};
    \node[vertex] (b) at (1,1) {};
    \node[vertex,label=above left:$w$] (c) at (1,2) {};
    \node[vertex,label=above left:$v$] (d) at (2,3) {};
    \node[vertex] (e) at (2,4) {};
    \node[vertex] (f) at (3,2) {};

    \draw[->] (a)--(s1);
    \draw[->] (b)--(s1);
    \draw[->] (c)--(b);
    \draw[->] (d)--(c);
    \draw[->] (d) to[out=-85](s2);
    \draw[->] (c) to(a);
    \draw[->] (e)--(d);
    \draw[->] (d)--(f);
    \draw[->] (f)--(s2);
  \end{tikzpicture}\hspace{2cm}
  \begin{tikzpicture}[shorten >= 0.5mm]
    \tikzstyle{vertex}=[circle,draw,fill=black,minimum size=3pt,inner sep=0pt]

    \node[vertex,label=below:$s_1$] (s1) at (1,0) {};
    \node[vertex,label=below:$s_2$] (s2) at (3,0) {};

    \node[vertex] (a) at (0,1) {};
    \node[vertex] (b) at (1,1) {};
    \node[vertex] (c) at (1,2) {};
    \node[vertex,label=above left:$v$] (d) at (2,3) {};
    \node[vertex] (e) at (2,4) {};
    \node[vertex] (f) at (3,2) {};

    \draw[->] (a)--(s1);
    \draw[->] (b)--(s1);
    \draw[->] (c)--(b);
    \draw[->] (d) to[out=-95,in=45](s1);
    \draw[->] (d) to[out=-85] node[right]{$2$} (s2);
    \draw[->] (c) to(a);
    \draw[->] (e)--(d);
    \draw[->] (f)--(s2);
  \end{tikzpicture}

  \caption{The left side shows an input graph with unit weights to which \autoref{halfblind} is applicable.  The right side shows the same graph to which \autoref{halfblind} has been applied as often as possible. Since~$v$ can reach multiple sinks and its out-neighbors each can reach only one sink, its arcs got redirected. Unlabeled arcs have weight one.}
  \label{fig:halfblind}
\end{figure}

\begin{reductionrule}\label{halfblind}
  If there is an arc~$(v,w)$ such that $w$~can reach exactly one
  sink~$s\neq w$ and $v$~can reach multiple sinks, then
  \begin{itemize}
  \item if there is no arc~$(v,s)$, then add it with weight~$\wei(v,w)$,
  \item otherwise, increase $\wei{}(v,s)$ by $\wei{}(v,w)$, and
  \end{itemize}
  delete the arc~$(v,w)$.
\end{reductionrule}

\noindent Note that, in the formulation of the data reduction rule, both
$v$~and~$w$ may be \emph{connected} to an arbitrary number of sinks 
by an \upath{}. However, we require that $w$~\emph{can reach}
exactly one sink and that $v$~\emph{can reach} multiple sinks, that is,
using a \dpath{}.

\begin{lemma}
  Let $(G,\wei{},k)$ be a \DAGP{} instance and consider the
  graph~$G'$ with weights~$\wei{}'$ output by \autoref{halfblind}
  applied to an arc~$(v,w)$ of~$G$. Then $(G,\wei{},k)$~is a
  yes-instance if and only if~$(G',\wei{}',k)$~is a yes-instance.
\end{lemma}

\begin{proof}
  First, assume that $(G,\wei{},k)$ is a yes-instance and that~$\Sol{}$
  is a minimal \partSet{} of weight at most~$k$ for~$G$. We show how to
  transform~$\Sol{}$ into a \partSet{} of equal weight for~$G'$. We
  distinguish two cases: either $\Sol{}$ disconnects~$v$ from~$s$ or
  not, where $s$~is the only sink that~$w$ can reach.
  \begin{caselist}
  \item Assume that $S$~disconnects~$v$ from~$s$. Note that every
    subgraph of a directed acyclic graph is again a directed acyclic
    graph and that every vertex in a directed acyclic graph is not only
    connected to, but also can reach some sink. Hence, by
    \autoref{lem:no_new_sinks}, $\Sol{}$ cannot disconnect~$w$ from~$s$,
    since $w$~can only reach~$s$ in~$G$ and would have to reach some
    other, that is, new sink in~$G\setminus\Sol{}$. It follows that
    $\Sol{}$~contains the arc~$(v,w)$. Now, however,
    $S':=(S\setminus\{v,w\})\cup\{(v,s)\}$ is a \partSet{} for~$G'$,
    since $G\setminus S=G'\setminus S'$. Moreover, since
    $\wei{}'(v,s)=\wei{}(v,s)+\wei{}(v,w)$, we have
    $\wei{}'(S')=\wei{}(S)$, where we, for convenience, declare $\wei{}(v,s)=0$ if
    there is no arc~$(v,s)$ in~$G$.

  \item Assume that $\Sol{}$~does not disconnect~$v$ from~$s$ and, for
    the sake of a contradiction, that $\Sol{}$ is not a \partSet{}
    for~$G'$. Observe that $\Sol{}$~contains neither~$(v,w)$
    nor~$(v,s)$, because it is a minimal \partSet{} and does not
    disconnect~$v$ from~$s$. Therefore, $G'\setminus S$ differs
    from~$G\setminus S$ only in the fact that~$G'\setminus S$ does not
    have the arc~$(v,w)$ but an arc~$(v,s)$ that was possibly not
    present in~$G\setminus S$. Hence, since $\Sol{}$ is a \partSet{}
    for~$G$ but not for~$G'$, two sinks are connected to each other
    in~$G'\setminus\Sol{}$ via an \upath{} using the arc~$(v,s)$. Thus,
    one of the two sinks is~$s$ and the \upath{} consists of~$(v,s)$ and
    a subpath~$p$ between~$v$ and some sink~$s'$. Then, however, $s$~is
    connected to~$s'$ also in~$G\setminus\Sol{}$ via an \upath{}
    between~$s$ and~$w$ ($\Sol{}$~cannot disconnect~$s$ from~$w$ by
    \autoref{lem:no_new_sinks}), the arc~$(v,w)$ and the \upath{}~$p$
    from~$v$ to~$s'$. This contradicts $\Sol{}$~being a \partSet{}
    for~$G$.  We conclude that $S$~is a \partSet{} for~$G'$.  Moreover, since $S$~contains neither~$(v,w)$ nor~\((v,s)\), one has $\wei(S)=\wei'(S)$.
  \end{caselist}

  Now, assume that~$(G',\wei{}',k)$ is a yes-instance and that~$\Sol{}$
  is a minimal \partSet{} of weight at most~$k$ for~$G'$. We show how to
  transform~$\Sol{}$ into a \partSet{} of equal weight for~$G$. Again,
  we distinguish between two cases: either $\Sol{}$ disconnects~$v$
  from~$s$ or not.

  \begin{caselist}
  \item Assume that $S$~disconnects~$v$ from~$s$. Then, $(v,s)\in
    S$. Now, $S':=S\cup\{(v,w)\}$ is a \partSet{} for~$G$,
    since~$G\setminus S'=G'\setminus S$. Moreover, since
    $\wei{}'(v,s)=\wei{}(v,s)+\wei{}(v,w)$, we have
    $\wei{}(S')=\wei{}'(S)$, where we assume that $\wei{}(v,s)=0$ if
    there is no arc~$(v,s)$ in~$G$.

  \item Assume that $S$~does not disconnect~$v$ from~$s$ and, for the
    sake of a contradiction, assume that~$\Sol{}$ is not a \partSet{}
    for~$G$. Then, since $\Sol{}$~is minimal, $\Sol{}$ does not
    contain~$(v,s)$. Now, observe that $G\setminus S$
    and~$G'\setminus\Sol{}$ differ only in the fact
    that~$G\setminus\Sol{}$ has an additional arc~$(v,w)$ and that,
    possibly, $(v,s)$ is missing. Hence, since~$\Sol{}$ is a \partSet{}
    for~$G'$ but not for~$G$, there is an \upath{} between two sinks
    in~$G\setminus\Sol{}$ through~$(v,w)$. Because~$\Sol{}$ by
    \autoref{lem:no_new_sinks} cannot disconnect~$w$ from~$s$, one of
    these sinks is~$s$ and the \upath{} consists of a subpath between~$s$
    and~$w$, the arc~$(v,w)$, and a subpath~$p$ between~$v$ and a
    sink~$s'$. Then, however, $s$~and~$s'$ are also connected
    in~$G'\setminus\Sol{}$ via the arc~$(v,s)$ and the subpath~$p$
    between~$v$ and~$s'$. This contradicts $\Sol{}$~being a \partSet{}
    for~$G'$.  Finally, since $S$~does not contain~$(v,s)$, one has $\wei(S)=\wei'(S)$.\qed
  \end{caselist}
\end{proof}

\noindent After applying \autoref{halfblind} exhaustively, that is, as
often as it is applicable, we apply a second data reduction rule, which
is illustrated in \autoref{fig:kill-loners}.
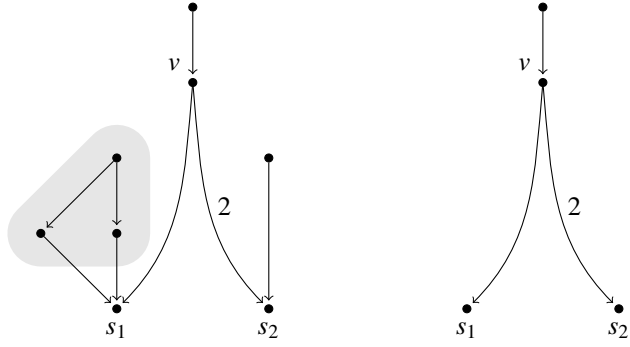
\begin{figure}
  \centering
    \begin{tikzpicture}[shorten >= 0.5mm]
    \tikzstyle{vertex}=[circle,draw,fill=black,minimum size=3pt,inner sep=0pt]

    \node[vertex,label=below:$s_1$] (s1) at (1,0) {};
    \node[vertex,label=below:$s_2$] (s2) at (3,0) {};

    \node[vertex] (a) at (0,1) {};
    \node[vertex] (b) at (1,1) {};
    \node[vertex] (c) at (1,2) {};
    \node[vertex,label=above left:$v$] (d) at (2,3) {};
    \node[vertex] (e) at (2,4) {};
    \node[vertex] (f) at (3,2) {};

    \draw[->] (a)--(s1);
    \draw[->] (b)--(s1);
    \draw[->] (c)--(b);
    \draw[->] (d) to[out=-95,in=45](s1);
    \draw[->] (d) to[out=-85] node[right]{$2$} (s2);
    \draw[->] (c) to(a);
    \draw[->] (e)--(d);
    \draw[->] (f)--(s2);

    \begin{pgfonlayer}{background}
      \draw[fill=black!10,line width=25pt,line join=round, line cap=round, draw=black!10] (a.center)--(b.center)--(c.center)--cycle;

    \end{pgfonlayer}
  \end{tikzpicture}\hspace{2cm}
    \begin{tikzpicture}[shorten >= 0.5mm]
    \tikzstyle{vertex}=[circle,draw,fill=black,minimum size=3pt,inner sep=0pt]

    \node[vertex,label=below:$s_1$] (s1) at (1,0) {};
    \node[vertex,label=below:$s_2$] (s2) at (3,0) {};

    \node[vertex,label=above left:$v$] (d) at (2,3) {};
    \node[vertex] (e) at (2,4) {};

    \draw[->] (d) to[out=-95,in=45](s1);
    \draw[->] (d) to[out=-85] node[right]{$2$} (s2);
    \draw[->] (e)--(d);
  \end{tikzpicture}
  \caption{On the left side, a graph that cannot be reduced by \autoref{halfblind} is shown.  The gray background highlights the non-sink vertices that can only reach the sink~$s_1$.  The right side shows the same graph to which \autoref{kill-loners} has been applied as often as possible. Unlabeled arcs have weight one.}
  \label{fig:kill-loners}
\end{figure}

\begin{reductionrule}\label{kill-loners}
  If, for some sink~\(s\), the set~\(L\) of non-sink vertices that can
  reach only~$s$ is nonempty, then delete all vertices in~$L$.
\end{reductionrule}

\begin{lemma}
  Let $G$~be a graph that is exhaustively reduced with respect to
  \autoref{halfblind} and~$G':=G-L$ be the graph output by
  \autoref{kill-loners} when applied to~$G$ for some sink~$s$.
  Then, any \partSet{} for~$G$ is a \partSet{} of equal weight for~$G'$
  and vice versa.
\end{lemma}

\begin{proof}
  In order to prove the lemma, we first make three structural
  observations about the set~$L$.
  \begin{enumerate}[i)]
  \item \label{obsL1} There is no arc~$(v,w)$ from a vertex~$v\notin L$
    to a vertex~$w\in L$ in~$G$: for the sake of a contradiction, assume that
    such an arc exists. Then, since $v\notin L$ and $v$~is obviously not
    a sink, $v$~can reach a sink~$s'\neq s$. It follows that $v$~can
    reach two sinks: $s$ via~$w$ and~$s'$. This contradicts the
    assumption that \autoref{halfblind} is not applicable.

  \item \label{obsL2} There is no arc~$(v,w)$ from a vertex~$v\in L$ to
    a vertex~$w\notin L$ with $w\neq s$ in~$G$: for the sake of a
    contradiction, assume that such an arc exists. Then, since $w\notin L$
    and~$w\neq s$, it follows that $w$~can reach a sink~$s'$ different
    from~$s$. Then, also $v$~can reach two sinks: $s'$~via~$w$
    and~$s$. This contradicts~$v\in L$.

  \item \label{obsL3} A minimal \partSet{}~$\Sol{}$ for~$G$ does not contain any arc between vertices in~$L\cup\{s\}$: this is because, by \autoref{lem:no_new_sinks}, no minimal \partSet{}~$\Sol{}$ can disconnect any vertex~$v\in L$ from~$s$, since otherwise~$v$ would reach another, that is, new sink in~$G\setminus\Sol{}$.
  \end{enumerate}
  Now, let $\Sol{}$~be a minimal \partSet{} for~$G$. Then, $\Sol{}$~is also a \partSet{} for~$G'$, since $G'\setminus\Sol{}=(G\setminus\Sol{})-L$ and deleting~$L$ from~$G\setminus\Sol{}$ cannot create new sinks.
  
  In the opposite direction, let $S$~be a \partSet{} for~$G'$. Then, $S$~is also a \partSet{} for~$G$, since $G\setminus\Sol{}$~is just $G'\setminus\Sol{}$~with the vertices in~$L$ and their arcs added. These, however, can reach only the sink~$s$ and are only connected to vertices to which~$s$ is connected. Hence, they do not create new sinks or connect distinct components of~$G'\setminus\Sol{}$.\qed
\end{proof}

\noindent We now show how to exhaustively apply both data reduction
rules in linear time. To this end, we apply \autoref{reducealg}: in
lines~\ref{startinit}--\ref{dynprog}, it computes an array~$L$ such
that, for each vertex~$v\in V$, we have $L[v]=\{s\}$ if $v$~reaches
exactly one sink~$s$ and $L[v]=\emptyset$~otherwise. It uses this
information to apply \autoref{halfblind} in
lines~\ref{apprule1}--\ref{delarc} and \autoref{kill-loners} in
lines~\ref{apprule2} and~\ref{apprule2end}.

\begin{algorithm}[t]
  \caption{Apply Reduction Rules~\ref{halfblind} and~\ref{kill-loners} exhaustively.}
  \label{reducealg}

  \KwIn{A directed acyclic graph~$G=(V,A)$ with arc weights~$\wei{}$.}
  \KwOut{The result of exhaustively applying Reduction
    Rules~\ref{halfblind} and~\ref{kill-loners} to~$G$.}
  $\Sin\gets{}$sinks of~$G$\nllabel{startinit}\; $L\gets{}$array with
  $n$~entries\; \lForEach(\nllabel{endinit}\tcp*[f]{\textrm{$L[v]=\{s\}$ for any sink~$s\in V$ will mean that~$v$ only reaches the sink~$s$.}}){$v\in
    \Sin$}{$L[v]\gets\{v\}$} \ForEach(\nllabel{dynprog}){$v\in
    V\setminus \Sin$ in reverse topological
    order}{$L[v]\gets\smashoperator{\bigcap\limits_{u\in N^+(v)}}L[u]$\nllabel{lin:bicap}\tcp*[f]{\textrm{Invariant: the intersection contains at most one sink.}}}
  \ForEach(\nllabel{apprule1}\tcp*[f]{\rm Application of \autoref{halfblind}}
){$v\in V$ with $L[v]=\emptyset$}{
    \ForEach(\nllabel{arrayhere}){$w\in N^+(v)$ with $L[w]=\{s\}$ for some~$s\in \Sin$ and
      $w\notin \Sin$}{ \lIf(\nllabel{isarcthere}){$(v,s)\notin A$}{add
        $(v,s)$ with $\wei{}(v,s):=0$ to~$A$}
      $\wei{}(v,s)\gets\wei{}(v,s)+\wei{}(v,w)$\nllabel{weiinc}\; delete
      arc $(v,w)$\nllabel{delarc}\;} } 
  \ForEach(\nllabel{apprule2}\tcp*[f]{\rm Application of
    \autoref{kill-loners}}){
    $v\in V\setminus \Sin$ such that $L[v]=\{s\}$ for some~$s\in \Sin$}{delete
    vertex~$v$\nllabel{apprule2end}}
  \Return{$(G,\wei{})$}\nllabel{returnred}\;
\end{algorithm}

\begin{lemma}\label{lem:datareductionrules}
  Given a directed acyclic graph~$G$ with weights $\wei{}$, in
  $\bigO(n+m)$~time \autoref{reducealg} produces a directed acyclic
  graph~$G'$ with weights $\wei{}'$ such that $G'$~is exhaustively
  reduced with respect to \autoref{halfblind} and \autoref{kill-loners}.
  In particular, $(G,\wei{},k)$ is a yes-instance if and only
  if~$(G',\wei{}',k)$ is a yes-instance.
\end{lemma}

\begin{proof}
  We first discuss the semantics of \autoref{reducealg}, then its running time.  After \autoref{lin:bicap}, $L[v]=\{s\}$ for some vertex~$v$ if~$v$ can reach exactly one sink~$s$ and $L[v]=\emptyset$ otherwise: this is, by definition, true for all~$L[v]$ with~$v\in \Sin$. For $v\in V\setminus \Sin$ it also holds, since $v$~can reach exactly one sink~$s$ if and only if all of its out-neighbors~$u\in N^+(v)$ can reach~$s$ and no other sinks, that is, if and only if $L[u]=\{s\}$ for all out-neighbors~$u\in N^+(v)$ of~$v$.

  Hence, the loop in lines~\ref{apprule1}--\ref{delarc} applies
  \autoref{halfblind} to all arcs to which \autoref{halfblind} is
  applicable. Moreover, \autoref{halfblind} does not change which sinks
  are reachable from any vertex and, hence, cannot create new arcs to
  which \autoref{halfblind} may be applied. Hence, when reaching
  \autoref{apprule2}, the graph will be exhaustively reduced with respect to
  \autoref{halfblind} and we do not have to update the array~$L$.

  The loop in lines~\ref{apprule2} and~\ref{apprule2end} now applies
  \autoref{kill-loners}, which is allowed, since the graph is exhaustively reduced
  with respect to \autoref{halfblind}. Moreover, an application of
  \autoref{kill-loners} cannot create new vertices to which
  \autoref{kill-loners} may become applicable or arcs to which
  \autoref{halfblind} may become applicable. Hence, \autoref{returnred}
  indeed returns a graph that is exhaustively reduced with respect to
  both data reduction rules.

  It remains to analyze the running time. Obviously, lines \ref{startinit}--\ref{endinit} of \autoref{reducealg} work in $\bigO(n)$~time. To execute \autoref{dynprog} in $\bigO(n+m)$~time, we iterate over the vertices in~$V\setminus \Sin$ in reverse topological order, which can be computed in $\bigO(n+m)$~time \citep[Section~22.4]{CLRS01}. Hence, when computing~$L[v]$ for some vertex in \autoref{dynprog}, we already know the values~$L[u]$ for all~$u\in N^+(v)$. Moreover, $L[v]$~is the intersection of sets with at most one element and, therefore, also contains at most one element. It follows that we can compute~$L[v]$ in $\bigO(|N^+(v)|)$~time for each vertex~$v\in V\setminus \Sin$ and, therefore, in $\bigO(n+m)$~total time for all vertices. The rest of the algorithm only iterates once over all arcs and vertices. Hence, to show that it works in $\bigO(n+m)$~time, it remains to show how to execute lines~\ref{isarcthere} and~\ref{weiinc} in constant time.

  Herein, the main difficulty is that an adjacency list cannot answer queries of the form ``$(v,s)\in A$?'' in constant time. However, since we earlier required to iterate over all out-neighbors of a vertex~$v$ in~$\bigO(|N^+(v)|)$~time, we cannot just use an adjacency matrix instead.  We exploit a different trick, which, for the sake of clarity is not made explicit in the pseudo code: assume that, when considering a vertex~$v\in V$ in \autoref{apprule1}, we have an $n$-element array~$A$ such that~$A[s]$ holds a pointer to the value~$\wei(v,s)$ if~$(v,s)\in A$ and $A[s]=\bot$ otherwise. Then, we could in constant time check in \autoref{isarcthere} whether~$A[s]=\bot$ to find out whether $(v,s)\in A$ and, if this is the case, get a pointer to (and increment) the weight~$\wei(v,s)$ in constant time in \autoref{weiinc}.  However, we cannot afford initializing an \(n\)-entry array~\(A\) for each vertex~\(v\in V\) and we cannot make assumptions on the value of uninitialized entries.  Luckily, we access~$A[s]$ only if there is a vertex~$w\in N^+(v)$ with $L[w]=\{s\}$ for some~$s$. Hence, we can create an \(n\)-entry array~\(A\) once in the beginning of the algorithm and then, between lines~\ref{apprule1} and~\ref{arrayhere}, set up~\(A\) for~\(v\in V\) as follows: for each~$w\in N^+(v)$ with $L[w]=\{s\}$, set $A_v[s]:=\bot$. Then, for each $s\in N^+(v)$ with $s\in\Sin$, let $A_v[s]$ point to $\wei(v,s)$. \qed
\end{proof}
\noindent
\autoref{lem:datareductionrules} shows that we can exhaustively apply
the two Reduction Rules~\ref{halfblind} and~\ref{kill-loners} in linear
time using \autoref{reducealg}.  A natural approach for evaluating the
quality of our preprocessing would be to provide a performance guarantee
in terms of small problem kernels.  Unfortunately, in
\autoref{sec:kkern} we show that, under widely accepted
complexity-theoretic assumptions, there is no problem kernel with size
polynomial in~\(k\) for \DAGP{}.  Nevertheless, the next section shows
that our data reduction technique performs remarkably well in empirical
tests.  Furthermore, we show that the running time of
\autoref{alg:simple-st} is significantly improved when it is applied
after \autoref{reducealg}.  We achieve the largest speedup of
\autoref{alg:simple-st} by interleaving the application of
\autoref{reducealg} and the branching; a technique that generally works
well for search tree algorithms~\citep{NR00}.

\subsection{Experimental evaluation}\label{sec:experiments}
In this section, we aim for giving a proof of concept of our $\bigO(2^k\cdot(n+m))$~time
search tree algorithm presented in \autoref{sec:klintime} by demonstrating to which
extent instances of \DAGP{} are solvable within a time frame of five minutes.
Moreover, using the optimal solutions found by our algorithm,
we evaluate the quality of a heuristic presented by \citet{LBK09},
which can be considered a variant of our search tree algorithm (\autoref{alg:simple-st}):
The difference is that, while our search tree algorithm branches into all possibilities
of putting a vertex into a connected component with some sink, the heuristic just puts
each vertex into the connected component with the sink it would be most expensive
to disconnect from.
This is the strategy described by \citet{LBK09} as yielding the best results and is in more detail
described by \citet{SHE+13}. The pseudocode of the heuristics (\autoref{alg:simple-trees}) can be found in \autoref{sec:tree}, where we prove that the heuristic works optimally on trees.

\paragraph{Implementation details}  We implemented the search tree algorithm
as well as the heuristic in three variants:
\begin{enumerate}
\item without data reduction,
\item with initially applying the data reduction algorithm presented in
  \autoref{reducealg}, and
\item with interleaving the data reduction using \autoref{reducealg}
  with the branching of \autoref{alg:simple-st}.
\end{enumerate}
The source code uses about 1000 lines of C++ and is freely
available.\footnote{\url{http://fpt.akt.tu-berlin.de/dagpart/}} The
experiments were run on a computer with a 3.6\,GHz Intel Xeon processor
and 64\,GiB RAM under Linux~3.2.0, where the source code has been
compiled using the GNU C++ compiler version~4.7.2 and using the
highest optimization level~(-O3).

\paragraph{Data} We tried to apply our algorithm to the data set
described by \citet{LBK09}; unfortunately, its optimum \partSet{}s have
too large weight to be found by our algorithm.\footnote{The exact
  weights of the optimum \partSet{}s remain unknown, since our algorithm
  could not compute them within several hours.} In order to prove the
feasibility of solving large instances with \emph{small}
minimum \partSet{}s, we generated artificial instances. Herein, however,
we stick to the clustering motivation of \DAGP{} and test our algorithm
using simulated citation networks: vertices in a graph represent
articles and if an article~$v$ cites an article~$w$, there is an
arc~$(v,w)$.  Herein, we consider only directed acyclic graphs, which
model that an article only cites older articles. A partitioning of such
a network into connected components of which each contains only one sink
can be interpreted as a clustering into different topics of which we
want to identify the origins.

To simulate citation networks, we employ preferential attachment graphs---a random graph model commonly used to model citations between articles~\cite{Pri76,BA99,JNB03}.  Preferential attachment graphs model the natural growth of citation networks, in which new articles are published over time and with high probability cite the already highly-cited articles. Indeed, \citet{JNB03} empirically verified that in this preferential attachment model, the probability of an article being cited is linear in the number of times the article has been cited in the past.

To create a preferential attachment graph, we first choose two parameters: the number~$c$ of sinks to create and the maximum outdegree~$d$ of each vertex. After creating $c$~sinks, $n$~new vertices are introduced one after another. After introducing a vertex, we add to it $d$~outgoing unit-weight arcs: for each of these \(d\)~outgoing arcs, the target is chosen independently at random among previously introduced vertices such that each vertex is chosen as the target with a probability proportional to its indegree plus one.  We do not add an outgoing arc twice, which might result in a vertex having less than \(d\)~outgoing arcs if the target nodes of two arcs to be added coincide.

We compared our algorithm to the heuristic of \citet{LBK09} on these
graphs, but our algorithm could solve only instances with up to 300~arcs optimally,
since the optimum solution weight grows too quickly in the sizes of the
generated graphs. To show the running time behavior of our algorithm on
larger graphs with small solution sizes, we employ an additional
approach: we generate multiple connected components, each being a
preferential attachment graph with one sink, and randomly add
$k$~additional arcs between these connected components in a way so that
the graph remains acyclic.
Then, obviously, an optimal \partSet{} cannot be larger than~$k$.
We call the set of $k$~randomly added arcs \emph{embedded partitioning set} 
and it can be viewed as noise in data that clusters well.

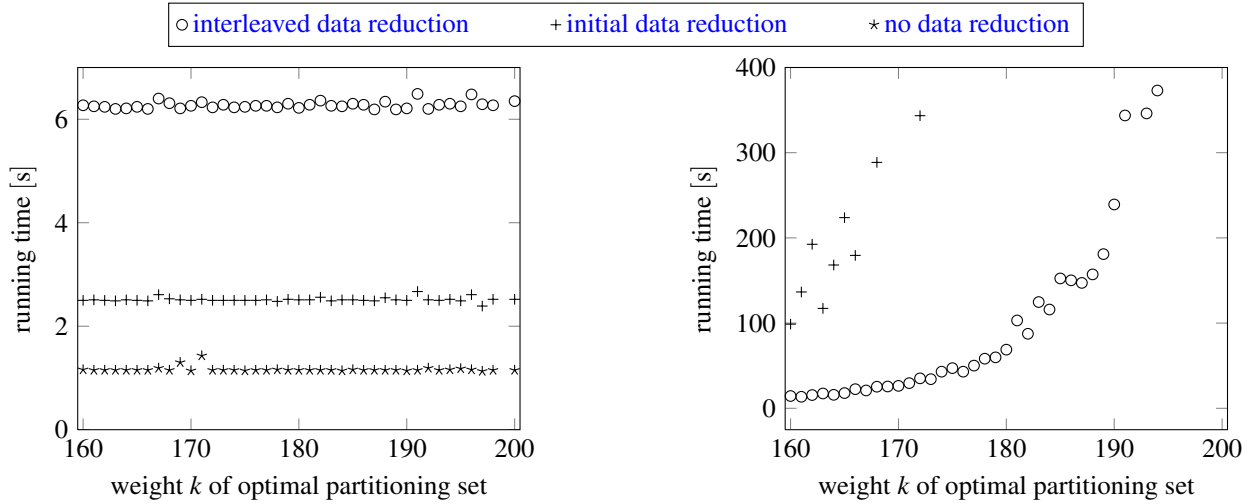
\begin{figure}
  \centering\ref{dagp-legends}
  \begin{tikzpicture}
    \begin{axis}[ymin=0, ymax=7, ylabel=running time,
      xlabel=weight~$k$ of optimal \partSet{}, y unit=s,
      width=0.45\textwidth, legend to name=dagp-legends, legend
      columns=-1, legend style={/tikz/every even column/.append
        style={column sep=1cm}}, xmax=200.5, xmin=159.5]
      
      \addplot[color=black,mark=o, only marks] table
      [x=STiDRk, y=MCiDRs, col sep=comma] {
V,E,DRV,DRE,DRs,MCiDRs,MCDRs,MCs,STiDRs,STDRs,MCiDRk,MCDRk,MCk,STiDRk,STDRk
1000000 547,18764908 1074,547,1074,2.05,6.27,2.5,1.16,14.4,98.77,160,160,160,160,160
1000000 714,18770910 1408,714,1408,2.03,6.25,2.51,1.15,13.52,136.53,161,161,161,161,161
1000000 507,18766714 994,507,994,2.06,6.24,2.5,1.15,15.54,192.47,162,162,162,162,162
1000000 384,18746771 748,384,748,2.04,6.2,2.49,1.15,17.38,117.24,163,163,163,163,163
1000000 491,18759591 962,491,962,2.06,6.21,2.51,1.15,15.81,168.18,164,164,164,164,164
1000000 517,18779564 1014,517,1014,2.03,6.24,2.5,1.15,17.86,223.77,165,165,165,165,165
1000000 451,18760003 882,451,882,2.05,6.2,2.49,1.15,22.45,179.45,166,166,166,166,166
1000000 4697,18764259 9481,4697,9481,2.07,6.4,2.61,1.19,20.93,,167,167,167,167,
1000000 790,18762325 1562,790,1562,2.05,6.31,2.53,1.15,25.32,288.61,168,168,168,168,168
1000000 603,18768885 1187,603,1187,2.06,6.21,2.51,1.3,25.59,,169,169,169,169,
1000000 636,18772622 1252,636,1252,2.03,6.26,2.5,1.14,26.32,442.82,170,170,170,170,170
1000000 1083,18763249 2151,1083,2151,2.05,6.33,2.52,1.43,29.52,430.91,171,171,171,171,171
1000000 625,18755903 1230,625,1230,2.04,6.23,2.5,1.15,35.14,343.45,172,172,172,172,172
1000000 498,18752502 976,498,976,2.05,6.28,2.5,1.15,34.17,541.18,173,173,173,173,173
1000000 758,18760875 1496,758,1496,2.05,6.23,2.5,1.15,42.98,570.39,174,174,174,174,174
1000000 639,18761487 1258,639,1258,2.05,6.24,2.5,1.14,47.11,711.18,175,175,175,175,175
1000000 479,18744904 938,479,938,2.04,6.26,2.5,1.15,42.98,491.98,176,176,176,176,176
1000000 824,18753844 1629,824,1629,2.05,6.26,2.51,1.15,50.1,741.57,177,177,177,177,177
1000000 436,18757283 853,436,853,2.04,6.23,2.48,1.16,58.23,644.88,178,178,178,178,178
1000000 1217,18760997 2419,1217,2419,2.06,6.3,2.52,1.15,59.77,867.64,179,179,179,179,179
1000000 816,18764816 1613,816,1613,2.05,6.22,2.51,1.15,68.82,,180,180,180,180,
1000000 868,18758176 1716,868,1716,2.05,6.28,2.51,1.15,103.09,1402.45,181,181,181,181,181
1000000 2117,18762995 4222,2117,4222,2.05,6.36,2.56,1.15,87.46,,182,182,182,182,
1000000 614,18763594 1208,614,1208,2.05,6.26,2.49,1.15,124.71,,183,183,183,183,
1000000 742,18762792 1464,742,1464,2.05,6.25,2.51,1.14,115.92,2026.76,184,184,184,184,184
1000000 729,18763622 1438,729,1438,2.05,6.3,2.51,1.16,152.4,2601.62,185,185,185,185,185
1000000 596,18772695 1172,596,1172,2.05,6.28,2.5,1.15,150.19,,186,186,186,186,
1000000 575,18765267 1130,575,1130,2.04,6.19,2.49,1.15,147.16,,187,187,187,187,
1000000 1916,18759066 3827,1916,3827,2.06,6.34,2.55,1.15,157.14,1630.15,188,188,188,188,188
1000000 602,18758051 1184,602,1184,2.04,6.19,2.51,1.15,180.95,,189,189,189,189,
1000000 698,18759017 1376,698,1376,2.04,6.21,2.5,1.14,239.23,,190,190,190,190,
1000000 6702,18765728 13631,6702,13631,2.07,6.49,2.67,1.15,343.69,,191,191,191,191,
1000000 794,18764629 1571,794,1571,2.04,6.2,2.51,1.19,420.52,,192,192,192,192,
1000000 806,18760556 1594,806,1594,2.05,6.28,2.5,1.15,346.09,,193,193,193,193,
1000000 670,18759452 1321,670,1321,2.06,6.3,2.52,1.16,372.98,,194,194,194,194,
1000000 555,18760070 1091,555,1091,2.06,6.25,2.49,1.18,450.42,,195,195,195,195,
1000000 4102,18759508 8255,4102,8255,2.07,6.48,2.61,1.16,516.69,,196,196,196,196,
1000000 820,18759429 1620,820,1620,2.05,6.29,2.39,1.13,521.22,,197,197,197,197,
1000000 922,18756012 1825,922,1825,2.06,6.27,2.52,1.15,593.02,,198,198,198,198,
1000000 1847,18753561 3680,1847,3680,2.04,6.35,2.42,1.14,,,199,199,199,,
1000000 1151,18754132 2284,1151,2284,2.05,6.35,2.52,1.15,748.9,,200,200,200,200,
};
      \addlegendentry{interleaved data reduction};

      \addplot[color=black,mark=+, only marks] table
      [x=STiDRk, y=MCDRs, col sep=comma] {V,E,DRV,DRE,DRs,MCiDRs,MCDRs,MCs,STiDRs,STDRs,MCiDRk,MCDRk,MCk,STiDRk,STDRk
1000000 547,18764908 1074,547,1074,2.05,6.27,2.5,1.16,14.4,98.77,160,160,160,160,160
1000000 714,18770910 1408,714,1408,2.03,6.25,2.51,1.15,13.52,136.53,161,161,161,161,161
1000000 507,18766714 994,507,994,2.06,6.24,2.5,1.15,15.54,192.47,162,162,162,162,162
1000000 384,18746771 748,384,748,2.04,6.2,2.49,1.15,17.38,117.24,163,163,163,163,163
1000000 491,18759591 962,491,962,2.06,6.21,2.51,1.15,15.81,168.18,164,164,164,164,164
1000000 517,18779564 1014,517,1014,2.03,6.24,2.5,1.15,17.86,223.77,165,165,165,165,165
1000000 451,18760003 882,451,882,2.05,6.2,2.49,1.15,22.45,179.45,166,166,166,166,166
1000000 4697,18764259 9481,4697,9481,2.07,6.4,2.61,1.19,20.93,,167,167,167,167,
1000000 790,18762325 1562,790,1562,2.05,6.31,2.53,1.15,25.32,288.61,168,168,168,168,168
1000000 603,18768885 1187,603,1187,2.06,6.21,2.51,1.3,25.59,,169,169,169,169,
1000000 636,18772622 1252,636,1252,2.03,6.26,2.5,1.14,26.32,442.82,170,170,170,170,170
1000000 1083,18763249 2151,1083,2151,2.05,6.33,2.52,1.43,29.52,430.91,171,171,171,171,171
1000000 625,18755903 1230,625,1230,2.04,6.23,2.5,1.15,35.14,343.45,172,172,172,172,172
1000000 498,18752502 976,498,976,2.05,6.28,2.5,1.15,34.17,541.18,173,173,173,173,173
1000000 758,18760875 1496,758,1496,2.05,6.23,2.5,1.15,42.98,570.39,174,174,174,174,174
1000000 639,18761487 1258,639,1258,2.05,6.24,2.5,1.14,47.11,711.18,175,175,175,175,175
1000000 479,18744904 938,479,938,2.04,6.26,2.5,1.15,42.98,491.98,176,176,176,176,176
1000000 824,18753844 1629,824,1629,2.05,6.26,2.51,1.15,50.1,741.57,177,177,177,177,177
1000000 436,18757283 853,436,853,2.04,6.23,2.48,1.16,58.23,644.88,178,178,178,178,178
1000000 1217,18760997 2419,1217,2419,2.06,6.3,2.52,1.15,59.77,867.64,179,179,179,179,179
1000000 816,18764816 1613,816,1613,2.05,6.22,2.51,1.15,68.82,,180,180,180,180,
1000000 868,18758176 1716,868,1716,2.05,6.28,2.51,1.15,103.09,1402.45,181,181,181,181,181
1000000 2117,18762995 4222,2117,4222,2.05,6.36,2.56,1.15,87.46,,182,182,182,182,
1000000 614,18763594 1208,614,1208,2.05,6.26,2.49,1.15,124.71,,183,183,183,183,
1000000 742,18762792 1464,742,1464,2.05,6.25,2.51,1.14,115.92,2026.76,184,184,184,184,184
1000000 729,18763622 1438,729,1438,2.05,6.3,2.51,1.16,152.4,2601.62,185,185,185,185,185
1000000 596,18772695 1172,596,1172,2.05,6.28,2.5,1.15,150.19,,186,186,186,186,
1000000 575,18765267 1130,575,1130,2.04,6.19,2.49,1.15,147.16,,187,187,187,187,
1000000 1916,18759066 3827,1916,3827,2.06,6.34,2.55,1.15,157.14,1630.15,188,188,188,188,188
1000000 602,18758051 1184,602,1184,2.04,6.19,2.51,1.15,180.95,,189,189,189,189,
1000000 698,18759017 1376,698,1376,2.04,6.21,2.5,1.14,239.23,,190,190,190,190,
1000000 6702,18765728 13631,6702,13631,2.07,6.49,2.67,1.15,343.69,,191,191,191,191,
1000000 794,18764629 1571,794,1571,2.04,6.2,2.51,1.19,420.52,,192,192,192,192,
1000000 806,18760556 1594,806,1594,2.05,6.28,2.5,1.15,346.09,,193,193,193,193,
1000000 670,18759452 1321,670,1321,2.06,6.3,2.52,1.16,372.98,,194,194,194,194,
1000000 555,18760070 1091,555,1091,2.06,6.25,2.49,1.18,450.42,,195,195,195,195,
1000000 4102,18759508 8255,4102,8255,2.07,6.48,2.61,1.16,516.69,,196,196,196,196,
1000000 820,18759429 1620,820,1620,2.05,6.29,2.39,1.13,521.22,,197,197,197,197,
1000000 922,18756012 1825,922,1825,2.06,6.27,2.52,1.15,593.02,,198,198,198,198,
1000000 1847,18753561 3680,1847,3680,2.04,6.35,2.42,1.14,,,199,199,199,,
1000000 1151,18754132 2284,1151,2284,2.05,6.35,2.52,1.15,748.9,,200,200,200,200,
};
      \addlegendentry{initial data reduction};

      \addplot[color=black,mark=star, only marks]
      table [x=STiDRk, y=MCs, col sep=comma] {V,E,DRV,DRE,DRs,MCiDRs,MCDRs,MCs,STiDRs,STDRs,MCiDRk,MCDRk,MCk,STiDRk,STDRk
1000000 547,18764908 1074,547,1074,2.05,6.27,2.5,1.16,14.4,98.77,160,160,160,160,160
1000000 714,18770910 1408,714,1408,2.03,6.25,2.51,1.15,13.52,136.53,161,161,161,161,161
1000000 507,18766714 994,507,994,2.06,6.24,2.5,1.15,15.54,192.47,162,162,162,162,162
1000000 384,18746771 748,384,748,2.04,6.2,2.49,1.15,17.38,117.24,163,163,163,163,163
1000000 491,18759591 962,491,962,2.06,6.21,2.51,1.15,15.81,168.18,164,164,164,164,164
1000000 517,18779564 1014,517,1014,2.03,6.24,2.5,1.15,17.86,223.77,165,165,165,165,165
1000000 451,18760003 882,451,882,2.05,6.2,2.49,1.15,22.45,179.45,166,166,166,166,166
1000000 4697,18764259 9481,4697,9481,2.07,6.4,2.61,1.19,20.93,,167,167,167,167,
1000000 790,18762325 1562,790,1562,2.05,6.31,2.53,1.15,25.32,288.61,168,168,168,168,168
1000000 603,18768885 1187,603,1187,2.06,6.21,2.51,1.3,25.59,,169,169,169,169,
1000000 636,18772622 1252,636,1252,2.03,6.26,2.5,1.14,26.32,442.82,170,170,170,170,170
1000000 1083,18763249 2151,1083,2151,2.05,6.33,2.52,1.43,29.52,430.91,171,171,171,171,171
1000000 625,18755903 1230,625,1230,2.04,6.23,2.5,1.15,35.14,343.45,172,172,172,172,172
1000000 498,18752502 976,498,976,2.05,6.28,2.5,1.15,34.17,541.18,173,173,173,173,173
1000000 758,18760875 1496,758,1496,2.05,6.23,2.5,1.15,42.98,570.39,174,174,174,174,174
1000000 639,18761487 1258,639,1258,2.05,6.24,2.5,1.14,47.11,711.18,175,175,175,175,175
1000000 479,18744904 938,479,938,2.04,6.26,2.5,1.15,42.98,491.98,176,176,176,176,176
1000000 824,18753844 1629,824,1629,2.05,6.26,2.51,1.15,50.1,741.57,177,177,177,177,177
1000000 436,18757283 853,436,853,2.04,6.23,2.48,1.16,58.23,644.88,178,178,178,178,178
1000000 1217,18760997 2419,1217,2419,2.06,6.3,2.52,1.15,59.77,867.64,179,179,179,179,179
1000000 816,18764816 1613,816,1613,2.05,6.22,2.51,1.15,68.82,,180,180,180,180,
1000000 868,18758176 1716,868,1716,2.05,6.28,2.51,1.15,103.09,1402.45,181,181,181,181,181
1000000 2117,18762995 4222,2117,4222,2.05,6.36,2.56,1.15,87.46,,182,182,182,182,
1000000 614,18763594 1208,614,1208,2.05,6.26,2.49,1.15,124.71,,183,183,183,183,
1000000 742,18762792 1464,742,1464,2.05,6.25,2.51,1.14,115.92,2026.76,184,184,184,184,184
1000000 729,18763622 1438,729,1438,2.05,6.3,2.51,1.16,152.4,2601.62,185,185,185,185,185
1000000 596,18772695 1172,596,1172,2.05,6.28,2.5,1.15,150.19,,186,186,186,186,
1000000 575,18765267 1130,575,1130,2.04,6.19,2.49,1.15,147.16,,187,187,187,187,
1000000 1916,18759066 3827,1916,3827,2.06,6.34,2.55,1.15,157.14,1630.15,188,188,188,188,188
1000000 602,18758051 1184,602,1184,2.04,6.19,2.51,1.15,180.95,,189,189,189,189,
1000000 698,18759017 1376,698,1376,2.04,6.21,2.5,1.14,239.23,,190,190,190,190,
1000000 6702,18765728 13631,6702,13631,2.07,6.49,2.67,1.15,343.69,,191,191,191,191,
1000000 794,18764629 1571,794,1571,2.04,6.2,2.51,1.19,420.52,,192,192,192,192,
1000000 806,18760556 1594,806,1594,2.05,6.28,2.5,1.15,346.09,,193,193,193,193,
1000000 670,18759452 1321,670,1321,2.06,6.3,2.52,1.16,372.98,,194,194,194,194,
1000000 555,18760070 1091,555,1091,2.06,6.25,2.49,1.18,450.42,,195,195,195,195,
1000000 4102,18759508 8255,4102,8255,2.07,6.48,2.61,1.16,516.69,,196,196,196,196,
1000000 820,18759429 1620,820,1620,2.05,6.29,2.39,1.13,521.22,,197,197,197,197,
1000000 922,18756012 1825,922,1825,2.06,6.27,2.52,1.15,593.02,,198,198,198,198,
1000000 1847,18753561 3680,1847,3680,2.04,6.35,2.42,1.14,,,199,199,199,,
1000000 1151,18754132 2284,1151,2284,2.05,6.35,2.52,1.15,748.9,,200,200,200,200,
};
      
      \addlegendentry{no data reduction};
    \end{axis}
  \end{tikzpicture}\hfill{}
  \begin{tikzpicture}
    \begin{axis}[ylabel=running time, xlabel=weight~$k$ of
      optimal \partSet{}, y unit=s, xmax=200.5, xmin=159.5, ymax=400,
      width=0.45\textwidth]
      
      \addplot[color=black,mark=o, only marks] table
      [x=STiDRk, y=STiDRs, col sep=comma] {V,E,DRV,DRE,DRs,MCiDRs,MCDRs,MCs,STiDRs,STDRs,MCiDRk,MCDRk,MCk,STiDRk,STDRk
1000000 547,18764908 1074,547,1074,2.05,6.27,2.5,1.16,14.4,98.77,160,160,160,160,160
1000000 714,18770910 1408,714,1408,2.03,6.25,2.51,1.15,13.52,136.53,161,161,161,161,161
1000000 507,18766714 994,507,994,2.06,6.24,2.5,1.15,15.54,192.47,162,162,162,162,162
1000000 384,18746771 748,384,748,2.04,6.2,2.49,1.15,17.38,117.24,163,163,163,163,163
1000000 491,18759591 962,491,962,2.06,6.21,2.51,1.15,15.81,168.18,164,164,164,164,164
1000000 517,18779564 1014,517,1014,2.03,6.24,2.5,1.15,17.86,223.77,165,165,165,165,165
1000000 451,18760003 882,451,882,2.05,6.2,2.49,1.15,22.45,179.45,166,166,166,166,166
1000000 4697,18764259 9481,4697,9481,2.07,6.4,2.61,1.19,20.93,,167,167,167,167,
1000000 790,18762325 1562,790,1562,2.05,6.31,2.53,1.15,25.32,288.61,168,168,168,168,168
1000000 603,18768885 1187,603,1187,2.06,6.21,2.51,1.3,25.59,,169,169,169,169,
1000000 636,18772622 1252,636,1252,2.03,6.26,2.5,1.14,26.32,442.82,170,170,170,170,170
1000000 1083,18763249 2151,1083,2151,2.05,6.33,2.52,1.43,29.52,430.91,171,171,171,171,171
1000000 625,18755903 1230,625,1230,2.04,6.23,2.5,1.15,35.14,343.45,172,172,172,172,172
1000000 498,18752502 976,498,976,2.05,6.28,2.5,1.15,34.17,541.18,173,173,173,173,173
1000000 758,18760875 1496,758,1496,2.05,6.23,2.5,1.15,42.98,570.39,174,174,174,174,174
1000000 639,18761487 1258,639,1258,2.05,6.24,2.5,1.14,47.11,711.18,175,175,175,175,175
1000000 479,18744904 938,479,938,2.04,6.26,2.5,1.15,42.98,491.98,176,176,176,176,176
1000000 824,18753844 1629,824,1629,2.05,6.26,2.51,1.15,50.1,741.57,177,177,177,177,177
1000000 436,18757283 853,436,853,2.04,6.23,2.48,1.16,58.23,644.88,178,178,178,178,178
1000000 1217,18760997 2419,1217,2419,2.06,6.3,2.52,1.15,59.77,867.64,179,179,179,179,179
1000000 816,18764816 1613,816,1613,2.05,6.22,2.51,1.15,68.82,,180,180,180,180,
1000000 868,18758176 1716,868,1716,2.05,6.28,2.51,1.15,103.09,1402.45,181,181,181,181,181
1000000 2117,18762995 4222,2117,4222,2.05,6.36,2.56,1.15,87.46,,182,182,182,182,
1000000 614,18763594 1208,614,1208,2.05,6.26,2.49,1.15,124.71,,183,183,183,183,
1000000 742,18762792 1464,742,1464,2.05,6.25,2.51,1.14,115.92,2026.76,184,184,184,184,184
1000000 729,18763622 1438,729,1438,2.05,6.3,2.51,1.16,152.4,2601.62,185,185,185,185,185
1000000 596,18772695 1172,596,1172,2.05,6.28,2.5,1.15,150.19,,186,186,186,186,
1000000 575,18765267 1130,575,1130,2.04,6.19,2.49,1.15,147.16,,187,187,187,187,
1000000 1916,18759066 3827,1916,3827,2.06,6.34,2.55,1.15,157.14,1630.15,188,188,188,188,188
1000000 602,18758051 1184,602,1184,2.04,6.19,2.51,1.15,180.95,,189,189,189,189,
1000000 698,18759017 1376,698,1376,2.04,6.21,2.5,1.14,239.23,,190,190,190,190,
1000000 6702,18765728 13631,6702,13631,2.07,6.49,2.67,1.15,343.69,,191,191,191,191,
1000000 794,18764629 1571,794,1571,2.04,6.2,2.51,1.19,420.52,,192,192,192,192,
1000000 806,18760556 1594,806,1594,2.05,6.28,2.5,1.15,346.09,,193,193,193,193,
1000000 670,18759452 1321,670,1321,2.06,6.3,2.52,1.16,372.98,,194,194,194,194,
1000000 555,18760070 1091,555,1091,2.06,6.25,2.49,1.18,450.42,,195,195,195,195,
1000000 4102,18759508 8255,4102,8255,2.07,6.48,2.61,1.16,516.69,,196,196,196,196,
1000000 820,18759429 1620,820,1620,2.05,6.29,2.39,1.13,521.22,,197,197,197,197,
1000000 922,18756012 1825,922,1825,2.06,6.27,2.52,1.15,593.02,,198,198,198,198,
1000000 1847,18753561 3680,1847,3680,2.04,6.35,2.42,1.14,,,199,199,199,,
1000000 1151,18754132 2284,1151,2284,2.05,6.35,2.52,1.15,748.9,,200,200,200,200,
};

      \addplot[color=black,mark=+, only marks] table
      [x=STDRk, y=STDRs, col sep=comma] {V,E,DRV,DRE,DRs,MCiDRs,MCDRs,MCs,STiDRs,STDRs,MCiDRk,MCDRk,MCk,STiDRk,STDRk
1000000 547,18764908 1074,547,1074,2.05,6.27,2.5,1.16,14.4,98.77,160,160,160,160,160
1000000 714,18770910 1408,714,1408,2.03,6.25,2.51,1.15,13.52,136.53,161,161,161,161,161
1000000 507,18766714 994,507,994,2.06,6.24,2.5,1.15,15.54,192.47,162,162,162,162,162
1000000 384,18746771 748,384,748,2.04,6.2,2.49,1.15,17.38,117.24,163,163,163,163,163
1000000 491,18759591 962,491,962,2.06,6.21,2.51,1.15,15.81,168.18,164,164,164,164,164
1000000 517,18779564 1014,517,1014,2.03,6.24,2.5,1.15,17.86,223.77,165,165,165,165,165
1000000 451,18760003 882,451,882,2.05,6.2,2.49,1.15,22.45,179.45,166,166,166,166,166
1000000 4697,18764259 9481,4697,9481,2.07,6.4,2.61,1.19,20.93,,167,167,167,167,
1000000 790,18762325 1562,790,1562,2.05,6.31,2.53,1.15,25.32,288.61,168,168,168,168,168
1000000 603,18768885 1187,603,1187,2.06,6.21,2.51,1.3,25.59,,169,169,169,169,
1000000 636,18772622 1252,636,1252,2.03,6.26,2.5,1.14,26.32,442.82,170,170,170,170,170
1000000 1083,18763249 2151,1083,2151,2.05,6.33,2.52,1.43,29.52,430.91,171,171,171,171,171
1000000 625,18755903 1230,625,1230,2.04,6.23,2.5,1.15,35.14,343.45,172,172,172,172,172
1000000 498,18752502 976,498,976,2.05,6.28,2.5,1.15,34.17,541.18,173,173,173,173,173
1000000 758,18760875 1496,758,1496,2.05,6.23,2.5,1.15,42.98,570.39,174,174,174,174,174
1000000 639,18761487 1258,639,1258,2.05,6.24,2.5,1.14,47.11,711.18,175,175,175,175,175
1000000 479,18744904 938,479,938,2.04,6.26,2.5,1.15,42.98,491.98,176,176,176,176,176
1000000 824,18753844 1629,824,1629,2.05,6.26,2.51,1.15,50.1,741.57,177,177,177,177,177
1000000 436,18757283 853,436,853,2.04,6.23,2.48,1.16,58.23,644.88,178,178,178,178,178
1000000 1217,18760997 2419,1217,2419,2.06,6.3,2.52,1.15,59.77,867.64,179,179,179,179,179
1000000 816,18764816 1613,816,1613,2.05,6.22,2.51,1.15,68.82,,180,180,180,180,
1000000 868,18758176 1716,868,1716,2.05,6.28,2.51,1.15,103.09,1402.45,181,181,181,181,181
1000000 2117,18762995 4222,2117,4222,2.05,6.36,2.56,1.15,87.46,,182,182,182,182,
1000000 614,18763594 1208,614,1208,2.05,6.26,2.49,1.15,124.71,,183,183,183,183,
1000000 742,18762792 1464,742,1464,2.05,6.25,2.51,1.14,115.92,2026.76,184,184,184,184,184
1000000 729,18763622 1438,729,1438,2.05,6.3,2.51,1.16,152.4,2601.62,185,185,185,185,185
1000000 596,18772695 1172,596,1172,2.05,6.28,2.5,1.15,150.19,,186,186,186,186,
1000000 575,18765267 1130,575,1130,2.04,6.19,2.49,1.15,147.16,,187,187,187,187,
1000000 1916,18759066 3827,1916,3827,2.06,6.34,2.55,1.15,157.14,1630.15,188,188,188,188,188
1000000 602,18758051 1184,602,1184,2.04,6.19,2.51,1.15,180.95,,189,189,189,189,
1000000 698,18759017 1376,698,1376,2.04,6.21,2.5,1.14,239.23,,190,190,190,190,
1000000 6702,18765728 13631,6702,13631,2.07,6.49,2.67,1.15,343.69,,191,191,191,191,
1000000 794,18764629 1571,794,1571,2.04,6.2,2.51,1.19,420.52,,192,192,192,192,
1000000 806,18760556 1594,806,1594,2.05,6.28,2.5,1.15,346.09,,193,193,193,193,
1000000 670,18759452 1321,670,1321,2.06,6.3,2.52,1.16,372.98,,194,194,194,194,
1000000 555,18760070 1091,555,1091,2.06,6.25,2.49,1.18,450.42,,195,195,195,195,
1000000 4102,18759508 8255,4102,8255,2.07,6.48,2.61,1.16,516.69,,196,196,196,196,
1000000 820,18759429 1620,820,1620,2.05,6.29,2.39,1.13,521.22,,197,197,197,197,
1000000 922,18756012 1825,922,1825,2.06,6.27,2.52,1.15,593.02,,198,198,198,198,
1000000 1847,18753561 3680,1847,3680,2.04,6.35,2.42,1.14,,,199,199,199,,
1000000 1151,18754132 2284,1151,2284,2.05,6.35,2.52,1.15,748.9,,200,200,200,200,
};
    \end{axis}
  \end{tikzpicture}
  \caption{Comparisons of the running time of \citet{LBK09}'s heuristic
    (left) with the running time of our search tree algorithm (right).
    All graphs have $10^6$~vertices and roughly $18\cdot 10^6$~arcs and
    were generated by adding $k$~random arcs between ten connected
    components, each being a preferential attachment graph on
    $10^5$~vertices with outdegree twenty and a single sink. The
    heuristic solved all 40~instances optimally.  On the right hand
    side, no instance was solved in less than an hour without data
    reduction.}
  \label{fig:ktime}
\end{figure}

\paragraph{Experimental results}

\autoref{fig:ktime} compares the running time of the heuristic of
\citet{LBK09} to the running time of our \autoref{alg:simple-st} with
increasing optimal \partSet{} size~$k$. On the left side, it can be seen
that using the data reduction from \autoref{reducealg} slows down the
heuristic. This is not surprising, since the heuristic itself is
implemented to run in linear time and, hence, instead of first shrinking
the input instance by \autoref{reducealg} in linear time, one might
right away solve the instance heuristically. On the right side, one can
observe that, as expected, the running time of \autoref{alg:simple-st}
increases exponentially in~$k$. We only show the running time of the
implementations with data reduction: without data reduction, we could
not solve any instance in less than an hour. We can solve instances
with~$k\leq 190$ optimally within five minutes. This allowed us to
verify that the heuristic solved all 40~generated instances optimally,
regardless of the type of data reduction applied.

\begin{figure}
  \centering\ref{dagp-legends}
  \begin{tikzpicture}
    \begin{axis}[ylabel=running time, xmax=80000000, xlabel=input arcs,
      y unit=s, change x base, axis base prefix={axis x base -6 prefix
        {}}, x unit=10^6, width=0.45\textwidth]
      
      \addplot[color=black,mark=o,only marks, each nth point=2] table [x=E,
      y=MCiDRs, col sep=comma] {V,E,DRV,DRE,DRs,MCiDRs,MCDRs,MCs,STiDRs,STDRs,MCiDRk,MCDRk,MCk,STiDRk,STDRk
1000000,18767086,651,1283,2.08,6.49,2.52,1.12,225.68,,190,190,190,190,
1050000,19709159,660,1300,2.17,6.89,2.65,1.16,276.93,,190,190,190,190,
1100000,20655825,469,919,2.27,7.21,2.76,1.22,232.47,,190,190,190,190,
1150000,21604341,2901,5801,2.4,7.79,2.99,1.28,235.17,,190,190,190,190,
1200000,22564319,820,1621,2.49,7.92,3.04,1.34,276.79,,190,190,190,190,
1250000,23531472,1394,2769,2.6,8.49,3.21,1.42,241.42,,190,190,190,190,
1300000,24480730,501,982,2.72,8.81,3.3,1.49,277.68,,190,190,190,190,
1350000,25440966,467,914,2.82,9.04,3.26,1.53,233.7,,190,190,190,190,
1400000,26383694,658,1296,2.93,9.45,3.56,1.59,237.37,,190,190,190,190,
1450000,27342944,585,1151,3.03,9.8,3.7,1.67,232.26,,190,190,190,190,
1500000,28291769,445,870,3.13,10.16,3.8,1.71,237,,190,190,190,190,
1550000,29269274,2200,4399,3.28,10.78,4.03,1.88,222.51,,190,190,190,190,
1600000,30211184,1070,2120,3.34,11.01,4.09,1.84,220.52,,190,190,190,190,
1650000,31153242,1061,2102,3.47,11.38,4.23,1.96,259.86,,190,190,190,190,
1700000,32118855,1043,2067,3.59,11.85,4.37,1.97,197.58,,190,190,190,190,
1750000,33087214,1549,3083,3.7,12.16,4.53,2.1,223.84,,190,190,190,190,
1800000,34052317,1038,2056,3.79,12.57,4.64,2.1,304.66,,190,190,190,190,
1850000,34969772,711,1402,3.9,12.93,4.75,2.17,247.05,,190,190,190,190,
1900000,35955033,664,1309,4.01,13.38,4.87,2.23,227.51,,190,190,190,190,
1950000,36919935,593,1166,4.12,13.69,5.02,2.34,240.66,,190,190,190,190,
2000000,37879080,846,1674,4.25,14.25,5.17,2.42,265.78,,190,190,190,190,
2050000,38826564,562,1104,4.33,14.41,5.29,2.44,238.87,,190,190,190,190,
2100000,39796860,1075,2130,4.45,14.89,5.45,2.51,200.9,,190,190,190,190,
2150000,40753264,698,1376,4.57,15.47,5.58,2.59,225.85,,190,190,190,190,
2200000,41687628,1652,3284,4.68,15.78,5.74,2.64,191.86,,190,190,190,190,
2250000,42676754,949,1878,4.78,16.06,5.84,2.7,211.87,,190,190,190,190,
2300000,43640247,1108,2199,4.9,16.6,5.99,2.83,219.43,,190,190,190,190,
2350000,44569595,7875,15886,5.05,17.64,6.35,2.85,231.67,,190,190,190,190,
2400000,45529873,509,999,5.11,17.22,6.22,2.92,199.19,,190,190,190,190,
2450000,46502969,1385,2752,5.24,17.96,6.4,3.25,224.79,,190,190,190,190,
2500000,47473942,739,1458,5.34,18.14,6.51,3.41,204.44,,190,190,190,190,
2550000,48407592,532,1044,5.47,18.53,6.64,3.13,236.06,,190,190,190,190,
2600000,49356528,869,1718,5.57,19.45,6.78,3.26,225.7,,190,190,190,190,
2650000,50353329,734,1448,5.69,19.34,6.92,3.37,207.39,,190,190,190,190,
2700000,51309157,567,1114,5.76,19.91,7.05,3.36,231.68,,190,190,190,190,
2750000,52280543,1271,2524,5.91,20.35,7.24,3.42,218.43,,190,190,190,190,
2800000,53230505,511,1002,6.03,20.42,7.32,3.58,209.18,,190,190,190,190,
2850000,54187746,563,1106,6.16,21.03,7.49,3.57,190.19,,190,190,190,190,
2900000,55161765,1032,2045,6.24,21.38,7.61,3.62,199.89,,190,190,190,190,
2950000,56118167,1224,2429,6.35,21.78,7.78,3.74,225.98,,190,190,190,190,
3000000,57082768,980,1942,6.46,22.25,7.9,3.74,241.76,,190,190,190,190,
3050000,58006703,746,1472,6.59,22.8,8.02,3.82,237.72,,190,190,190,190,
3100000,58997271,1432,2845,6.68,23.26,8.2,3.98,196.85,,190,190,190,190,
3150000,60003314,496,974,6.8,23.44,8.3,3.96,229.18,,190,190,190,190,
3200000,60919212,1745,3474,6.92,24.05,8.5,4.13,197.47,,190,190,190,190,
3250000,61874877,474,928,7.04,24.43,8.59,4.11,191.67,,190,190,190,190,
3300000,62849494,1218,2417,7.13,24.83,8.74,4.63,211.13,,190,190,190,190,
3350000,63798424,1131,2242,7.23,25.41,8.85,4.25,217.74,,190,190,190,190,
3400000,64745452,699,1378,7.36,25.69,9,4.42,235.04,,190,190,190,190,
3450000,65700700,434,848,7.47,26.04,9.1,4.4,205.8,,190,190,190,190,
3500000,66733434,669,1319,7.59,26.4,9.31,4.71,200.11,,190,190,190,190,
3550000,67646422,698,1376,7.68,26.82,9.38,4.54,209.94,,190,190,190,190,
3600000,68604668,647,1274,7.79,27.53,9.53,4.6,227.55,,190,190,190,190,
3650000,69587650,484,948,7.92,27.66,9.69,4.81,238.75,,190,190,190,190,
3700000,70541366,883,1746,8.01,28.28,9.84,4.77,246.52,,190,190,190,190,
3750000,71517645,450,880,8.12,28.81,9.94,4.83,205.23,,190,190,190,190,
3800000,72483078,588,1156,8.28,29.11,10.1,4.93,229.07,,190,190,190,190,
3850000,73419536,571,1122,8.32,29.37,10.24,5.35,247.33,,190,190,190,190,
3900000,74381081,918,1816,8.47,30.02,10.37,5.15,215.47,,190,190,190,190,
3950000,75353990,715,1410,8.57,30.13,10.5,5.18,205.36,,190,190,190,190,
4000000,76311956,769,1518,8.7,30.67,10.66,5.21,231.4,,190,190,190,190,
4050000,77313266,668,1316,8.8,30.95,10.77,5.39,228.89,,190,190,190,190,
4100000,78249789,1168,2317,8.94,31.71,10.97,5.35,224.48,,190,190,190,190,
4150000,79221466,549,1078,9.05,32.28,11.08,5.53,265.09,,190,190,190,190,
4200000,80154009,512,1004,7.38,35.59,10.97,5.4,226.14,,190,190,190,190,
};

      \addplot[color=black,mark=star, only marks, each nth point=2]
      table [x=E, y=MCs, col sep=comma] {V,E,DRV,DRE,DRs,MCiDRs,MCDRs,MCs,STiDRs,STDRs,MCiDRk,MCDRk,MCk,STiDRk,STDRk
1000000,18767086,651,1283,2.08,6.49,2.52,1.12,225.68,,190,190,190,190,
1050000,19709159,660,1300,2.17,6.89,2.65,1.16,276.93,,190,190,190,190,
1100000,20655825,469,919,2.27,7.21,2.76,1.22,232.47,,190,190,190,190,
1150000,21604341,2901,5801,2.4,7.79,2.99,1.28,235.17,,190,190,190,190,
1200000,22564319,820,1621,2.49,7.92,3.04,1.34,276.79,,190,190,190,190,
1250000,23531472,1394,2769,2.6,8.49,3.21,1.42,241.42,,190,190,190,190,
1300000,24480730,501,982,2.72,8.81,3.3,1.49,277.68,,190,190,190,190,
1350000,25440966,467,914,2.82,9.04,3.26,1.53,233.7,,190,190,190,190,
1400000,26383694,658,1296,2.93,9.45,3.56,1.59,237.37,,190,190,190,190,
1450000,27342944,585,1151,3.03,9.8,3.7,1.67,232.26,,190,190,190,190,
1500000,28291769,445,870,3.13,10.16,3.8,1.71,237,,190,190,190,190,
1550000,29269274,2200,4399,3.28,10.78,4.03,1.88,222.51,,190,190,190,190,
1600000,30211184,1070,2120,3.34,11.01,4.09,1.84,220.52,,190,190,190,190,
1650000,31153242,1061,2102,3.47,11.38,4.23,1.96,259.86,,190,190,190,190,
1700000,32118855,1043,2067,3.59,11.85,4.37,1.97,197.58,,190,190,190,190,
1750000,33087214,1549,3083,3.7,12.16,4.53,2.1,223.84,,190,190,190,190,
1800000,34052317,1038,2056,3.79,12.57,4.64,2.1,304.66,,190,190,190,190,
1850000,34969772,711,1402,3.9,12.93,4.75,2.17,247.05,,190,190,190,190,
1900000,35955033,664,1309,4.01,13.38,4.87,2.23,227.51,,190,190,190,190,
1950000,36919935,593,1166,4.12,13.69,5.02,2.34,240.66,,190,190,190,190,
2000000,37879080,846,1674,4.25,14.25,5.17,2.42,265.78,,190,190,190,190,
2050000,38826564,562,1104,4.33,14.41,5.29,2.44,238.87,,190,190,190,190,
2100000,39796860,1075,2130,4.45,14.89,5.45,2.51,200.9,,190,190,190,190,
2150000,40753264,698,1376,4.57,15.47,5.58,2.59,225.85,,190,190,190,190,
2200000,41687628,1652,3284,4.68,15.78,5.74,2.64,191.86,,190,190,190,190,
2250000,42676754,949,1878,4.78,16.06,5.84,2.7,211.87,,190,190,190,190,
2300000,43640247,1108,2199,4.9,16.6,5.99,2.83,219.43,,190,190,190,190,
2350000,44569595,7875,15886,5.05,17.64,6.35,2.85,231.67,,190,190,190,190,
2400000,45529873,509,999,5.11,17.22,6.22,2.92,199.19,,190,190,190,190,
2450000,46502969,1385,2752,5.24,17.96,6.4,3.25,224.79,,190,190,190,190,
2500000,47473942,739,1458,5.34,18.14,6.51,3.41,204.44,,190,190,190,190,
2550000,48407592,532,1044,5.47,18.53,6.64,3.13,236.06,,190,190,190,190,
2600000,49356528,869,1718,5.57,19.45,6.78,3.26,225.7,,190,190,190,190,
2650000,50353329,734,1448,5.69,19.34,6.92,3.37,207.39,,190,190,190,190,
2700000,51309157,567,1114,5.76,19.91,7.05,3.36,231.68,,190,190,190,190,
2750000,52280543,1271,2524,5.91,20.35,7.24,3.42,218.43,,190,190,190,190,
2800000,53230505,511,1002,6.03,20.42,7.32,3.58,209.18,,190,190,190,190,
2850000,54187746,563,1106,6.16,21.03,7.49,3.57,190.19,,190,190,190,190,
2900000,55161765,1032,2045,6.24,21.38,7.61,3.62,199.89,,190,190,190,190,
2950000,56118167,1224,2429,6.35,21.78,7.78,3.74,225.98,,190,190,190,190,
3000000,57082768,980,1942,6.46,22.25,7.9,3.74,241.76,,190,190,190,190,
3050000,58006703,746,1472,6.59,22.8,8.02,3.82,237.72,,190,190,190,190,
3100000,58997271,1432,2845,6.68,23.26,8.2,3.98,196.85,,190,190,190,190,
3150000,60003314,496,974,6.8,23.44,8.3,3.96,229.18,,190,190,190,190,
3200000,60919212,1745,3474,6.92,24.05,8.5,4.13,197.47,,190,190,190,190,
3250000,61874877,474,928,7.04,24.43,8.59,4.11,191.67,,190,190,190,190,
3300000,62849494,1218,2417,7.13,24.83,8.74,4.63,211.13,,190,190,190,190,
3350000,63798424,1131,2242,7.23,25.41,8.85,4.25,217.74,,190,190,190,190,
3400000,64745452,699,1378,7.36,25.69,9,4.42,235.04,,190,190,190,190,
3450000,65700700,434,848,7.47,26.04,9.1,4.4,205.8,,190,190,190,190,
3500000,66733434,669,1319,7.59,26.4,9.31,4.71,200.11,,190,190,190,190,
3550000,67646422,698,1376,7.68,26.82,9.38,4.54,209.94,,190,190,190,190,
3600000,68604668,647,1274,7.79,27.53,9.53,4.6,227.55,,190,190,190,190,
3650000,69587650,484,948,7.92,27.66,9.69,4.81,238.75,,190,190,190,190,
3700000,70541366,883,1746,8.01,28.28,9.84,4.77,246.52,,190,190,190,190,
3750000,71517645,450,880,8.12,28.81,9.94,4.83,205.23,,190,190,190,190,
3800000,72483078,588,1156,8.28,29.11,10.1,4.93,229.07,,190,190,190,190,
3850000,73419536,571,1122,8.32,29.37,10.24,5.35,247.33,,190,190,190,190,
3900000,74381081,918,1816,8.47,30.02,10.37,5.15,215.47,,190,190,190,190,
3950000,75353990,715,1410,8.57,30.13,10.5,5.18,205.36,,190,190,190,190,
4000000,76311956,769,1518,8.7,30.67,10.66,5.21,231.4,,190,190,190,190,
4050000,77313266,668,1316,8.8,30.95,10.77,5.39,228.89,,190,190,190,190,
4100000,78249789,1168,2317,8.94,31.71,10.97,5.35,224.48,,190,190,190,190,
4150000,79221466,549,1078,9.05,32.28,11.08,5.53,265.09,,190,190,190,190,
4200000,80154009,512,1004,7.38,35.59,10.97,5.4,226.14,,190,190,190,190,
};

      \addplot[color=black,mark=+, only marks, each nth point=2] table
      [x=E, y=MCDRs, col sep=comma] {V,E,DRV,DRE,DRs,MCiDRs,MCDRs,MCs,STiDRs,STDRs,MCiDRk,MCDRk,MCk,STiDRk,STDRk
1000000,18767086,651,1283,2.08,6.49,2.52,1.12,225.68,,190,190,190,190,
1050000,19709159,660,1300,2.17,6.89,2.65,1.16,276.93,,190,190,190,190,
1100000,20655825,469,919,2.27,7.21,2.76,1.22,232.47,,190,190,190,190,
1150000,21604341,2901,5801,2.4,7.79,2.99,1.28,235.17,,190,190,190,190,
1200000,22564319,820,1621,2.49,7.92,3.04,1.34,276.79,,190,190,190,190,
1250000,23531472,1394,2769,2.6,8.49,3.21,1.42,241.42,,190,190,190,190,
1300000,24480730,501,982,2.72,8.81,3.3,1.49,277.68,,190,190,190,190,
1350000,25440966,467,914,2.82,9.04,3.26,1.53,233.7,,190,190,190,190,
1400000,26383694,658,1296,2.93,9.45,3.56,1.59,237.37,,190,190,190,190,
1450000,27342944,585,1151,3.03,9.8,3.7,1.67,232.26,,190,190,190,190,
1500000,28291769,445,870,3.13,10.16,3.8,1.71,237,,190,190,190,190,
1550000,29269274,2200,4399,3.28,10.78,4.03,1.88,222.51,,190,190,190,190,
1600000,30211184,1070,2120,3.34,11.01,4.09,1.84,220.52,,190,190,190,190,
1650000,31153242,1061,2102,3.47,11.38,4.23,1.96,259.86,,190,190,190,190,
1700000,32118855,1043,2067,3.59,11.85,4.37,1.97,197.58,,190,190,190,190,
1750000,33087214,1549,3083,3.7,12.16,4.53,2.1,223.84,,190,190,190,190,
1800000,34052317,1038,2056,3.79,12.57,4.64,2.1,304.66,,190,190,190,190,
1850000,34969772,711,1402,3.9,12.93,4.75,2.17,247.05,,190,190,190,190,
1900000,35955033,664,1309,4.01,13.38,4.87,2.23,227.51,,190,190,190,190,
1950000,36919935,593,1166,4.12,13.69,5.02,2.34,240.66,,190,190,190,190,
2000000,37879080,846,1674,4.25,14.25,5.17,2.42,265.78,,190,190,190,190,
2050000,38826564,562,1104,4.33,14.41,5.29,2.44,238.87,,190,190,190,190,
2100000,39796860,1075,2130,4.45,14.89,5.45,2.51,200.9,,190,190,190,190,
2150000,40753264,698,1376,4.57,15.47,5.58,2.59,225.85,,190,190,190,190,
2200000,41687628,1652,3284,4.68,15.78,5.74,2.64,191.86,,190,190,190,190,
2250000,42676754,949,1878,4.78,16.06,5.84,2.7,211.87,,190,190,190,190,
2300000,43640247,1108,2199,4.9,16.6,5.99,2.83,219.43,,190,190,190,190,
2350000,44569595,7875,15886,5.05,17.64,6.35,2.85,231.67,,190,190,190,190,
2400000,45529873,509,999,5.11,17.22,6.22,2.92,199.19,,190,190,190,190,
2450000,46502969,1385,2752,5.24,17.96,6.4,3.25,224.79,,190,190,190,190,
2500000,47473942,739,1458,5.34,18.14,6.51,3.41,204.44,,190,190,190,190,
2550000,48407592,532,1044,5.47,18.53,6.64,3.13,236.06,,190,190,190,190,
2600000,49356528,869,1718,5.57,19.45,6.78,3.26,225.7,,190,190,190,190,
2650000,50353329,734,1448,5.69,19.34,6.92,3.37,207.39,,190,190,190,190,
2700000,51309157,567,1114,5.76,19.91,7.05,3.36,231.68,,190,190,190,190,
2750000,52280543,1271,2524,5.91,20.35,7.24,3.42,218.43,,190,190,190,190,
2800000,53230505,511,1002,6.03,20.42,7.32,3.58,209.18,,190,190,190,190,
2850000,54187746,563,1106,6.16,21.03,7.49,3.57,190.19,,190,190,190,190,
2900000,55161765,1032,2045,6.24,21.38,7.61,3.62,199.89,,190,190,190,190,
2950000,56118167,1224,2429,6.35,21.78,7.78,3.74,225.98,,190,190,190,190,
3000000,57082768,980,1942,6.46,22.25,7.9,3.74,241.76,,190,190,190,190,
3050000,58006703,746,1472,6.59,22.8,8.02,3.82,237.72,,190,190,190,190,
3100000,58997271,1432,2845,6.68,23.26,8.2,3.98,196.85,,190,190,190,190,
3150000,60003314,496,974,6.8,23.44,8.3,3.96,229.18,,190,190,190,190,
3200000,60919212,1745,3474,6.92,24.05,8.5,4.13,197.47,,190,190,190,190,
3250000,61874877,474,928,7.04,24.43,8.59,4.11,191.67,,190,190,190,190,
3300000,62849494,1218,2417,7.13,24.83,8.74,4.63,211.13,,190,190,190,190,
3350000,63798424,1131,2242,7.23,25.41,8.85,4.25,217.74,,190,190,190,190,
3400000,64745452,699,1378,7.36,25.69,9,4.42,235.04,,190,190,190,190,
3450000,65700700,434,848,7.47,26.04,9.1,4.4,205.8,,190,190,190,190,
3500000,66733434,669,1319,7.59,26.4,9.31,4.71,200.11,,190,190,190,190,
3550000,67646422,698,1376,7.68,26.82,9.38,4.54,209.94,,190,190,190,190,
3600000,68604668,647,1274,7.79,27.53,9.53,4.6,227.55,,190,190,190,190,
3650000,69587650,484,948,7.92,27.66,9.69,4.81,238.75,,190,190,190,190,
3700000,70541366,883,1746,8.01,28.28,9.84,4.77,246.52,,190,190,190,190,
3750000,71517645,450,880,8.12,28.81,9.94,4.83,205.23,,190,190,190,190,
3800000,72483078,588,1156,8.28,29.11,10.1,4.93,229.07,,190,190,190,190,
3850000,73419536,571,1122,8.32,29.37,10.24,5.35,247.33,,190,190,190,190,
3900000,74381081,918,1816,8.47,30.02,10.37,5.15,215.47,,190,190,190,190,
3950000,75353990,715,1410,8.57,30.13,10.5,5.18,205.36,,190,190,190,190,
4000000,76311956,769,1518,8.7,30.67,10.66,5.21,231.4,,190,190,190,190,
4050000,77313266,668,1316,8.8,30.95,10.77,5.39,228.89,,190,190,190,190,
4100000,78249789,1168,2317,8.94,31.71,10.97,5.35,224.48,,190,190,190,190,
4150000,79221466,549,1078,9.05,32.28,11.08,5.53,265.09,,190,190,190,190,
4200000,80154009,512,1004,7.38,35.59,10.97,5.4,226.14,,190,190,190,190,
};
    \end{axis}
  \end{tikzpicture}\hfill{}
  \begin{tikzpicture}
    \begin{axis}[ymin=0,ylabel=running time, xmax=80000000,
      xlabel=input arcs,change x base, axis base prefix={axis x base -6
        prefix {}}, x unit=10^6, y unit=s, width=0.45\textwidth,ymax=400]
      
      \addplot[color=black,mark=o,only marks] table
      [x=E, y=STiDRs, col sep=comma] {V,E,DRV,DRE,DRs,MCiDRs,MCDRs,MCs,STiDRs,STDRs,MCiDRk,MCDRk,MCk,STiDRk,STDRk
1000000,18767086,651,1283,2.08,6.49,2.52,1.12,225.68,,190,190,190,190,
1050000,19709159,660,1300,2.17,6.89,2.65,1.16,276.93,,190,190,190,190,
1100000,20655825,469,919,2.27,7.21,2.76,1.22,232.47,,190,190,190,190,
1150000,21604341,2901,5801,2.4,7.79,2.99,1.28,235.17,,190,190,190,190,
1200000,22564319,820,1621,2.49,7.92,3.04,1.34,276.79,,190,190,190,190,
1250000,23531472,1394,2769,2.6,8.49,3.21,1.42,241.42,,190,190,190,190,
1300000,24480730,501,982,2.72,8.81,3.3,1.49,277.68,,190,190,190,190,
1350000,25440966,467,914,2.82,9.04,3.26,1.53,233.7,,190,190,190,190,
1400000,26383694,658,1296,2.93,9.45,3.56,1.59,237.37,,190,190,190,190,
1450000,27342944,585,1151,3.03,9.8,3.7,1.67,232.26,,190,190,190,190,
1500000,28291769,445,870,3.13,10.16,3.8,1.71,237,,190,190,190,190,
1550000,29269274,2200,4399,3.28,10.78,4.03,1.88,222.51,,190,190,190,190,
1600000,30211184,1070,2120,3.34,11.01,4.09,1.84,220.52,,190,190,190,190,
1650000,31153242,1061,2102,3.47,11.38,4.23,1.96,259.86,,190,190,190,190,
1700000,32118855,1043,2067,3.59,11.85,4.37,1.97,197.58,,190,190,190,190,
1750000,33087214,1549,3083,3.7,12.16,4.53,2.1,223.84,,190,190,190,190,
1800000,34052317,1038,2056,3.79,12.57,4.64,2.1,304.66,,190,190,190,190,
1850000,34969772,711,1402,3.9,12.93,4.75,2.17,247.05,,190,190,190,190,
1900000,35955033,664,1309,4.01,13.38,4.87,2.23,227.51,,190,190,190,190,
1950000,36919935,593,1166,4.12,13.69,5.02,2.34,240.66,,190,190,190,190,
2000000,37879080,846,1674,4.25,14.25,5.17,2.42,265.78,,190,190,190,190,
2050000,38826564,562,1104,4.33,14.41,5.29,2.44,238.87,,190,190,190,190,
2100000,39796860,1075,2130,4.45,14.89,5.45,2.51,200.9,,190,190,190,190,
2150000,40753264,698,1376,4.57,15.47,5.58,2.59,225.85,,190,190,190,190,
2200000,41687628,1652,3284,4.68,15.78,5.74,2.64,191.86,,190,190,190,190,
2250000,42676754,949,1878,4.78,16.06,5.84,2.7,211.87,,190,190,190,190,
2300000,43640247,1108,2199,4.9,16.6,5.99,2.83,219.43,,190,190,190,190,
2350000,44569595,7875,15886,5.05,17.64,6.35,2.85,231.67,,190,190,190,190,
2400000,45529873,509,999,5.11,17.22,6.22,2.92,199.19,,190,190,190,190,
2450000,46502969,1385,2752,5.24,17.96,6.4,3.25,224.79,,190,190,190,190,
2500000,47473942,739,1458,5.34,18.14,6.51,3.41,204.44,,190,190,190,190,
2550000,48407592,532,1044,5.47,18.53,6.64,3.13,236.06,,190,190,190,190,
2600000,49356528,869,1718,5.57,19.45,6.78,3.26,225.7,,190,190,190,190,
2650000,50353329,734,1448,5.69,19.34,6.92,3.37,207.39,,190,190,190,190,
2700000,51309157,567,1114,5.76,19.91,7.05,3.36,231.68,,190,190,190,190,
2750000,52280543,1271,2524,5.91,20.35,7.24,3.42,218.43,,190,190,190,190,
2800000,53230505,511,1002,6.03,20.42,7.32,3.58,209.18,,190,190,190,190,
2850000,54187746,563,1106,6.16,21.03,7.49,3.57,190.19,,190,190,190,190,
2900000,55161765,1032,2045,6.24,21.38,7.61,3.62,199.89,,190,190,190,190,
2950000,56118167,1224,2429,6.35,21.78,7.78,3.74,225.98,,190,190,190,190,
3000000,57082768,980,1942,6.46,22.25,7.9,3.74,241.76,,190,190,190,190,
3050000,58006703,746,1472,6.59,22.8,8.02,3.82,237.72,,190,190,190,190,
3100000,58997271,1432,2845,6.68,23.26,8.2,3.98,196.85,,190,190,190,190,
3150000,60003314,496,974,6.8,23.44,8.3,3.96,229.18,,190,190,190,190,
3200000,60919212,1745,3474,6.92,24.05,8.5,4.13,197.47,,190,190,190,190,
3250000,61874877,474,928,7.04,24.43,8.59,4.11,191.67,,190,190,190,190,
3300000,62849494,1218,2417,7.13,24.83,8.74,4.63,211.13,,190,190,190,190,
3350000,63798424,1131,2242,7.23,25.41,8.85,4.25,217.74,,190,190,190,190,
3400000,64745452,699,1378,7.36,25.69,9,4.42,235.04,,190,190,190,190,
3450000,65700700,434,848,7.47,26.04,9.1,4.4,205.8,,190,190,190,190,
3500000,66733434,669,1319,7.59,26.4,9.31,4.71,200.11,,190,190,190,190,
3550000,67646422,698,1376,7.68,26.82,9.38,4.54,209.94,,190,190,190,190,
3600000,68604668,647,1274,7.79,27.53,9.53,4.6,227.55,,190,190,190,190,
3650000,69587650,484,948,7.92,27.66,9.69,4.81,238.75,,190,190,190,190,
3700000,70541366,883,1746,8.01,28.28,9.84,4.77,246.52,,190,190,190,190,
3750000,71517645,450,880,8.12,28.81,9.94,4.83,205.23,,190,190,190,190,
3800000,72483078,588,1156,8.28,29.11,10.1,4.93,229.07,,190,190,190,190,
3850000,73419536,571,1122,8.32,29.37,10.24,5.35,247.33,,190,190,190,190,
3900000,74381081,918,1816,8.47,30.02,10.37,5.15,215.47,,190,190,190,190,
3950000,75353990,715,1410,8.57,30.13,10.5,5.18,205.36,,190,190,190,190,
4000000,76311956,769,1518,8.7,30.67,10.66,5.21,231.4,,190,190,190,190,
4050000,77313266,668,1316,8.8,30.95,10.77,5.39,228.89,,190,190,190,190,
4100000,78249789,1168,2317,8.94,31.71,10.97,5.35,224.48,,190,190,190,190,
4150000,79221466,549,1078,9.05,32.28,11.08,5.53,265.09,,190,190,190,190,
4200000,80154009,512,1004,7.38,35.59,10.97,5.4,226.14,,190,190,190,190,
};
    \end{axis}
  \end{tikzpicture}

  \caption{Comparisons of the running time of \citet{LBK09}'s heuristic
    (left) with the running time of our search tree algorithm
    (right). Without interleaved data reduction, the search tree solved
    no instance in less than 5~minutes. The graphs were generated by
    adding $k=190$ random arcs between ten connected components, each
    being a preferential attachment graph with outdegree twenty and a
    single sink. The heuristic solved all 80~instances optimally.}
  \label{fig:runtime}
\end{figure}
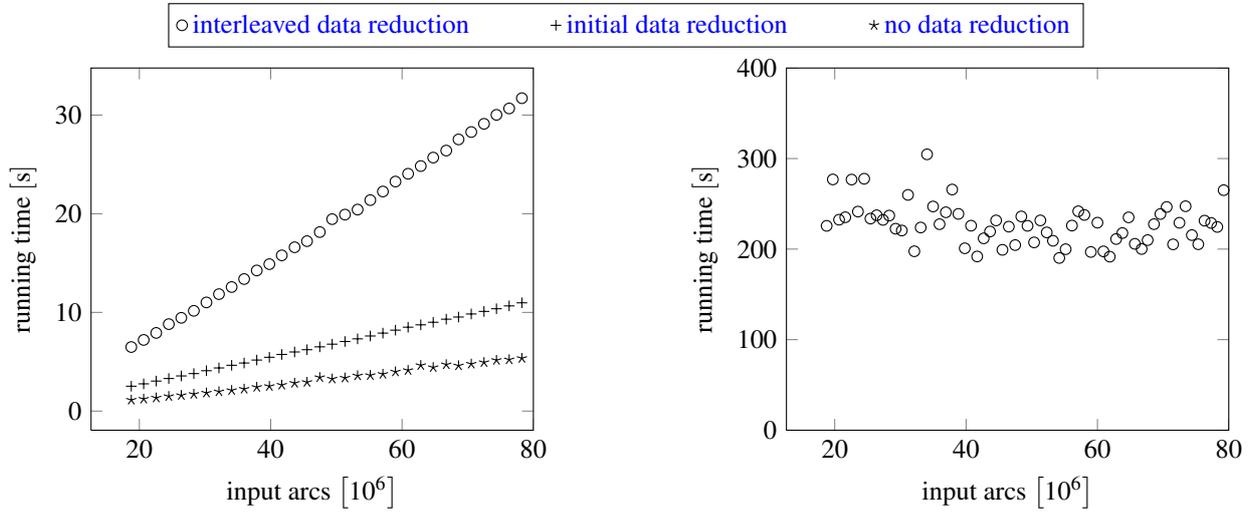

\autoref{fig:runtime} compares the running time of the heuristic of
\citet{LBK09} to the running time of our \autoref{alg:simple-st} with
increasing graph size. While the heuristic shows a linear increase of
running time with the graph size (on the left side), such a behavior
cannot be observed for the search tree algorithm (on the right
side). The reason for this can be seen in \autoref{fig:dreffect}: the
data reduction applied by \autoref{reducealg} initially shrinks most
input instances to about 2000~arcs in less than ten seconds. Thus, what
we observe in the right plot of \autoref{fig:runtime} is, to a large
extent, the running time of \autoref{alg:simple-st} for constant~$k=190$
and roughly constant graph size.  Our search tree algorithm allowed us
to verify that the heuristic by \citet{LBK09} solved all 80~generated
instances optimally regardless of the type of data reduction applied.

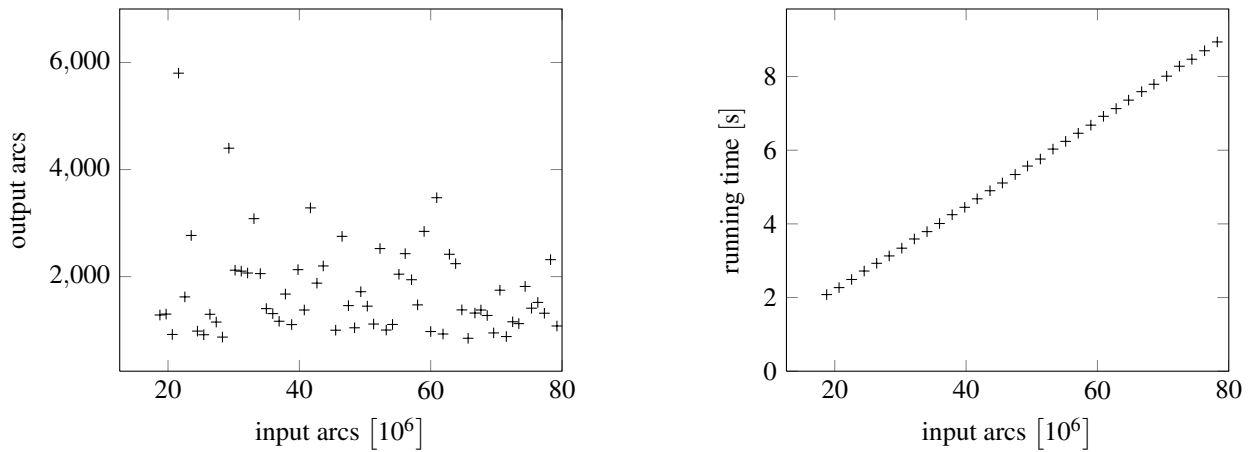
\begin{figure}
  \centering
  \begin{tikzpicture}
    \begin{axis}[ylabel=output arcs, xmax=80000000, xlabel=input arcs,
      change x base, axis base prefix={axis x base -6 prefix {}}, x
      unit=10^6, width=0.45\textwidth, ymax=7000]
      
      \addplot[color=black,mark=+,only marks] table [x=E,
      y=DRE, col sep=comma] {V,E,DRV,DRE,DRs,MCiDRs,MCDRs,MCs,STiDRs,STDRs,MCiDRk,MCDRk,MCk,STiDRk,STDRk
1000000,18767086,651,1283,2.08,6.49,2.52,1.12,225.68,,190,190,190,190,
1050000,19709159,660,1300,2.17,6.89,2.65,1.16,276.93,,190,190,190,190,
1100000,20655825,469,919,2.27,7.21,2.76,1.22,232.47,,190,190,190,190,
1150000,21604341,2901,5801,2.4,7.79,2.99,1.28,235.17,,190,190,190,190,
1200000,22564319,820,1621,2.49,7.92,3.04,1.34,276.79,,190,190,190,190,
1250000,23531472,1394,2769,2.6,8.49,3.21,1.42,241.42,,190,190,190,190,
1300000,24480730,501,982,2.72,8.81,3.3,1.49,277.68,,190,190,190,190,
1350000,25440966,467,914,2.82,9.04,3.26,1.53,233.7,,190,190,190,190,
1400000,26383694,658,1296,2.93,9.45,3.56,1.59,237.37,,190,190,190,190,
1450000,27342944,585,1151,3.03,9.8,3.7,1.67,232.26,,190,190,190,190,
1500000,28291769,445,870,3.13,10.16,3.8,1.71,237,,190,190,190,190,
1550000,29269274,2200,4399,3.28,10.78,4.03,1.88,222.51,,190,190,190,190,
1600000,30211184,1070,2120,3.34,11.01,4.09,1.84,220.52,,190,190,190,190,
1650000,31153242,1061,2102,3.47,11.38,4.23,1.96,259.86,,190,190,190,190,
1700000,32118855,1043,2067,3.59,11.85,4.37,1.97,197.58,,190,190,190,190,
1750000,33087214,1549,3083,3.7,12.16,4.53,2.1,223.84,,190,190,190,190,
1800000,34052317,1038,2056,3.79,12.57,4.64,2.1,304.66,,190,190,190,190,
1850000,34969772,711,1402,3.9,12.93,4.75,2.17,247.05,,190,190,190,190,
1900000,35955033,664,1309,4.01,13.38,4.87,2.23,227.51,,190,190,190,190,
1950000,36919935,593,1166,4.12,13.69,5.02,2.34,240.66,,190,190,190,190,
2000000,37879080,846,1674,4.25,14.25,5.17,2.42,265.78,,190,190,190,190,
2050000,38826564,562,1104,4.33,14.41,5.29,2.44,238.87,,190,190,190,190,
2100000,39796860,1075,2130,4.45,14.89,5.45,2.51,200.9,,190,190,190,190,
2150000,40753264,698,1376,4.57,15.47,5.58,2.59,225.85,,190,190,190,190,
2200000,41687628,1652,3284,4.68,15.78,5.74,2.64,191.86,,190,190,190,190,
2250000,42676754,949,1878,4.78,16.06,5.84,2.7,211.87,,190,190,190,190,
2300000,43640247,1108,2199,4.9,16.6,5.99,2.83,219.43,,190,190,190,190,
2350000,44569595,7875,15886,5.05,17.64,6.35,2.85,231.67,,190,190,190,190,
2400000,45529873,509,999,5.11,17.22,6.22,2.92,199.19,,190,190,190,190,
2450000,46502969,1385,2752,5.24,17.96,6.4,3.25,224.79,,190,190,190,190,
2500000,47473942,739,1458,5.34,18.14,6.51,3.41,204.44,,190,190,190,190,
2550000,48407592,532,1044,5.47,18.53,6.64,3.13,236.06,,190,190,190,190,
2600000,49356528,869,1718,5.57,19.45,6.78,3.26,225.7,,190,190,190,190,
2650000,50353329,734,1448,5.69,19.34,6.92,3.37,207.39,,190,190,190,190,
2700000,51309157,567,1114,5.76,19.91,7.05,3.36,231.68,,190,190,190,190,
2750000,52280543,1271,2524,5.91,20.35,7.24,3.42,218.43,,190,190,190,190,
2800000,53230505,511,1002,6.03,20.42,7.32,3.58,209.18,,190,190,190,190,
2850000,54187746,563,1106,6.16,21.03,7.49,3.57,190.19,,190,190,190,190,
2900000,55161765,1032,2045,6.24,21.38,7.61,3.62,199.89,,190,190,190,190,
2950000,56118167,1224,2429,6.35,21.78,7.78,3.74,225.98,,190,190,190,190,
3000000,57082768,980,1942,6.46,22.25,7.9,3.74,241.76,,190,190,190,190,
3050000,58006703,746,1472,6.59,22.8,8.02,3.82,237.72,,190,190,190,190,
3100000,58997271,1432,2845,6.68,23.26,8.2,3.98,196.85,,190,190,190,190,
3150000,60003314,496,974,6.8,23.44,8.3,3.96,229.18,,190,190,190,190,
3200000,60919212,1745,3474,6.92,24.05,8.5,4.13,197.47,,190,190,190,190,
3250000,61874877,474,928,7.04,24.43,8.59,4.11,191.67,,190,190,190,190,
3300000,62849494,1218,2417,7.13,24.83,8.74,4.63,211.13,,190,190,190,190,
3350000,63798424,1131,2242,7.23,25.41,8.85,4.25,217.74,,190,190,190,190,
3400000,64745452,699,1378,7.36,25.69,9,4.42,235.04,,190,190,190,190,
3450000,65700700,434,848,7.47,26.04,9.1,4.4,205.8,,190,190,190,190,
3500000,66733434,669,1319,7.59,26.4,9.31,4.71,200.11,,190,190,190,190,
3550000,67646422,698,1376,7.68,26.82,9.38,4.54,209.94,,190,190,190,190,
3600000,68604668,647,1274,7.79,27.53,9.53,4.6,227.55,,190,190,190,190,
3650000,69587650,484,948,7.92,27.66,9.69,4.81,238.75,,190,190,190,190,
3700000,70541366,883,1746,8.01,28.28,9.84,4.77,246.52,,190,190,190,190,
3750000,71517645,450,880,8.12,28.81,9.94,4.83,205.23,,190,190,190,190,
3800000,72483078,588,1156,8.28,29.11,10.1,4.93,229.07,,190,190,190,190,
3850000,73419536,571,1122,8.32,29.37,10.24,5.35,247.33,,190,190,190,190,
3900000,74381081,918,1816,8.47,30.02,10.37,5.15,215.47,,190,190,190,190,
3950000,75353990,715,1410,8.57,30.13,10.5,5.18,205.36,,190,190,190,190,
4000000,76311956,769,1518,8.7,30.67,10.66,5.21,231.4,,190,190,190,190,
4050000,77313266,668,1316,8.8,30.95,10.77,5.39,228.89,,190,190,190,190,
4100000,78249789,1168,2317,8.94,31.71,10.97,5.35,224.48,,190,190,190,190,
4150000,79221466,549,1078,9.05,32.28,11.08,5.53,265.09,,190,190,190,190,
4200000,80154009,512,1004,7.38,35.59,10.97,5.4,226.14,,190,190,190,190,
};
    \end{axis}
  \end{tikzpicture}\hfill{}
  \begin{tikzpicture}
    \begin{axis}[ymin=0,ylabel=running time, xmax=80000000,
      xlabel=input arcs,change x base, axis base prefix={axis x base -6
        prefix {}}, x unit=10^6, y unit=s, width=0.45\textwidth]
      
      \addplot[color=black,mark=+,only marks, each nth point=2] table
      [x=E, y=DRs, col sep=comma] {V,E,DRV,DRE,DRs,MCiDRs,MCDRs,MCs,STiDRs,STDRs,MCiDRk,MCDRk,MCk,STiDRk,STDRk
1000000,18767086,651,1283,2.08,6.49,2.52,1.12,225.68,,190,190,190,190,
1050000,19709159,660,1300,2.17,6.89,2.65,1.16,276.93,,190,190,190,190,
1100000,20655825,469,919,2.27,7.21,2.76,1.22,232.47,,190,190,190,190,
1150000,21604341,2901,5801,2.4,7.79,2.99,1.28,235.17,,190,190,190,190,
1200000,22564319,820,1621,2.49,7.92,3.04,1.34,276.79,,190,190,190,190,
1250000,23531472,1394,2769,2.6,8.49,3.21,1.42,241.42,,190,190,190,190,
1300000,24480730,501,982,2.72,8.81,3.3,1.49,277.68,,190,190,190,190,
1350000,25440966,467,914,2.82,9.04,3.26,1.53,233.7,,190,190,190,190,
1400000,26383694,658,1296,2.93,9.45,3.56,1.59,237.37,,190,190,190,190,
1450000,27342944,585,1151,3.03,9.8,3.7,1.67,232.26,,190,190,190,190,
1500000,28291769,445,870,3.13,10.16,3.8,1.71,237,,190,190,190,190,
1550000,29269274,2200,4399,3.28,10.78,4.03,1.88,222.51,,190,190,190,190,
1600000,30211184,1070,2120,3.34,11.01,4.09,1.84,220.52,,190,190,190,190,
1650000,31153242,1061,2102,3.47,11.38,4.23,1.96,259.86,,190,190,190,190,
1700000,32118855,1043,2067,3.59,11.85,4.37,1.97,197.58,,190,190,190,190,
1750000,33087214,1549,3083,3.7,12.16,4.53,2.1,223.84,,190,190,190,190,
1800000,34052317,1038,2056,3.79,12.57,4.64,2.1,304.66,,190,190,190,190,
1850000,34969772,711,1402,3.9,12.93,4.75,2.17,247.05,,190,190,190,190,
1900000,35955033,664,1309,4.01,13.38,4.87,2.23,227.51,,190,190,190,190,
1950000,36919935,593,1166,4.12,13.69,5.02,2.34,240.66,,190,190,190,190,
2000000,37879080,846,1674,4.25,14.25,5.17,2.42,265.78,,190,190,190,190,
2050000,38826564,562,1104,4.33,14.41,5.29,2.44,238.87,,190,190,190,190,
2100000,39796860,1075,2130,4.45,14.89,5.45,2.51,200.9,,190,190,190,190,
2150000,40753264,698,1376,4.57,15.47,5.58,2.59,225.85,,190,190,190,190,
2200000,41687628,1652,3284,4.68,15.78,5.74,2.64,191.86,,190,190,190,190,
2250000,42676754,949,1878,4.78,16.06,5.84,2.7,211.87,,190,190,190,190,
2300000,43640247,1108,2199,4.9,16.6,5.99,2.83,219.43,,190,190,190,190,
2350000,44569595,7875,15886,5.05,17.64,6.35,2.85,231.67,,190,190,190,190,
2400000,45529873,509,999,5.11,17.22,6.22,2.92,199.19,,190,190,190,190,
2450000,46502969,1385,2752,5.24,17.96,6.4,3.25,224.79,,190,190,190,190,
2500000,47473942,739,1458,5.34,18.14,6.51,3.41,204.44,,190,190,190,190,
2550000,48407592,532,1044,5.47,18.53,6.64,3.13,236.06,,190,190,190,190,
2600000,49356528,869,1718,5.57,19.45,6.78,3.26,225.7,,190,190,190,190,
2650000,50353329,734,1448,5.69,19.34,6.92,3.37,207.39,,190,190,190,190,
2700000,51309157,567,1114,5.76,19.91,7.05,3.36,231.68,,190,190,190,190,
2750000,52280543,1271,2524,5.91,20.35,7.24,3.42,218.43,,190,190,190,190,
2800000,53230505,511,1002,6.03,20.42,7.32,3.58,209.18,,190,190,190,190,
2850000,54187746,563,1106,6.16,21.03,7.49,3.57,190.19,,190,190,190,190,
2900000,55161765,1032,2045,6.24,21.38,7.61,3.62,199.89,,190,190,190,190,
2950000,56118167,1224,2429,6.35,21.78,7.78,3.74,225.98,,190,190,190,190,
3000000,57082768,980,1942,6.46,22.25,7.9,3.74,241.76,,190,190,190,190,
3050000,58006703,746,1472,6.59,22.8,8.02,3.82,237.72,,190,190,190,190,
3100000,58997271,1432,2845,6.68,23.26,8.2,3.98,196.85,,190,190,190,190,
3150000,60003314,496,974,6.8,23.44,8.3,3.96,229.18,,190,190,190,190,
3200000,60919212,1745,3474,6.92,24.05,8.5,4.13,197.47,,190,190,190,190,
3250000,61874877,474,928,7.04,24.43,8.59,4.11,191.67,,190,190,190,190,
3300000,62849494,1218,2417,7.13,24.83,8.74,4.63,211.13,,190,190,190,190,
3350000,63798424,1131,2242,7.23,25.41,8.85,4.25,217.74,,190,190,190,190,
3400000,64745452,699,1378,7.36,25.69,9,4.42,235.04,,190,190,190,190,
3450000,65700700,434,848,7.47,26.04,9.1,4.4,205.8,,190,190,190,190,
3500000,66733434,669,1319,7.59,26.4,9.31,4.71,200.11,,190,190,190,190,
3550000,67646422,698,1376,7.68,26.82,9.38,4.54,209.94,,190,190,190,190,
3600000,68604668,647,1274,7.79,27.53,9.53,4.6,227.55,,190,190,190,190,
3650000,69587650,484,948,7.92,27.66,9.69,4.81,238.75,,190,190,190,190,
3700000,70541366,883,1746,8.01,28.28,9.84,4.77,246.52,,190,190,190,190,
3750000,71517645,450,880,8.12,28.81,9.94,4.83,205.23,,190,190,190,190,
3800000,72483078,588,1156,8.28,29.11,10.1,4.93,229.07,,190,190,190,190,
3850000,73419536,571,1122,8.32,29.37,10.24,5.35,247.33,,190,190,190,190,
3900000,74381081,918,1816,8.47,30.02,10.37,5.15,215.47,,190,190,190,190,
3950000,75353990,715,1410,8.57,30.13,10.5,5.18,205.36,,190,190,190,190,
4000000,76311956,769,1518,8.7,30.67,10.66,5.21,231.4,,190,190,190,190,
4050000,77313266,668,1316,8.8,30.95,10.77,5.39,228.89,,190,190,190,190,
4100000,78249789,1168,2317,8.94,31.71,10.97,5.35,224.48,,190,190,190,190,
4150000,79221466,549,1078,9.05,32.28,11.08,5.53,265.09,,190,190,190,190,
4200000,80154009,512,1004,7.38,35.59,10.97,5.4,226.14,,190,190,190,190,
};
    \end{axis}
  \end{tikzpicture}

  \caption{Effect (left) and running time (right) of initially running
    \autoref{reducealg} for data reduction. The graphs were generated by
    adding $k=190$ random arcs between ten connected components, each
    being a preferential attachment graph with outdegree twenty and a
    single sink.}
  \label{fig:dreffect}
\end{figure}

Finally, \autoref{fig:prefat} presents instances that could not be
optimally solved by \citet{LBK09}'s heuristic. In the left plot, we see
that in instances with large embedded \partSet{}s of several hundred thousand arcs, the heuristic of
\citet{LBK09} does not find the embedded \partSet{} but an about
5\textperthousand{} larger one. In all cases, the heuristic found the
same \partSet{}s regardless of the type of data reduction applied. Note
that the plot only gives a lower bound on the deviation factor, since
there might be even better \partSet{}s in the instances than the
embedded one; we were unable to compute the optimal \partSet{}s in these
instances.  In the right plot of \autoref{fig:prefat}, we used smaller
preferential attachment graphs (this time without embedded \partSet{}s)
and see that \citet{LBK09}'s heuristic can be off by more than a factor of
two from the optimal \partSet{}. Data reduction had no effect on the
quality of the \partSet{}s found.

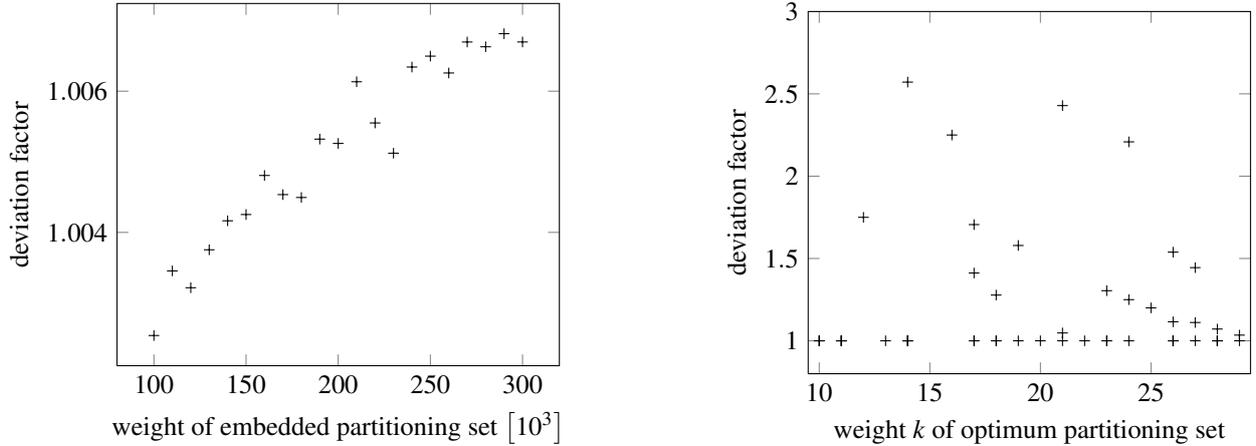
\begin{figure}
  \centering
  \begin{tikzpicture}
    \begin{axis}[xlabel=weight of embedded \partSet{}, ylabel=deviation
      factor, width=0.45\textwidth, yticklabel
      style={/pgf/number format/fixed, /pgf/number format/precision=4},
      change x base, axis base prefix={axis x base -3 prefix {}}, x
      unit=10^3]

      \addplot[color=black,mark=+, only marks] table [x=OPT1, y=DIV1,
      col sep=comma] {F,V1,E1,DRV1,DRE1,DRt1,MCiDRk1,MCiDRs1,MCDRk1,MCDRs1,MCk1,MCs1,STiDRk1,STiDRs1,STDRk1,STDRs1,STk1,STs1,DIViDR1,DIVDR1,DIV1,OPT1,V2,E2,DRV2,DRE2,DRt2,MCiDRk2,MCiDRs2,MCDRk2,MCDRs2,MCk2,MCs2,STiDRk2,STiDRs2,STDRk2,STDRs2,STk2,STs2,DIViDR2,DIVDR2,DIV2,OPT2,V3,E3,DRV3,DRE3,DRt3,MCiDRk3,MCiDRs3,MCDRk3,MCDRs3,MCk3,MCs3,STiDRk3,STiDRs3,STDRk3,STDRs3,STk3,STs3,DIViDR3,DIVDR3,DIV3,OPT3,V4,E4,DRV4,DRE4,DRt4,MCiDRk4,MCiDRs4,MCDRk4,MCDRs4,MCk4,MCs4,STiDRk4,STiDRs4,STDRk4,STDRs4,STk4,STs4,DIViDR4,DIVDR4,DIV4,OPT4,V5,E5,DRV5,DRE5,DRt5,MCiDRk5,MCiDRs5,MCDRk5,MCDRs5,MCk5,MCs5,STiDRk5,STiDRs5,STDRk5,STDRs5,STk5,STs5,DIViDR5,DIVDR5,DIV5,OPT5,
dagpart:-A:-c:10:-n:100000:-o:2:-k:100000:-t:5:1,1000000,2098871,169288,355285,1.08,100254,4,100254,2.48,100254,0.62,,TIMEOUT 41s,,,,,1.00254000000000000000,1.00254000000000000000,1.00254000000000000000,100000,1000000,2098727,166761,350133,1.07,100270,4.01,100270,2.48,100270,0.62,,TIMEOUT 41s,,,,,1.00270000000000000000,1.00270000000000000000,1.00270000000000000000,100000,1000000,2098718,174263,365775,1.07,100276,4.05,100276,2.49,100276,0.63,,TIMEOUT 42s,,,,,1.00276000000000000000,1.00276000000000000000,1.00276000000000000000,100000,1000000,2098748,170269,357350,1.08,100297,4.01,100297,2.49,100297,0.63,,TIMEOUT 41s,,,,,1.00297000000000000000,1.00297000000000000000,1.00297000000000000000,100000,1000000,2098716,167831,352032,1.08,100331,4,100331,2.49,100331,0.63,,TIMEOUT 41s,,,,,1.00331000000000000000,1.00331000000000000000,1.00331000000000000000,100000,
dagpart:-A:-c:10:-n:100000:-o:2:-k:110000:-t:5:1,1000000,2108657,180465,380707,1.09,110380,4.08,110380,2.49,110380,0.63,,TIMEOUT 41s,,,,,1.00345454545454545454,1.00345454545454545454,1.00345454545454545454,110000,1000000,2108649,182016,383844,1.1,110333,4.07,110333,2.5,110333,0.64,,TIMEOUT 42s,,,,,1.00302727272727272727,1.00302727272727272727,1.00302727272727272727,110000,1000000,2108745,193444,408125,1.11,110367,4.13,110367,2.5,110367,0.63,,TIMEOUT 41s,,,,,1.00333636363636363636,1.00333636363636363636,1.00333636363636363636,110000,1000000,2108690,187046,394383,1.09,110327,4.1,110327,2.51,110327,0.63,,TIMEOUT 42s,,,,,1.00297272727272727272,1.00297272727272727272,1.00297272727272727272,110000,1000000,2108686,184396,388793,1.1,110393,4.09,110393,2.49,110393,0.63,,TIMEOUT 41s,,,,,1.00357272727272727272,1.00357272727272727272,1.00357272727272727272,110000,
dagpart:-A:-c:10:-n:100000:-o:2:-k:120000:-t:5:1,1000000,2118693,198200,420015,1.13,120386,4.17,120386,2.52,120386,0.64,,,,,,,1.00321666666666666666,1.00321666666666666666,1.00321666666666666666,120000,1000000,2118784,196928,417565,1.13,120418,4.14,120418,2.49,120418,0.63,,,,,,,1.00348333333333333333,1.00348333333333333333,1.00348333333333333333,120000,1000000,2118731,207374,439512,1.13,120406,4.21,120406,2.52,120406,0.64,,,,,,,1.00338333333333333333,1.00338333333333333333,1.00338333333333333333,120000,1000000,2118641,202470,428975,1.13,120434,4.18,120434,2.51,120434,0.64,,,,,,,1.00361666666666666666,1.00361666666666666666,1.00361666666666666666,120000,1000000,2118780,198951,421599,1.12,120372,4.16,120372,2.35,120372,0.61,,,,,,,1.00310000000000000000,1.00310000000000000000,1.00310000000000000000,120000,
dagpart:-A:-c:10:-n:100000:-o:2:-k:130000:-t:5:1,1000000,2128766,216857,461629,1.14,130488,4.26,130488,2.53,130488,0.64,,,,,,,1.00375384615384615384,1.00375384615384615384,1.00375384615384615384,130000,1000000,2128602,205349,437192,1.14,130469,4.21,130469,2.51,130469,0.65,,,,,,,1.00360769230769230769,1.00360769230769230769,1.00360769230769230769,130000,1000000,2128689,209671,446517,1.14,130425,4.24,130425,2.52,130425,0.64,,,,,,,1.00326923076923076923,1.00326923076923076923,1.00326923076923076923,130000,1000000,2128607,208905,445017,1.14,130467,4.22,130467,2.52,130467,0.64,,,,,,,1.00359230769230769230,1.00359230769230769230,1.00359230769230769230,130000,1000000,2128804,207599,442001,1.14,130403,4.23,130403,2.51,130403,0.63,,,,,,,1.00310000000000000000,1.00310000000000000000,1.00310000000000000000,130000,
dagpart:-A:-c:10:-n:100000:-o:2:-k:140000:-t:5:1,1000000,2138600,220044,470818,1.17,140583,4.3,140583,2.53,140583,0.65,,,,,,,1.00416428571428571428,1.00416428571428571428,1.00416428571428571428,140000,1000000,2138649,229995,492265,1.17,140594,4.33,140594,2.54,140594,0.65,,,,,,,1.00424285714285714285,1.00424285714285714285,1.00424285714285714285,140000,1000000,2138716,223297,477814,1.17,140541,4.3,140541,2.53,140541,0.65,,,,,,,1.00386428571428571428,1.00386428571428571428,1.00386428571428571428,140000,1000000,2138743,224510,480283,1.17,140570,4.29,140570,2.52,140570,0.64,,,,,,,1.00407142857142857142,1.00407142857142857142,1.00407142857142857142,140000,1000000,2138769,216819,463559,1.16,140499,4.28,140499,2.52,140499,0.64,,,,,,,1.00356428571428571428,1.00356428571428571428,1.00356428571428571428,140000,
dagpart:-A:-c:10:-n:100000:-o:2:-k:150000:-t:5:1,1000000,2148684,239163,513971,1.18,150638,4.38,150638,2.54,150638,0.65,,,,,,,1.00425333333333333333,1.00425333333333333333,1.00425333333333333333,150000,1000000,2148735,231857,498184,1.18,150563,4.35,150563,2.55,150563,0.65,,,,,,,1.00375333333333333333,1.00375333333333333333,1.00375333333333333333,150000,1000000,2148768,231288,496921,1.19,150604,4.37,150604,2.54,150604,0.66,,,,,,,1.00402666666666666666,1.00402666666666666666,1.00402666666666666666,150000,1000000,2148732,236596,508555,1.18,150639,4.38,150639,2.52,150639,0.65,,,,,,,1.00426000000000000000,1.00426000000000000000,1.00426000000000000000,150000,1000000,2148768,233552,501908,1.19,150611,4.37,150611,2.55,150611,0.66,,,,,,,1.00407333333333333333,1.00407333333333333333,1.00407333333333333333,150000,
dagpart:-A:-c:10:-n:100000:-o:2:-k:160000:-t:5:1,1000000,2158694,252062,544405,1.2,160769,4.46,160769,2.56,160769,0.66,,,,,,,1.00480625000000000000,1.00480625000000000000,1.00480625000000000000,160000,1000000,2158608,250113,540249,1.2,160677,4.45,160677,2.55,160677,0.67,,,,,,,1.00423125000000000000,1.00423125000000000000,1.00423125000000000000,160000,1000000,2158692,245107,528934,1.21,160674,4.43,160674,2.54,160674,0.67,,,,,,,1.00421250000000000000,1.00421250000000000000,1.00421250000000000000,160000,1000000,2158835,304409,657829,1.2,160676,4.42,160676,2.44,160676,0.68,,,,,,,1.00422500000000000000,1.00422500000000000000,1.00422500000000000000,160000,1000000,2158740,245507,529670,1.22,160847,4.43,160847,2.57,160847,0.66,,,,,,,1.00529375000000000000,1.00529375000000000000,1.00529375000000000000,160000,
dagpart:-A:-c:10:-n:100000:-o:2:-k:170000:-t:5:1,1000000,2168649,257513,558369,1.22,170771,4.49,170771,2.56,170771,0.67,,,,,,,1.00453529411764705882,1.00453529411764705882,1.00453529411764705882,170000,1000000,2168683,262475,569692,1.23,170887,4.53,170887,2.54,170887,0.68,,,,,,,1.00521764705882352941,1.00521764705882352941,1.00521764705882352941,170000,1000000,2168872,257627,559037,1.22,170675,4.5,170675,2.55,170675,0.67,,,,,,,1.00397058823529411764,1.00397058823529411764,1.00397058823529411764,170000,1000000,2168740,257962,559665,1.23,170774,4.49,170774,2.55,170774,0.67,,,,,,,1.00455294117647058823,1.00455294117647058823,1.00455294117647058823,170000,1000000,2168735,266142,577316,1.22,170789,4.54,170789,2.58,170789,0.67,,,,,,,1.00464117647058823529,1.00464117647058823529,1.00464117647058823529,170000,
dagpart:-A:-c:10:-n:100000:-o:2:-k:180000:-t:5:1,1000000,2178654,275685,600881,1.25,180809,4.58,180809,2.58,180809,0.68,,,,,,,1.00449444444444444444,1.00449444444444444444,1.00449444444444444444,180000,1000000,2178561,271851,592462,1.23,180855,4.58,180855,2.54,180855,0.67,,,,,,,1.00475000000000000000,1.00475000000000000000,1.00475000000000000000,180000,1000000,2178751,271700,592271,1.25,180836,4.55,180836,2.57,180836,0.68,,,,,,,1.00464444444444444444,1.00464444444444444444,1.00464444444444444444,180000,1000000,2178581,274668,598300,1.24,180775,4.58,180775,2.58,180775,0.68,,,,,,,1.00430555555555555555,1.00430555555555555555,1.00430555555555555555,180000,1000000,2178662,271874,592746,1.24,180887,4.59,180887,2.58,180887,0.68,,,,,,,1.00492777777777777777,1.00492777777777777777,1.00492777777777777777,180000,
dagpart:-A:-c:10:-n:100000:-o:2:-k:190000:-t:5:1,1000000,2188659,284104,621994,1.26,191011,4.63,191011,2.59,191011,0.69,,,,,,,1.00532105263157894736,1.00532105263157894736,1.00532105263157894736,190000,1000000,2188721,285658,625533,1.27,191095,4.64,191095,2.6,191095,0.69,,,,,,,1.00576315789473684210,1.00576315789473684210,1.00576315789473684210,190000,1000000,2188691,280366,613496,1.25,190978,4.58,190978,2.58,190978,0.68,,,,,,,1.00514736842105263157,1.00514736842105263157,1.00514736842105263157,190000,1000000,2188708,286055,626530,1.26,190881,4.63,190881,2.6,190881,0.69,,,,,,,1.00463684210526315789,1.00463684210526315789,1.00463684210526315789,190000,1000000,2188785,282609,618678,1.26,190947,4.63,190947,2.58,190947,0.68,,,,,,,1.00498421052631578947,1.00498421052631578947,1.00498421052631578947,190000,
dagpart:-A:-c:10:-n:100000:-o:2:-k:200000:-t:5:1,1000000,2198744,292044,641931,1.28,201052,4.71,201052,2.59,201052,0.7,,,,,,,1.00526000000000000000,1.00526000000000000000,1.00526000000000000000,200000,1000000,2198738,300311,660321,1.26,201013,4.69,201013,2.58,201013,0.71,,,,,,,1.00506500000000000000,1.00506500000000000000,1.00506500000000000000,200000,1000000,2198608,296246,651487,1.28,201088,4.68,201088,2.57,201088,0.69,,,,,,,1.00544000000000000000,1.00544000000000000000,1.00544000000000000000,200000,1000000,2198536,291200,640097,1.28,201061,4.68,201061,2.59,201061,0.7,,,,,,,1.00530500000000000000,1.00530500000000000000,1.00530500000000000000,200000,1000000,2198747,293304,644659,1.28,201009,4.67,201009,2.6,201009,0.7,,,,,,,1.00504500000000000000,1.00504500000000000000,1.00504500000000000000,200000,
dagpart:-A:-c:10:-n:100000:-o:2:-k:210000:-t:5:1,1000000,2208632,305952,676009,1.28,211288,4.76,211288,2.6,211288,0.71,,,,,,,1.00613333333333333333,1.00613333333333333333,1.00613333333333333333,210000,1000000,2208704,306792,677816,1.29,211069,4.72,211069,2.61,211069,0.71,,,,,,,1.00509047619047619047,1.00509047619047619047,1.00509047619047619047,210000,1000000,2208700,312588,691087,1.28,211179,4.8,211179,2.61,211179,0.7,,,,,,,1.00561428571428571428,1.00561428571428571428,1.00561428571428571428,210000,1000000,2208705,306132,676462,1.29,211136,4.76,211136,2.58,211136,0.71,,,,,,,1.00540952380952380952,1.00540952380952380952,1.00540952380952380952,210000,1000000,2208580,307022,678542,1.29,211077,4.74,211077,2.62,211077,0.7,,,,,,,1.00512857142857142857,1.00512857142857142857,1.00512857142857142857,210000,
dagpart:-A:-c:10:-n:100000:-o:2:-k:220000:-t:5:1,1000000,2218626,313833,696341,1.3,221221,4.84,221221,2.61,221221,0.72,,,,,,,1.00555000000000000000,1.00555000000000000000,1.00555000000000000000,220000,1000000,2218688,318074,706218,1.3,221208,4.85,221208,2.46,221208,0.7,,,,,,,1.00549090909090909090,1.00549090909090909090,1.00549090909090909090,220000,1000000,2218756,315946,701035,1.31,221219,4.84,221219,2.6,221219,0.72,,,,,,,1.00554090909090909090,1.00554090909090909090,1.00554090909090909090,220000,1000000,2218745,312341,693310,1.3,221367,4.8,221367,2.62,221367,0.71,,,,,,,1.00621363636363636363,1.00621363636363636363,1.00621363636363636363,220000,1000000,2218769,320420,711378,1.31,221154,4.87,221154,2.62,221154,0.72,,,,,,,1.00524545454545454545,1.00524545454545454545,1.00524545454545454545,220000,
dagpart:-A:-c:10:-n:100000:-o:2:-k:230000:-t:5:1,1000000,2228658,361431,807046,1.32,231178,5.03,231178,2.64,231178,0.73,,,,,,,1.00512173913043478260,1.00512173913043478260,1.00512173913043478260,230000,1000000,2228608,318308,709697,1.31,231296,4.89,231296,2.62,231296,0.73,,,,,,,1.00563478260869565217,1.00563478260869565217,1.00563478260869565217,230000,1000000,2228638,321081,715399,1.32,231374,4.88,231374,2.61,231374,0.72,,,,,,,1.00597391304347826086,1.00597391304347826086,1.00597391304347826086,230000,1000000,2228782,323009,719750,1.32,231138,4.91,231138,2.63,231138,0.73,,,,,,,1.00494782608695652173,1.00494782608695652173,1.00494782608695652173,230000,1000000,2228569,323371,720475,1.31,231331,4.91,231331,2.63,231331,0.72,,,,,,,1.00578695652173913043,1.00578695652173913043,1.00578695652173913043,230000,
dagpart:-A:-c:10:-n:100000:-o:2:-k:240000:-t:5:1,1000000,2238768,337667,756498,1.35,241522,4.97,241522,2.65,241522,0.73,,,,,,,1.00634166666666666666,1.00634166666666666666,1.00634166666666666666,240000,1000000,2238769,340028,762203,1.33,241399,4.96,241399,2.64,241399,0.73,,,,,,,1.00582916666666666666,1.00582916666666666666,1.00582916666666666666,240000,1000000,2238645,332559,744899,1.34,241353,4.94,241353,2.64,241353,0.73,,,,,,,1.00563750000000000000,1.00563750000000000000,1.00563750000000000000,240000,1000000,2238713,334884,750345,1.34,241377,4.98,241377,2.65,241377,0.73,,,,,,,1.00573750000000000000,1.00573750000000000000,1.00573750000000000000,240000,1000000,2238731,332653,744982,1.35,241474,4.97,241474,2.64,241474,0.74,,,,,,,1.00614166666666666666,1.00614166666666666666,1.00614166666666666666,240000,
dagpart:-A:-c:10:-n:100000:-o:2:-k:250000:-t:5:1,1000000,2248692,350131,787773,1.35,251624,5.07,251624,2.66,251624,0.73,,,,,,,1.00649600000000000000,1.00649600000000000000,1.00649600000000000000,250000,1000000,2248662,348516,784112,1.35,251676,5.04,251676,2.68,251676,0.74,,,,,,,1.00670400000000000000,1.00670400000000000000,1.00670400000000000000,250000,1000000,2248776,343885,773650,1.35,251560,5.03,251560,2.65,251560,0.73,,,,,,,1.00624000000000000000,1.00624000000000000000,1.00624000000000000000,250000,1000000,2248683,345637,777686,1.36,251622,5.04,251622,2.67,251622,0.73,,,,,,,1.00648800000000000000,1.00648800000000000000,1.00648800000000000000,250000,1000000,2248597,346661,779628,1.35,251542,5.04,251542,2.64,251542,0.74,,,,,,,1.00616800000000000000,1.00616800000000000000,1.00616800000000000000,250000,
dagpart:-A:-c:10:-n:100000:-o:2:-k:260000:-t:5:1,1000000,2258807,357913,809077,1.37,261627,5.11,261627,2.7,261627,0.77,,,,,,,1.00625769230769230769,1.00625769230769230769,1.00625769230769230769,260000,1000000,2258697,356233,805208,1.36,261592,5.09,261592,2.67,261592,0.74,,,,,,,1.00612307692307692307,1.00612307692307692307,1.00612307692307692307,260000,1000000,2258732,354685,801795,1.37,261509,5.09,261509,2.68,261509,0.73,,,,,,,1.00580384615384615384,1.00580384615384615384,1.00580384615384615384,260000,1000000,2258757,351351,794011,1.36,261655,5.09,261655,2.65,261655,0.74,,,,,,,1.00636538461538461538,1.00636538461538461538,1.00636538461538461538,260000,1000000,2258631,352852,797149,1.37,261574,5.1,261574,2.7,261574,0.74,,,,,,,1.00605384615384615384,1.00605384615384615384,1.00605384615384615384,260000,
dagpart:-A:-c:10:-n:100000:-o:2:-k:270000:-t:5:1,1000000,2268691,361474,820697,1.38,271808,5.17,271808,2.7,271808,0.74,,,,,,,1.00669629629629629629,1.00669629629629629629,1.00669629629629629629,270000,1000000,2268644,362093,821866,1.37,271751,5.15,271751,2.68,271751,0.75,,,,,,,1.00648518518518518518,1.00648518518518518518,1.00648518518518518518,270000,1000000,2268726,362557,823030,1.38,271621,5.16,271621,2.69,271621,0.74,,,,,,,1.00600370370370370370,1.00600370370370370370,1.00600370370370370370,270000,1000000,2268742,376682,855524,1.39,271682,5.25,271682,2.7,271682,0.75,,,,,,,1.00622962962962962962,1.00622962962962962962,1.00622962962962962962,270000,1000000,2268685,362337,822312,1.38,271760,5.13,271760,2.66,271760,0.75,,,,,,,1.00651851851851851851,1.00651851851851851851,1.00651851851851851851,270000,
dagpart:-A:-c:10:-n:100000:-o:2:-k:280000:-t:5:1,1000000,2278722,380784,868590,1.4,281856,5.28,281856,2.58,281856,0.74,,,,,,,1.00662857142857142857,1.00662857142857142857,1.00662857142857142857,280000,1000000,2278756,380115,867063,1.39,281931,5.28,281931,2.7,281931,0.76,,,,,,,1.00689642857142857142,1.00689642857142857142,1.00689642857142857142,280000,1000000,2278611,369348,842441,1.39,281899,5.21,281899,2.71,281899,0.75,,,,,,,1.00678214285714285714,1.00678214285714285714,1.00678214285714285714,280000,1000000,2278759,374822,854933,1.39,281815,5.26,281815,2.71,281815,0.75,,,,,,,1.00648214285714285714,1.00648214285714285714,1.00648214285714285714,280000,1000000,2278765,374389,854229,1.4,281717,5.26,281717,2.72,281717,0.76,,,,,,,1.00613214285714285714,1.00613214285714285714,1.00613214285714285714,280000,
dagpart:-A:-c:10:-n:100000:-o:2:-k:290000:-t:5:1,1000000,2288773,384608,881648,1.41,291976,5.32,291976,2.73,291976,0.76,,,,,,,1.00681379310344827586,1.00681379310344827586,1.00681379310344827586,290000,1000000,2288649,390000,893495,1.41,292079,5.35,292079,2.72,292079,0.77,,,,,,,1.00716896551724137931,1.00716896551724137931,1.00716896551724137931,290000,1000000,2288673,381487,873873,1.41,292123,5.32,292123,2.73,292123,0.76,,,,,,,1.00732068965517241379,1.00732068965517241379,1.00732068965517241379,290000,1000000,2288661,379151,868496,1.41,291881,5.3,291881,2.73,291881,0.76,,,,,,,1.00648620689655172413,1.00648620689655172413,1.00648620689655172413,290000,1000000,2288760,377997,865306,1.41,291929,5.29,291929,2.72,291929,0.76,,,,,,,1.00665172413793103448,1.00665172413793103448,1.00665172413793103448,290000,
dagpart:-A:-c:10:-n:100000:-o:2:-k:300000:-t:5:1,1000000,2298603,391124,899883,1.42,302009,5.37,302009,2.74,302009,0.76,,,,,,,1.00669666666666666666,1.00669666666666666666,1.00669666666666666666,300000,1000000,2298666,399048,918565,1.42,302039,5.41,302039,2.74,302039,0.76,,,,,,,1.00679666666666666666,1.00679666666666666666,1.00679666666666666666,300000,1000000,2298862,387149,890492,1.42,302345,5.39,302345,2.74,302345,0.77,,,,,,,1.00781666666666666666,1.00781666666666666666,1.00781666666666666666,300000,1000000,2298739,388921,895051,1.42,302094,5.33,302094,2.73,302094,0.76,,,,,,,1.00698000000000000000,1.00698000000000000000,1.00698000000000000000,300000,1000000,2298647,390824,899043,1.43,302071,5.37,302071,2.73,302071,0.77,,,,,,,1.00690333333333333333,1.00690333333333333333,1.00690333333333333333,300000,
}; 
    \end{axis}
  \end{tikzpicture}\hfill{}
  \begin{tikzpicture}
    \begin{axis}[xlabel=weight~$k$ of optimum \partSet{},
      ylabel=deviation factor, width=0.45\textwidth, 
      xmin=9.5, ymax=3, xmax=29.5]

      \addplot[color=black,mark=+,only marks] table [x=STiDRk, y
      expr=(\thisrow{MCk}/\thisrow{STiDRk}), col sep=comma]
      {F,V,E,DRV,DRE,MCiDRk,MCDRk,MCk,STiDRk,STDRk
iter-dagpart:-C:-c:2:-n:101:-o:3:-t:4000,103,282,102,281,74,51,51,21,
iter-dagpart:-C:-c:2:-n:102:-o:3:-t:4000,104,289,104,289,28,21,21,21,
iter-dagpart:-C:-c:2:-n:104:-o:3:-t:4000,106,283,106,283,77,45,45,33,
iter-dagpart:-C:-c:2:-n:105:-o:3:-t:4000,107,291,57,126,11,11,11,11,
iter-dagpart:-C:-c:2:-n:106:-o:3:-t:4000,108,289,97,240,33,28,28,28,
iter-dagpart:-C:-c:2:-n:107:-o:3:-t:4000,109,297,109,297,75,21,21,12,
iter-dagpart:-C:-c:2:-n:110:-o:3:-t:4000,112,302,111,301,76,27,27,27,
iter-dagpart:-C:-c:2:-n:113:-o:3:-t:4000,115,313,115,313,74,34,34,34,
iter-dagpart:-C:-c:2:-n:115:-o:3:-t:4000,117,313,111,300,40,40,40,26,
iter-dagpart:-C:-c:2:-n:116:-o:3:-t:4000,118,317,73,161,22,22,22,21,
iter-dagpart:-C:-c:2:-n:118:-o:3:-t:4000,120,331,116,312,37,29,29,29,
iter-dagpart:-C:-c:2:-n:119:-o:3:-t:4000,121,327,60,127,29,29,29,17,
iter-dagpart:-C:-c:2:-n:120:-o:3:-t:4000,122,327,122,327,74,8,8,8,
iter-dagpart:-C:-c:2:-n:121:-o:3:-t:4000,123,334,123,334,88,53,53,24,
iter-dagpart:-C:-c:2:-n:122:-o:3:-t:4000,124,338,122,331,50,50,50,14,
iter-dagpart:-C:-c:2:-n:123:-o:3:-t:4000,125,344,123,334,39,39,39,27,
iter-dagpart:-C:-c:2:-n:126:-o:3:-t:4000,128,350,127,349,36,36,36,16,
iter-dagpart:-C:-c:2:-n:128:-o:3:-t:4000,130,343,117,293,38,32,32,32,
iter-dagpart:-C:-c:2:-n:50:-o:3:-t:4000,52,134,5,6,3,3,3,3,
iter-dagpart:-C:-c:2:-n:52:-o:3:-t:4000,54,141,52,138,21,19,19,19,
iter-dagpart:-C:-c:2:-n:53:-o:3:-t:4000,55,131,51,126,34,17,17,17,
iter-dagpart:-C:-c:2:-n:54:-o:3:-t:4000,56,141,43,94,30,26,26,26,
iter-dagpart:-C:-c:2:-n:55:-o:3:-t:4000,57,144,52,131,21,10,10,10,
iter-dagpart:-C:-c:2:-n:56:-o:3:-t:4000,58,145,57,142,18,13,13,13,
iter-dagpart:-C:-c:2:-n:57:-o:3:-t:4000,59,149,59,149,23,23,23,18,
iter-dagpart:-C:-c:2:-n:58:-o:3:-t:4000,60,158,58,152,25,23,23,23,
iter-dagpart:-C:-c:2:-n:59:-o:3:-t:4000,61,156,53,130,30,30,30,19,
iter-dagpart:-C:-c:2:-n:60:-o:3:-t:4000,62,158,62,158,48,28,28,28,
iter-dagpart:-C:-c:2:-n:61:-o:3:-t:4000,63,165,61,153,36,28,28,28,
iter-dagpart:-C:-c:2:-n:62:-o:3:-t:4000,64,169,51,118,10,10,10,10,
iter-dagpart:-C:-c:2:-n:63:-o:3:-t:4000,65,168,61,156,31,29,29,26,
iter-dagpart:-C:-c:2:-n:64:-o:3:-t:4000,66,170,66,170,43,30,30,23,
iter-dagpart:-C:-c:2:-n:65:-o:3:-t:4000,67,171,12,21,8,7,7,7,
iter-dagpart:-C:-c:2:-n:66:-o:3:-t:4000,68,171,45,97,15,14,14,14,
iter-dagpart:-C:-c:2:-n:67:-o:3:-t:4000,69,179,68,178,33,30,30,28,
iter-dagpart:-C:-c:2:-n:68:-o:3:-t:4000,70,179,70,179,25,18,18,18,
iter-dagpart:-C:-c:2:-n:69:-o:3:-t:4000,71,181,68,174,44,14,14,14,
iter-dagpart:-C:-c:2:-n:70:-o:3:-t:4000,72,192,71,191,35,30,30,24,
iter-dagpart:-C:-c:2:-n:71:-o:3:-t:4000,73,197,72,195,42,11,11,11,
iter-dagpart:-C:-c:2:-n:72:-o:3:-t:4000,74,190,74,190,55,22,22,22,
iter-dagpart:-C:-c:2:-n:73:-o:3:-t:4000,75,196,44,96,33,18,18,18,
iter-dagpart:-C:-c:2:-n:74:-o:3:-t:4000,76,204,7,10,4,4,4,4,
iter-dagpart:-C:-c:2:-n:76:-o:3:-t:4000,78,206,65,164,19,11,11,11,
iter-dagpart:-C:-c:2:-n:77:-o:3:-t:4000,79,215,79,215,54,24,24,17,
iter-dagpart:-C:-c:2:-n:78:-o:3:-t:4000,80,211,55,121,19,17,17,17,
iter-dagpart:-C:-c:2:-n:79:-o:3:-t:4000,81,226,72,178,41,41,41,37,
iter-dagpart:-C:-c:2:-n:80:-o:3:-t:4000,82,222,79,204,51,26,26,26,
iter-dagpart:-C:-c:2:-n:81:-o:3:-t:4000,83,220,54,114,11,10,10,10,
iter-dagpart:-C:-c:2:-n:82:-o:3:-t:4000,84,220,80,203,27,20,20,20,
iter-dagpart:-C:-c:2:-n:83:-o:3:-t:4000,85,221,81,209,52,11,11,11,
iter-dagpart:-C:-c:2:-n:84:-o:3:-t:4000,86,225,85,224,18,14,14,14,
iter-dagpart:-C:-c:2:-n:86:-o:3:-t:4000,88,238,85,224,48,48,48,31,
iter-dagpart:-C:-c:2:-n:87:-o:3:-t:4000,89,241,89,241,64,30,30,25,
iter-dagpart:-C:-c:2:-n:88:-o:3:-t:4000,90,239,20,37,16,14,14,14,
iter-dagpart:-C:-c:2:-n:89:-o:3:-t:4000,91,234,90,232,64,30,30,29,
iter-dagpart:-C:-c:2:-n:90:-o:3:-t:4000,92,245,92,245,61,36,36,14,
iter-dagpart:-C:-c:2:-n:91:-o:3:-t:4000,93,244,93,244,38,30,30,27,
iter-dagpart:-C:-c:2:-n:92:-o:3:-t:4000,94,251,91,247,58,35,35,35,
iter-dagpart:-C:-c:2:-n:93:-o:3:-t:4000,95,257,95,257,50,23,23,23,
iter-dagpart:-C:-c:2:-n:94:-o:3:-t:4000,96,252,92,238,35,35,35,32,
iter-dagpart:-C:-c:2:-n:96:-o:3:-t:4000,98,259,97,258,66,30,30,30,
iter-dagpart:-C:-c:2:-n:97:-o:3:-t:4000,99,262,99,262,24,24,24,24,
iter-dagpart:-C:-c:2:-n:99:-o:3:-t:4000,101,274,101,274,60,26,26,26,
};

    \end{axis}
  \end{tikzpicture}

  \caption{Comparison of \partSet{}s found by \citet{LBK09}'s heuristic
    with an embedded \partSet{} of size~$k$ (left) and an
    optimal \partSet{} (right).  All graphs on the left side have
    $10^6$~vertices and roughly $18\cdot 10^6$~arcs and were generated
    by adding $k$~random arcs between ten connected components, each
    being a preferential attachment graph on $10^5$~vertices with
    outdegree fifteen and a single sink. All graphs on the right side
    are preferential attachment graphs with varying number of vertices,
    two sinks and outdegree three.}
  \label{fig:prefat}
\end{figure}

\paragraph{Summary} We have seen that solving large instances
with \partSet{}s of small weight is realistic using our algorithm. In
particular, instances with more than $10^7$~arcs and $k\leq 190$ could
be solved in less than five minutes. A crucial ingredient in this
success is the data reduction executed by \autoref{reducealg}; without
its help, we could not solve any of our instances in less than five
minutes.

\looseness=-1 However, we also observed that our algorithm works best on those
instances that can already be solved mostly optimally by \citet{LBK09}'s
heuristic and that the data reduction executed by \autoref{reducealg}
slows down the heuristic.

Having seen that the heuristic by \citet{LBK09} can be off by more than a
factor of two from the optimum on random preferential attachment graphs
diminishes the hope that, in spite of the non-approximability results of
\DAGP{} by \citet{AM12}, the heuristic of \citet{LBK09} might find good
approximations on naturally occurring instances. As we see, we do not
have to construct adversarial instances to make the heuristic find
solutions far from optimal.

\subsection{Limits of data reduction and fixed-parameter
  algorithms}\label{sec:kkern}
\noindent In \autoref{sec:klintime}, we have seen linear-time data
reduction rules for \DAGP{}.  Moreover, the experiments in
\autoref{sec:experiments} have shown that, on all input instancse we
tested our algorithm on, the running time of our
$\bigO(2^k\cdot (n+m))$~time \autoref{alg:simple-st} merely depended
on~\(k\) because \autoref{reducealg} shrunk our random input instances
to roughly the same size.

Therefore it is natural to ask whether we can provide data reduction
rules that \emph{provably} shrink the size of each possible input
instance to some fixed polynomial in~$k$, that is, whether there is a
polynomial-size problem kernel for \DAGP{}.  Unfortunately, in this section, we give a negative answer to this question.  Specifically, %
we prove that \DAGP{} does not
admit problem kernels with size polynomial in~$k$, unless
\NoPolyKernelAssume. Moreover, we show that the running time
$\bigO(2^k\cdot (n+m))$ of \autoref{alg:simple-st} cannot be improved to
$2^{\bigo(k)}\mathrm{poly}(n)$, unless the Exponential Time Hypothesis
fails. Herein, the Exponential Time Hypothesis as well as
\NegNoPolyKernelAssume{} are hypotheses stronger than~P${}\neq{}$NP, but
widely accepted among complexity theorists~\citep{IPZ01,FS11}.

Towards proving these results, we first recall the polynomial-time
many-to-one reduction from \textsc{3-Sat} to \DAGP{} given
by~\citet{AM12}. The \textsc{3-Sat} problem is, given a
formula~$\varphi$ in conjunctive normal form with at most three literals
per clause, to decide whether~$\varphi$ admits a satisfying assignment.

\begin{construction}[{\citet{AM12}}]\label{lem:red-from-3Sat-WDAGP}
  \upshape Let $\varphi$ be an instance of \textsc{3-Sat} with the
  variables $x_1, \ldots, x_n$ and the clauses $C_1, \dots, C_m$. We
  construct a \DAGP{} instance~$(G,\wei{},k)$ with $k:=4n+2m$ that is a
  yes-instance if and only if $\varphi$ is satisfiable. The weight
  function~$\wei{}$ will assign only two different weights to the arcs:
  a \emph{normal arc} has weight one and a \emph{heavy arc} has weight
  $k+1$ and thus cannot be contained in any \partSet{} of
  weight~$k$. The remainder of this construction is illustrated in
  \autoref{fig:SAT-to-DAGP}.  \begin{figure} \centering
    \begin{tikzpicture}[shorten >= 0.5mm]
      \tikzstyle{vertex}=[circle,draw,fill=black,minimum size=3pt,inner sep=0pt]

      \node[vertex, label=right:$f$] at (6,0) (f) {};
      \node[vertex, label=right:$f'$] at (6,-3) (f') {};
      \node[vertex, label=left:$t$] at (0,0) (t) {};
      \node[vertex, label=left:$t'$] at (0,-3) (t') {};

      \node[vertex, label=left:$x_1^t$] at (1,-1) (x1t) {};
      \node[vertex, label=left:$x_1$] at (1,-2) (x1) {};
      \node[vertex, label=right:$x_1^f$] at (2,-1) (x1f) {};
      \node[vertex, label=right:$\bar x_1$] at (2,-2) (bx1) {};

      \node[vertex, label=left:$x_2^t$] at (4,-1) (x2t) {};
      \node[vertex, label=left:$x_2$] at (4,-2) (x2) {};
      \node[vertex, label=right:$x_2^f$] at (5,-1) (x2f) {};
      \node[vertex, label=right:$\bar x_2$] at (5,-2) (bx2) {};

      \node[vertex, label=below:$C_1$] at (3,-1.5) (C) {};
      
      \draw[very thick, ->] (f) -> (f');
      \draw[very thick, ->] (t) -> (t');
      \draw[very thick, ->] (t) -> (x1t);
      \draw[very thick, ->] (t) to[in=90,out=0,tension=0.1] (x2t);
      \draw[very thick, ->] (f) to[in=90,out=180] (x1f);
      \draw[very thick, ->] (f) -> (x2f);

      \draw[->] (x1t)--(x1);
      \draw[->,dotted] (x1t)--(bx1);
      \draw[->,dotted] (x1f)--(x1);
      \draw[->] (x1f)--(bx1);
      \draw[->] (x1) to[out=240,in=30] (t');
      \draw[->,dotted] (x1) to[out=-30,in=190] (f');
      \draw[->,dotted] (bx1) to[out=230,in=0] (t');
      \draw[->] (bx1) to[out=-30,in=180] (f');

      \draw[->,dotted] (x2t)--(x2);
      \draw[->] (x2t)--(bx2);
      \draw[->] (x2f)--(x2);
      \draw[->,dotted] (x2f)--(bx2);
      \draw[->,dotted] (x2) to[out=210,in=0] (t');
      \draw[->] (x2) to[out=-50,in=180] (f');
      \draw[->] (bx2) to[out=210,in=-10] (t');
      \draw[->,dotted] (bx2) to[out=-60,in=150] (f');

      \draw[->] (C) to[out=180,in=30] (x1);
      \draw[->] (C) to[out=0,in=150] (bx2);
      \draw[very thick, ->] (t) to[in=90,out=-20] (C);

      \begin{pgfonlayer}{background}
        \draw[fill=black!10,line width=35pt,line join=round,
        draw=black!10] (x1t.center) rectangle (bx1.center);

        \draw[fill=black!10,line width=35pt,line join=round,
        draw=black!10] (x2t.center) rectangle (bx2.center);
      \end{pgfonlayer}
    \end{tikzpicture}
    \caption{\DAGP{} instance constructed from the formula consisting
      only of the clause $C_1:=(x_1\vee\bar x_2)$. Heavy arcs are drawn
      bold; dotted arcs are a \partSet{} that corresponds to the
      satisfying assignment setting~$x_1$ to true and $x_2$ to
      false. The variable gadgets are drawn on a gray background.}
    \label{fig:SAT-to-DAGP}
  \end{figure}
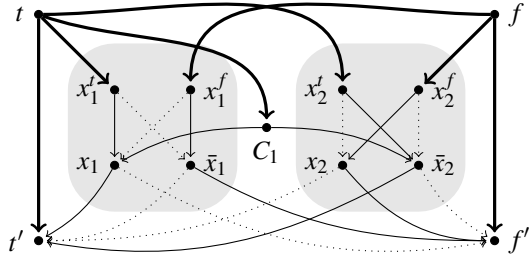

  We start constructing the directed acyclic graph~$G$ by adding the
  special vertices~$f,f',t,$ and~$t'$ together with the heavy
  arcs~$(f,f')$ and~$(t,t')$.  The vertices~$f'$ and~$t'$ will be the
  only sinks in~$G$.  For each variable $x_i$, introduce the
  vertices~$x^t_i, x_i^f, x_i$ and~$\bar x_i$ together with the
  heavy arcs~$(t,x_i^t)$ and $(f,x_i^f)$ and the normal arcs
  $(x_i^t,x_i)$, $(x_i^t,\bar x_i)$, $(x_i^f,x_i)$,
  $(x_i^f,\bar x_i)$, $(x_i,f')$, $(\bar x_i,f')$,
  $(x_i,t')$, and $(\bar x_i,t')$. For each clause~$C_j$, add a
  vertex~$C_j$ together with the heavy arc~$(t,C_j)$.  Finally, if some
  clause~$C_j$ contains the literal~$x_i$, then add the arc~$(C_j,x_i)$;
  if some clause~$C_j$ contains the literal~$\bar x_i$, then add the
  arc~$(C_j,\bar x_i)$.
\end{construction}

\noindent \citet{AM12} showed that, given a formula~$\varphi$ in \DREICNF{}
with $n$~variables and $m$~clauses, \autoref{lem:red-from-3Sat-WDAGP}
outputs a graph~$G$ with arc weights~$\wei{}$ such that $\varphi$~is
satisfiable if and only if there is a \partSet{}~$\Sol{}$ for~$G$ that
does not contain heavy arcs.

\looseness=-1 Since $G$~has only the two sinks~$t'$ and~$f'$, by \autoref{lem:no_new_sinks}, a minimal such \partSet{} has to partition~$G$ into two connected components, one connected component containing the heavy arc~$(t,t')$ and the other containing~$(f,f')$.  Moreover, if \(\varphi\)~is satisfiable, then such a \partSet{} has weight at most~$4n+2m$, since for each~$x_i$ of the $n$~variables of~$\varphi$, it deletes at most one of two arcs outgoing from each of the vertices~$x^t_i, x_i^f, x_i$, and~$\bar x_i$, and for each~$C_j$ of the $m$~clauses, it deletes at most two out of the three arcs outgoing from the clause vertex~$C_j$. We thus obtain the following lemma:

\begin{lemma}\label{lem:constr1}
  Given a formula~$\varphi$ in \DREICNF{} with $n$~variables and $m$~clauses,
  \autoref{lem:red-from-3Sat-WDAGP} outputs a graph~$G$ with arc
  weights~$\wei{}$ such that $\varphi$~is satisfiable if and only if
  there is a \partSet{}~$\Sol{}$ for~$G$ that does not contain heavy
  arcs. %

  Moreover, if \(\varphi\)~is satisfiable, then $\Sol{}$ has weight at most~$4n+2m$ and partitions~$G$ into
  one connected component containing the constructed vertices~$t$
  and~$t'$ and the other containing~$f$ and~$f'$.
\end{lemma}

\noindent We will later exploit \autoref{lem:constr1} to show our hardness results.
Next, we show that arcs with non-unit weights in our constructions
can be simulated by arcs with unit weights.
This allows us to show stronger hardness results and
to keep our constructions simple.

\subsubsection{Strengthening of hardness results to unit-weight graphs}\label{sec:tounit}
\autoref{lem:red-from-3Sat-WDAGP} heavily relied on
forbidding the deletion of certain arcs by giving them a high weight.
The next lemma shows that we can replace these arcs by a gadget only
using unit-weight arcs without changing the weight of the \partSet{} sought.

\begin{lemma}
  \label{lem:red-WDAGP-to-DAGP} There is a \kred{} from \DAGP{} with
  polynomially bounded weights to unweighted \DAGP{} that does not change
  the weight~$k$ of the \partSet{} sought.
\end{lemma}

\begin{figure}[t]
  \centering
  \begin{tikzpicture}[shorten >= 0.5mm,baseline=(v)]
    \tikzstyle{vertex}=[circle,draw,fill=black,minimum size=3pt,inner sep=0pt]

    \node[vertex, label=left:$v$] (v) at (0,0) {};
    \node[vertex, label=right:$w$] (w) at (2,0) {};

    \draw[->] (v) -- node[midway,above] {5} (w);
  \end{tikzpicture}\hspace{2cm}
  \begin{tikzpicture}[shorten >= 0.5mm,baseline=(v)]
    \tikzstyle{vertex}=[circle,draw,fill=black,minimum size=3pt,inner sep=0pt]

    \node[vertex, label=left:$v$] (v) at (0,0) {};
    \node[vertex, label=right:$w$] (w) at (2,0) {};
    \node[vertex] (a1) at (1,1) {};
    \node[vertex] (a2) at (1,0.5) {};
    \node[vertex] (a3) at (1,-0.5) {};
    \node[vertex] (a4) at (1,-1) {};

    \draw[->] (v) -- (w);
    \draw[->] (v) -- (a1);
    \draw[->] (v) -- (a2);
    \draw[->] (v) -- (a3);
    \draw[->] (v) -- (a4);

    \draw[->] (a1) -- (w);
    \draw[->] (a2) -- (w);
    \draw[->] (a3) -- (w);
    \draw[->] (a4) -- (w);
  \end{tikzpicture}

  \caption{Replacing an arc of weight~5 on the left by the gadget of unit-weight arcs on the right.}
  \label{fig:unweighted}
\end{figure}
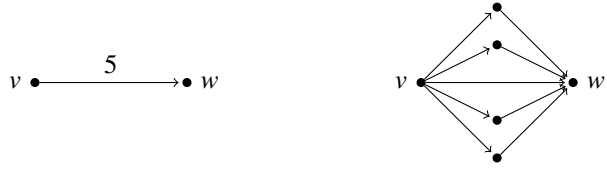

\begin{proof}
  \newcommand{\Va}{X} Let $(G,\wei{},k)$ be an instance of \DAGP. We
  show how to obtain an instance~$(G',\wei{}',k)$ by replacing a single
  arc of weight more than one by arcs of weight one such
  that~$(G,\wei{},k)$ is a yes-instance if and only if
  $(G',\wei{}',k)$~is a yes-instance. The replacement will be done as
  illustrated in \autoref{fig:unweighted}. The claim then follows by
  repeating this procedure for every arc of weight more than one.

  Consider an arc~$a=(v,w)$ in~$G$ with $\wei{}(a)>1$. We obtain~$G'$
  and~$\wei{}'$ from~$G$ and~$\wei{}$ by setting the
  weight~$\wei{}'(a)=1$, adding a set~$\Va{}$ of $\wei{}(a)-1$~vertices
  to~$G'$, and inserting for each $u\in \Va{}$ a weight-one arc~$(v,u)$
  and a weight-one arc~$(u,w)$.

  First, assume that $(G,\wei{},k)$ is a yes-instance and that
  $\Sol{}$~is a minimal \partSet{} for~$G$.  We show how to obtain
  a \partSet{} of weight~$k$ for~$G'$. Clearly, if $a\notin \Sol{}$,
  then~$\Sol{}$ is a \partSet{} of equal weight for~$(G',\wei{}',k)$. If
  $a\in \Sol{}$, then we get a \partSet{} of equal weight for
  $(G',\wei{}',k)$ by adding the arcs between~$v$ and~$\Va{}$
  to~$\Sol{}$.

  Second, assume that $(G',\wei{}',k)$ is a yes-instance and
  that~$\Sol{}$ is a minimal \partSet{} for~$G'$.  We show how to obtain
  a \partSet{} of weight~$k$ for~$G$. To this end, we consider two
  cases: $v$~and~$w$ are in a common or in separate connected components
  of~$G'\setminus\Sol{}$.

  \begin{caselist}
  \item If $v$~and~$w$ are in one connected component
    of~$G'\setminus\Sol{}$, then, by minimality, $\Sol{}$~does not
    contain~$a$ or any arc incident to vertices in~$\Va{}$. Hence,
    $\Sol{}$ is a \partSet{} of equal weight for~$G$.
  \item If $v$~and~$w$ are in separate connected components
    of~$G'\setminus\Sol{}$, then $a\in\Sol{}$. Moreover, the vertices
    in~$\Va{}$ have only one outgoing arc. Hence, by
    \autoref{lem:no_new_sinks}, $\Sol{}$~does not contain arcs
    from~$\Va{}$ to~$w$ but, therefore, contains all arcs from~$v$
    to~$\Va{}$.  Removing these arcs from~$\Sol{}$ results in
    a \partSet{} of equal weight for~$(G,\wei{},k)$.\qed
  \end{caselist}
\end{proof}

\subsubsection{Limits of fixed-parameter algorithms}
\label{sec:limits}
We now show that \DAGP{} cannot be solved in $2^{\bigo(k)}\textrm{poly}(n)$~time
unless the Exponential Time Hypothesis fails.
Thus, if our search tree algorithm for \DAGP{} can be improved, then only by
replacing the base of the exponential $2^k$-term by some smaller constant.

The Exponential Time Hypothesis was introduced by \citet{IP01} and
states that $n$-variable \textsc{3-Sat} cannot be solved
in~$2^{\bigo(n)}\poly(n)$ time. Using the reduction from \textsc{3-Sat} to
\DAGP{} given by \citet{AM12} (\autoref{lem:red-from-3Sat-WDAGP}), we
can easily show the following:

\begin{theorem}\label{cor:no-sub-exp-algorithm}
  Unless the Exponential Time Hypothesis fails, \DAGP{} cannot be solved
  in $2^{\bigo(k)}\poly(n)$ time even if all arcs have unit weight.
\end{theorem}

\begin{proof}
  \autoref{lem:red-from-3Sat-WDAGP} reduces an instance of
  \textsc{3-Sat} consisting of a formula with $n$~variables and
  $m$~clauses to an equivalent instance~$(G,\wei{},k)$ of \DAGP{} with
  $k=4n+ 2m$.  Thus, a $2^{\bigo(k)}\poly(n)$-time algorithm for \DAGP{}
  would yield a $2^{\bigo(m)}\poly(m)$-time algorithm for
  \textsc{3-Sat}. This, in turn, by the so-called \emph{Sparsification
    Lemma} of \citet[Corollary~2]{IPZ01}, would imply a
  $2^{\bigo(n)}\poly(n)$~time algorithm for \textsc{3-Sat}, which
  contradicts the Exponential Time Hypothesis.
  Since the weights used in \autoref{lem:red-from-3Sat-WDAGP} are
  polynomial in the number of created vertices and edges,
  we can apply \autoref{lem:red-WDAGP-to-DAGP} to transfer
  the result to the unit-weight case. \qed
\end{proof}

\subsubsection{Limits of problem kernelization} We now show that \DAGP{}
has no polynomial-size problem kernel with respect to the
parameter~$k$---the weight of the \partSet{} sought.  It follows that,
despite the effectiveness of data reduction observed in experiments in
\autoref{sec:experiments}, we presumably cannot generally shrink a \DAGP{}
instance in polynomial time to a size polynomial in~$k$.

To show that \DAGP{} does not allow for polynomial-size kernels, we
first provide the necessary concepts and techniques introduced by
\citet{BJK13}.
\begin{definition}[{\citet[Definition~3.3]{BJK13}}]
  For some finite alphabet~\(\Sigma\), a problem~$L\subseteq\Sigma^*$ \emph{(OR-)cross-composes} into a parameterized problem~$Q\subseteq\Sigma^*\times\mathbb N$ if there is an algorithm (a \emph{(OR-)cross-composition}) that transforms instances~$x_1,\ldots,x_s$ of~$L$ into an instance~$(x^*,k)$ for~$Q$ in time polynomial in~\(\sum_{i=1}^s|x_i|\) such that
  \begin{enumerate}[i)]
  \item $k$~is bounded by a polynomial in~$\max^s_{i=1}|x_i|+\log s$ and
  \item $(x^*,k)\in Q$ if and only if there is an~$i\in\{1,\dots,s\}$
    such that~$x_i\in L$.
  \end{enumerate}
  Furthermore, the cross-composition may exploit that the input
  instances~$x_1,\ldots,x_s$ belong to the same equivalence class of a
  \emph{polynomial equivalence relation}~$R\subseteq \Sigma^*\times
  \Sigma^*$, which is an equivalence relation such that
  \begin{enumerate}[i)]
  \item  it can be
    decided in polynomial time whether two instances are equivalent and
  \item any finite set~$S\subseteq \Sigma^*$ is partitioned into
    $\poly(\max_{x\in S}|x|)$ equivalence classes.
  \end{enumerate}
\end{definition}

\noindent The assumption that all instances belong to the same
equivalence class of a polynomial equivalence relation can make the
construction of a cross-composition remarkably easier: when giving a
cross-composition from \textsc{3-Sat}, we can, for example, assume that
the input instances all have the same number of clauses and variables.

Cross-compositions can be used to prove that a parameterized problem has
no polynomial-size kernel unless \NoPolyKernelAssume{}.

\begin{theorem}[{\citet[Corollary~3.6]{BJK13}}] \label{thm:Bod-No-Poly-Kernel} If some
  problem~$L\subseteq \Sigma^*$ is NP-hard under \kred{}s and
  $L$~cross-composes into the parameterized problem~$Q\subseteq
  \Sigma^*\times \mathbb{N}$, then there is no polynomial-size problem
  kernel for~$Q$ unless \NoPolyKernelAssume.
\end{theorem}
 
\noindent
In the following, we show that \textsc{3-Sat} cross-composes into \DAGP{} parameterized by~$k$, which yields the following theorem:

\begin{theorem}\label{thm:no-poly}
  Unless \NoPolyKernelAssume, \DAGP{} does not have a polynomial-size problem kernel with respect to
  the weight~$k$ of the \partSet{} sought even if all arcs have unit weight.
\end{theorem}

\noindent Although the proof of \autoref{thm:no-poly} is based on the
following construction, which requires arc weights, by using \autoref{lem:red-WDAGP-to-DAGP}
from \autoref{sec:tounit}, we obtain that \autoref{thm:no-poly} holds even on graphs
with unit weights.

\begin{construction}\label{lem:3-sat-cross-comp-to-w-DAGP}\upshape
  \begin{figure}
    \centering
    \begin{tikzpicture}[x=2cm,y=1.25cm,shorten >= 0.5mm]
      \tikzstyle{vertex}=[circle,draw,fill=black,minimum size=3pt,inner sep=0pt]

      \node[vertex,label=left:$t$] (t) at (-0.5,1.5) {};
      \node[vertex,label=below:$t'$] (t') at (1.5,-2.5) {};
      \node[vertex,label=above:$t''$] (t'') at (1.5,1.5) {};
      \node[vertex] (i1) at (1,1) {};
      \node[vertex] (i2) at (2,1) {};

      \node at (2.1,1.6) {$O$};
      \node at (2.1,-2.6) {$I$};

      \node[vertex] (o1) at (1,-2) {};
      \node[vertex] (o2) at (2,-2) {};

      \node[vertex,label=above:$f$] (f) at (1.5,-0.5) {};
      \node[vertex,label=below:$f'$] (f') at (1.5,-1.5) {};

      \node[vertex,label=left:$t_1$] (t1) at (0,0) {};
      \node[vertex,label=left:$t_1'$] (t1') at (0,-1) {};

      \node[vertex,label=left:$t_2$] (t2) at (1,0) {};
      \node[vertex,label=left:$t_2'$] (t2') at (1,-1) {};

      \node[vertex,label=right:$t_3$] (t3) at (2,0) {};
      \node[vertex,label=right:$t_3'$] (t3') at (2,-1) {};

      \node[vertex,label=right:$t_4$] (t4) at (3,0) {};
      \node[vertex,label=right:$t_4'$] (t4') at (3,-1) {};

      \draw[very thick,->] (t) -- (t'');
      \draw[very thick,->] (t) -- (-0.5,-1) to[out=-90,in=180] (t');
      \draw[very thick,->] (f) -- (f');
      \draw[->] (t'') to (i1);
      \draw[->,dotted] (t'') to (i2);
      \draw[->] (i1) to (t1);
      \draw[->,dotted] (i1) to (t2);
      \draw[->] (i2) to (t3);
      \draw[->] (i2) to (t4);
      \draw[->] (t1') to (o1);
      \draw[->,dotted] (t2') to (o1);
      \draw[->] (t3') to (o2);
      \draw[->] (t4') to (o2);
      \draw[->] (o1) to (t');
      \draw[->,dotted ] (o2) to (t');
      \draw[->,dotted] (o1) to (f');
      \draw[->] (o2) to (f');
      \draw[->,dotted] (t1') to[out=-10,in=180] (f');
      \draw[->] (t2') to (f');
      \draw[->] (t3') to (f');
      \draw[->] (t4') to[out=-170,in=0] (f');

      \draw[->,very thick] (t1) to (t1');
      \draw[->,very thick] (t2) to (t2');
      \draw[->,very thick] (t3) to (t3');
      \draw[->,very thick] (t4) to (t4');

      \begin{pgfonlayer}{background}
        \tikzstyle{treedge}=[color=black!10, line width=25pt, line cap=round]
        \draw[treedge] (t''.center)--(i1.center);
        \draw[treedge] (i1.center)--(t2.center);
        \draw[treedge] (t''.center)--(i2.center);
        \draw[treedge] (i2.center)--(t4.center);
        \draw[treedge] (i2.center)--(t3.center);
        \draw[treedge] (i1.center)--(t1.center);
        \draw[treedge] (t1.center)--(t2.center);
        \draw[treedge] (t3.center)--(t4.center);

        \draw[treedge] (t'.center)--(o1.center);
        \draw[treedge] (o1.center)--(t2'.center);
        \draw[treedge] (t'.center)--(o2.center);
        \draw[treedge] (o2.center)--(t4'.center);
        \draw[treedge] (o2.center)--(t3'.center);
        \draw[treedge] (o1.center)--(t1'.center);
        \draw[treedge] (t1'.center)--(t2'.center);
        \draw[treedge] (t3'.center)--(t4'.center);
      \end{pgfonlayer}
    \end{tikzpicture}

    \caption{Cross composition of four formulas $\phi_1,\dots,\phi_4$
      into a \DAGP{} instance. Of each subgraph~$G_i$ corresponding to
      formula~$\phi_i$, only the vertices~$t_i$, $t_i'$, and their
      connecting heavy arc are shown. The introduced binary trees~$O$
      and~$I$ are highlighted in gray. Deleting the $3\log 4=6$~dotted
      arcs requires the graph~$G_1$ to be partitioned, since its
      vertices~$t_1$~and~$t'_1$ are in a different connected component
      than its vertices~$f_1=f$~and~$f_1'=f'$. The graphs~$G_i$ for
      $i>1$ do not have to be partitioned; they are completely
      contained in one connected component with~$f$ and~$f'$.}
    \label{fig:DAGP-crossco}
  \end{figure}
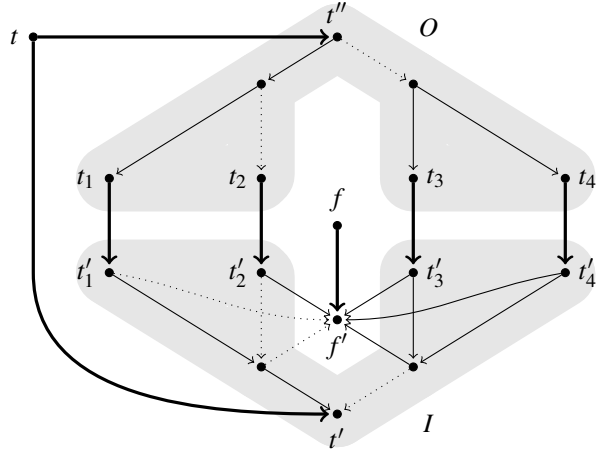
  Let $\varphi_1,\ldots, \varphi_s$ be instances of \textsc{3-Sat}.
  Since we may assume $\varphi_1,\dots,\varphi_s$ to be from the same
  equivalence class of a polynomial equivalence relation, we may assume
  that each of the formulas~$\varphi_1,\dots,\varphi_s$ has the same
  number~$n$ of variables and the same number~$m$ of clauses.  Moreover,
  we may assume that~$s$ is a power of two; otherwise we simply add
  unsatisfiable formulas to the list of instances.  We now construct a
  \DAGP{} instance~$(G,\wei{},k)$ with $k:=4n+2m+3 \log s$ that is a
  yes-instance if and only if $\varphi_i$~is satisfiable for at least
  one~$1\leq i\leq s$, where we use ``\(\log\)'' to denote the binary
  logarithm.  As in \autoref{lem:red-from-3Sat-WDAGP}, the weight
  function~$\wei{}$ will only assign two possible weight values: a
  \emph{heavy arc} has weight $k+1$ and thus cannot be contained in
  any \partSet{}. A \emph{normal arc} has weight one. The remainder of
  the construction is illustrated in \autoref{fig:DAGP-crossco}.

  For each instance~$\varphi_i$, let~$G_i$ be the graph produced by
  \autoref{lem:red-from-3Sat-WDAGP}. By \autoref{lem:constr1}, $G_i$~can
  be partitioned with $4n+2m$~arc deletions if and only if~$\varphi_i$
  is a yes-instance. We now build a gadget that, by means of additional
  $3\log s$~arc deletions, chooses exactly one graph~$G_i$ that has to
  be partitioned.

  \looseness=-1To distinguish between multiple instances, we denote the
  special vertices~$f,f',t,$ and~$t'$ in~$G_i$ by~$f_i,f'_i,t_i,$
  and~$t'_i$. For all $1\leq i\leq s$, we add~$G_i$ to the output
  graph~$G$ and merge the vertices~$f_1,f_2,\ldots,f_s$ into a
  vertex~$f$ and the vertices~$f'_1,f'_2,\ldots,f'_s$ into a
  vertex~$f'$.  Furthermore, we add the vertices~$t,t',$ and~$t''$ and
  the heavy arcs~$(t,t')$ and~$(t,t'')$ to~$G$.

  We add a balanced binary tree~$O$ rooted in~$t''$ and its leaves being
  the vertices~$t_1,\ldots,t_s$ that is formed by normal arcs directed
  from the root to the leaves. That is, $O$~is an \emph{out-tree}. %
  Moreover, add a balanced binary tree~$I$ rooted in~$t'$ and its leaves
  being the vertices~$t'_1,\ldots,t'_s$ that is formed by normal arcs
  directed from the leaves to the root. That is, $I$~is an
  \emph{in-tree}.  For each vertex~$v\ne t'$ in~$I$, add a normal
  arc~$(v,f')$.
\end{construction}

\noindent Using this construction, we can now prove
\autoref{thm:no-poly}.

\begin{proof}[of \autoref{thm:no-poly}]
  We only have to show that the instance~$(G,\wei,k)$ constructed by
  \autoref{lem:3-sat-cross-comp-to-w-DAGP} is a yes-instance if and only
  if at least one of the input formulas~$\varphi_i$ is
  satisfiable. Then, the theorem for the weighted case follows from
  \autoref{thm:Bod-No-Poly-Kernel}.  Since the weights used in
  \autoref{lem:3-sat-cross-comp-to-w-DAGP} are polynomial in the number
  of created vertices and edges, we can apply
  \autoref{lem:red-WDAGP-to-DAGP} to transfer the result to the
  unit-weight case.

  First, assume that a formula~$\varphi_i$ is satisfiable for
  some~$1\leq i\leq s$. By \autoref{lem:constr1} it follows that
  $G_i$ can be partitioned by $k':=4n+2m$~arc deletions into two
  connected components~$P_t$ and~$P_f$ such that $P_t$~contains~$t_i$
  and~$t'_i$ and such that $P_f$~contains~$f_i=f$ and~$f_i'=f'$. We
  apply these arc deletions to~$G$ and delete $3\log s$ additional arcs
  from~$G$ as follows.  Let $L$ be the unique \dpath{} in~$O$ from $t''$
  to $t_i$. Analogously, let~$L'$ be the unique \dpath{} in~$I$
  from~$t'_i$ to~$t'$.  Observe that each of these \dpath{}s has $\log
  s$ arcs.  We partition~$G$ into the connected component~$P'_t = P_t
  \cup \{t,t',t''\} \cup V(L) \cup V(L')$ with sink~$t'$ and into the
  connected component~$P'_f= V(G) \setminus P'_t$ with sink~$f'$. To
  this end, for each vertex~$v\in V(L) \setminus \{t_i\}$,
  we remove the outgoing arc that does not belong to~$L$. Hence, exactly
  $\log s$~arcs incident to vertices of~$O$ are removed. Similarly, for
  each vertex~$v\in V(L') \setminus \{t_i'\}$, we remove the
  incoming arc not belonging to~$L'$. For each vertex~$v\neq t'$ of~$L'$,
  we remove the arc to~$f'$. Hence, exactly $2\log s$~arcs incident to
  vertices of~$I$ are removed. Thus, in total, at most $k= 4n +2m + 3\log
  s$ normal arcs are removed to partition~$G$ into~$P'_t$ and~$P'_f$.

  Conversely, let $\Sol{}$ be a minimal \partSet{} for~$G$ with
  $\wei{}(\Sol{})\leq k$. Then, by \autoref{lem:no_new_sinks},
  $G\setminus\Sol{}$ has two connected components, namely $P'_t$~with
  sink~$t'$ and $P'_f$~with sink~$f'$. Since $\Sol{}$ cannot contain
  heavy arcs, $t$~and~$t''$~are in~$P'_t$. Hence, $t''$~can reach~$t'$
  in~$G\setminus\Sol{}$, since they are in the same component of~$G \setminus \Sol{}$. 
  As every \dpath{} from~$t''$ to~$t'$ goes
  through some vertices~$t_i$ and~$t'_i$, it follows that there is
  an~$i\in \{1,\ldots s\}$ such that~$t_i$ and~$t'_i$ are
  in~$P'_t$. Since $f=f_i$ and~$f'=f'_i$ are in~$P'_f$,
  the \partSet{}~$\Sol{} \cap A(G_i)$ partitions~$G_i$ into two
  connected components: one containing~$t_i$ and~$t'_i$ and the other
  containing~$f=f_i$ and~$f'=f_i'$. Since~$\Sol{}$ does not contain
  heavy arcs, from \autoref{lem:constr1} it follows that $\varphi_i$ is
  satisfiable.  \qed
\end{proof}

\section{Parameter treewidth}\label{sec:tw}

In \autoref{sec:smallsol}, we have seen that \DAGP{} is linear-time solvable when the weight~$k$ of the requested
\partSet{} is constant. %
\citet{AM12} asked whether \DAGP{} is fixed-parameter tractable with
respect to the parameter treewidth, which is a measure of the
``tree-likeness'' of a graph. We will answer this question
affirmatively.

In \autoref{sec:tree}, we first show that, if the input graph is indeed
a tree, then the heuristic by \citet{LBK09} solves the instance
optimally in linear time.  Afterwards, in \autoref{sec:twalg}, we prove
that this result can be generalized to graphs of bounded treewidth and
thus improve the algorithm for pathwidth given by \citet{AM12}, since the
treewidth of a graph is at most its pathwidth.

\subsection{Partitioning trees}\label{sec:tree}

In this section, we show that the heuristic by \citet{LBK09} solves
\DAGP{} in linear time on trees. This result will be generalized in the
next section, where we show how to solve \DAGP{} in linear time on
graphs of bounded treewidth.  The heuristic by \citet{LBK09} is similar
to our search tree algorithm presented in \autoref{alg:simple-st}:
instead of trying all possibilities of associating a vertex with one of
its feasible sinks, it associates each vertex with the sink that it
would be most expensive to disconnect from.  The algorithm is presented
in \autoref{alg:simple-trees}.

\begin{algorithm}[t]
  \caption{\citet{LBK09}'s heuristic to compute small \partSet{}s.}
  \label{alg:simple-trees}
  \KwIn{A directed acyclic graph~$G=(V,A)$ with arc weights~$\wei{}$.}
  \KwOut{A \partSet{}~$S$.}
  $(v_1,v_2,\dots,v_n)\gets{}$reverse topological order of the vertices
  of~$G$\;
  $L\gets{}$array with $n$~entries\; %
  $S\gets\emptyset$\;
  \For{$i=1$ to~$n$}{
    \lIf(\tcp*[f]{\textrm{associate~$v_i$ with itself}}\nllabel{treeL}){$v_i$ is a sink}{$L[v_i]\gets v_i$}
    \Else{
       $D\gets\{L[w]\mid w\in N^+(v_i)\}$\tcp*{\textrm{the set of feasible sinks for~$v_i$}}
       \For{$s\in D$}{$A_s\gets\{(v_i,w)\in A\mid w\in N^+(v_i)\wedge L[w]=s\}$\nllabel{lin:Ss}\tcp*{\textrm{set of arcs to keep if $v_i$~is associated with sink~$s$}}}
       $s^*\gets\arg\max_{s\in D}\wei(A_s)$\nllabel{lin:superstar}\tcp*{\textrm{cheapest sink $s^*$ to associate $v_i$ with}}
       $L[v_i]\gets s^*$\;
       $S\gets S\cup \{(v_i,w)\mid w\in N^+(v_i)\wedge L[w]\ne s^*\}$\;
    }
  }
  \Return{S}
\end{algorithm}

\begin{theorem}\label{thm:treelin}
  \autoref{alg:simple-trees} solves \DAGP{} optimally in linear time if
  the underlying undirected graph is a tree.
\end{theorem}

\begin{proof}
  \autoref{alg:simple-trees} clearly works in linear time: to implement
  it, we only have to iterate over the out-neighbors of each
  vertex~$v_i$ once.  In particular, all sets~$A_s$ in \autoref{lin:Ss}
  can be computed by one iteration over each~$w\in N^+(v_i)$ and adding
  the arc~$(v_i,w)$ to~$A_{L[w]}$.  Moreover, \autoref{alg:simple-trees}
  returns a \partSet{}: each vertex~$v$ is associated with exactly one
  sink~$L[v]$ of~$G$ and the returned set~$S$ deletes exactly the
  arcs~$(v,w)$ for which $L[v]\ne L[w]$.

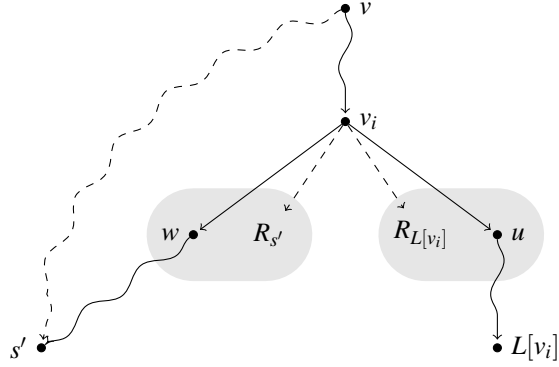
\begin{figure}
  \centering
    \begin{tikzpicture}[shorten >= 0.5mm,baseline=(v),y=0.75cm]
    \tikzstyle{vertex}=[circle,draw,fill=black,minimum size=3pt,inner sep=0pt]

    \node[vertex, label=right:$v_i$] (vi) at (0,0) {};
    \node[vertex, label=right:$v$] (v) at (0,2) {};
    \node[vertex, label=left:$w$] (w) at (-2,-2) {};
    \node (w') at (-1,-2) {$R_{s'}$};
    \node (u') at (1,-2) {$R_{L[v_i]}$};
    \node[vertex, label=left:$s'$] (s') at (-4,-4) {};
    \node[vertex, label=right:$u$] (u) at (2,-2) {};
    \node[vertex, label=right:${L[v_i]}$] (Lvi) at (2,-4) {};

    \draw[->] (vi)--(w);
    \draw[->,dashed] (vi)--(w');
    \draw[->,dashed] (vi)--(u');
    \draw[->] (vi)--(u);

    \draw[->, decorate, decoration = {snake, segment length = 0.9cm}, dashed] (v) to[bend right=30](s');
    \draw[->, decorate, decoration = {snake, segment length = 0.9cm}] (w)--(s');
    \draw[->, decorate, decoration = {snake, segment length = 0.9cm}] (u)--(Lvi);

    \draw[->, decorate, decoration = {snake, segment length = 0.9cm}] (v)--(vi);

      \begin{pgfonlayer}{background}
        \draw[fill=black!10,line width=35pt,line join=round, line cap=round,
        draw=black!10] (w.center) -- (w'.center);

        \draw[fill=black!10,line width=35pt,line join=round, line cap=round,
        draw=black!10] (u.center) -- (u'.center);
      \end{pgfonlayer}

  \end{tikzpicture}
  \caption{Illustration of the proof of \autoref{thm:treelin}.  Straight lines represent arcs. Wavy lines represent directed paths. Dashed arcs and paths represent arcs and paths that cannot exist since the underlying undirected graph is a tree.}
  \label{fig:treelin}
\end{figure}

  We show by induction on~$i$ that \autoref{alg:simple-trees} computes a
  \emph{minimum-weight} \partSet{} for~$G[\{v_1, \ldots, v_i\}]$.  For the induction base case,
  observe that $v_1$~is a sink and, thus, $v_1$ only reaches the
  sink $L[v_1]=v_1$ in~$G\setminus S$ for all possible
  minimum-weight \partSet{}s~$S$ for~$G$.  Now, assume that there is a
  minimum-weight \partSet{}~$S$ such that
  each~$v\in\{v_1,\dots,v_{i-1}\}$ only reaches the sink~$L[v]$
  in~$G\setminus S$.  We show that there is a
  minimum-weight \partSet{}~$S'$ such that
  each~$v\in\{v_1,\dots,v_{i}\}$ reaches only the sink~$L[v]$
  in~$G\setminus S'$.  If $v_i$~only reaches~$L[v_i]$ in~$G\setminus S$,
  then we are done.  Otherwise, $v_i$~reaches some sink~$s'\ne L[v_i]$
  in~$G\setminus S$ and, hence, is not itself a sink.  The graph~$G$ is partly illustrated in \autoref{fig:treelin}.

 \looseness=-1 Let $R_s$~be the set of vertices reachable from~$v_i$ that reach some sink~$s$ in~$G$. Since the underlying undirected graph of~$G$ is a tree, $v_i$~has exactly one arc into~$R_s$ for each sink~$s$ reachable from~$v_i$.  Let $(v_i,u)$~be the arc of~$v_i$ into~$R_{L[v_i]}$ and $(v_i,w)$~be the arc of~$v_i$ into~$R_{s'}$.  Observe that the arc $(v_i,u)$~exists since the algorithm can set $L[v_i]$ only to sinks reachable from~$v_i$.  To show that $S':=(S\setminus\{(v_i,w)\})\cup\{(v_i,u)\}$ is still a \partSet{}, we only have to verify that~$v_i$ and all vertices reaching~$v_i$ only reach one sink in~$G\setminus S'$.  For all other vertices, this follows from $S$~being a \partSet{}.
  \begin{enumerate}
  \item The vertex~$v_i$ only reaches the sink~$L[v_i]$ in~$G\setminus
    S'$: this is because $u=v_j$ for some $j<i$ which, by induction
    hypothesis, reaches exactly one sink in~$G\setminus S$ and, hence,
    in $G\setminus S'$.
  \item A vertex~$v$ that reaches~$v_i$ in~$G\setminus S'$ reaches only
    the sink~$L[v_i]$ in~$G\setminus S'$: otherwise $v$~reaches~$s'$
    in~$G\setminus S'$ since $v_i$ and, therefore $v$, reaches~$s'$
    in~$G\setminus S$.  This, however, means that $v$ has a path to $s'$ that bypasses~$v_i$ in $G\setminus S'$ and hence, in $G\setminus S$, which contradicts the undirected underlying graph of~$G\setminus S$ being a tree.
  \end{enumerate}
  It remains to show $\wei(S')\leq\wei(S)$, implying that $S'$ is also a
  minimum-weight \partSet{}.  To see this, we analyze the sets~$A_{s'}$
  and~$A_{L[v_i]}$ computed in \autoref{lin:Ss} of
  \autoref{alg:simple-trees}.  Observe that~$A_{s'}=\{(v_i,w)\}$ and
  that $A_{L[v_i]}=\{(v_i,u)\}$.  Since
  $\wei(A_{s'})\leq\wei(A_{L[v_i]})$ because of the choice of~$L[v_i]$
  in \autoref{lin:superstar}, we conclude that
  $\wei(v_i,u)\leq\wei(v_i,w)$ and hence, $\wei(S')\leq\wei(S)$.\qed
\end{proof}

\subsection{Partitioning DAGs of bounded treewidth}\label{sec:twalg}

\newcommand{\tab}{\textsl{Tab}}
\newcommand{\bagpart}{\mathcal{P}}
\newcommand{\bagreach}{\mathcal{G}}
\newcommand{\bagsol}{\mathcal{R}}
\newcommand{\q}{P}

We now give an algorithm that solves \DAGP{} in linear time on graphs of
bounded treewidth. In contrast to \autoref{sec:smallsol}, which
presented our search tree algorithm and an experimental evaluation
thereof, the algorithm presented below is of rather theoretical
interest: \citet{AM12} asked whether \DAGP{} is fixed-parameter
tractable with respect to the parameter treewidth. %
With a dynamic programming algorithm, we can prove the following
theorem, which answers their open question and is an improvement of
\citet{AM12}'s algorithm, since the treewidth of a graph is at most its
pathwidth.

\begin{theorem} \label{th:fpt-tw} Given a width-$\tw{}$ tree
  decomposition of the underlying undirected graph, \DAGP{} can be solved
  in $2^{\bigO(\tw^2)}\cdot n$ time.
\end{theorem}

\noindent We first formally define the tree decomposition of a graph and its width.

\begin{definition}[Treewidth, tree decomposition]\label{treedec}
  Let $G=(V,A)$~be a directed graph. A \emph{tree decomposition}~$(T,\beta)$ for~$G$ consists of a rooted tree~$T= (X,E)$ and
  a mapping~$\beta\colon X\to 2^V$ of each \emph{node}~$x$ of the
  tree~$T$ to a subset~$V_x:=\beta(x)\subseteq V$, called \emph{bag},
  such that
  \begin{enumerate}[i)]
  \item\label{treedec1} for each vertex~$v\in V$, there is a node~$x$ of~$T$
    with~$v\in V_x$,
  \item\label{treedec2} for each arc~$(u,w)\in A$, there is a node~$x$
    of~$T$ with $\{u,w\}\subseteq V_x$,
  \item\label{treedec3} for each vertex~$v\in V$, the nodes~$x$ of~$T$
    for which~$v\in V_x$ induce a subtree in~$T$.
  \end{enumerate}
  A tree decomposition is \emph{nice} if $V_r=\emptyset$ for the root~$r$ of~$T$ and each node~$x$ of~$T$ is either
  \begin{itemize}
  \item a \emph{leaf}: then, $V_x=\emptyset$,
  \item a \emph{forget node}: then, $x$ has exactly one child node~$y$
    and $V_x=V_y\setminus\{v\}$ for some~$v\in V_y$,
  \item an \emph{introduce node}: then, $x$ has exactly one child node~$y$ and
    $V_x=V_y\cup\{v\}$ for some~$v\in V\setminus V_y$, or
  \item a \emph{join node}: then, $x$ has exactly two child nodes~$y$ and~$z$
    such that~$V_x=V_y=V_z$.
  \end{itemize}
  The \emph{width of a tree decomposition} is one less
  than the size of its largest bag. The \emph{treewidth} of a graph~$G$
  is the minimum width of a tree decomposition for~$G$. For a node~$x$
  of $T$, we denote by~$U_x$ the union of~$V_y$ for all descendants~$y$
  of the node~$x$.
\end{definition}

\noindent\looseness=-1 For any constant~$\tw$, it can be decided in linear time whether a graph has treewidth~$\tw$ and the corresponding tree decomposition of width~$\tw$ can be constructed in linear time~\cite{Bod96}. Also in $\bigO(\tw n)$~time, the tree decomposition of width~$\tw$ can be transformed into a nice tree decomposition with the same width and $\bigO(\tw n)$~nodes~\citep{Klo94}. Hence, we assume without loss of generality that we are given a nice tree decomposition. 

\bigskip\noindent Our algorithm is based on leaf-to-root dynamic
programming. That is, intuitively, we start from the leaf nodes of the
tree decomposition and compute possible partial \partSet{}s for
each bag from the possible partial \partSet{}s for its child
bags.  Since our algorithm for \DAGP{} on graphs of bounded treewidth is relatively intricate, we refer the reader that is yet inexperienced with dynamic programming on tree decompositions to introductory chapters in corresponding text books \citep{Nie06,DF13,CFK+15,Klo94}.

We will now precisely define a partial \partSet{} and show that
any partial \partSet{} for the root bag is a \partSet{} for the entire
graph. The definition is illustrated in \autoref{fig:partsol}.

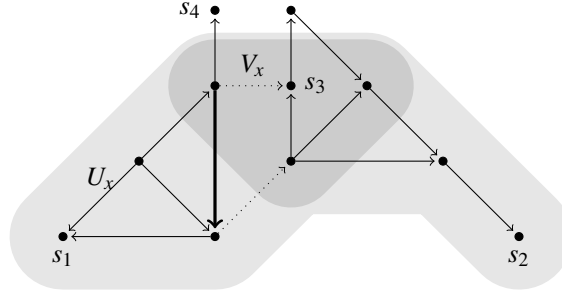
\begin{figure}
  \centering
  \begin{tikzpicture}[shorten >= 0.5mm]
    \tikzstyle{vertex}=[circle,draw,fill=black,minimum size=3pt,inner
    sep=0pt]

    \node[vertex] (v1) at (0,0) {};
    \node[vertex] (v2) at (4,0) {};
    \node[vertex,label=below:$s_1$] (v3) at (-1,-1) {};
    \node[vertex] (v4) at (1,1) {};
    \node[vertex,label=left:$s_4$] (v14) at (1,2) {};
    \node[vertex] (v5) at (3,1) {};
    \node[vertex,label=below:$s_2$] (v6) at (5,-1) {};
    \node[vertex] (v10) at (2,2) {};
    \node[vertex] (v11) at (1,-1) {};
    \node[vertex,label=right:$s_3$] (v12) at (2,1) {};
    \node[vertex] (v13) at (2,0) {};

    \draw[->] (v1)--(v11);
    \draw[->] (v11)--(v3);
    \draw[->, very thick] (v4)--(v11);
    \draw[->,dotted] (v11)--(v13);
    \draw[->] (v13)--(v12);
    \draw[->] (v13)--(v5);
    \draw[->] (v12)--(v10);
    \draw[->] (v1)--node[pos=0.5,label=above:$U_x$]{}(v3);
    \draw[->] (v1)--(v4);

    \draw[->] (v5)--(v2);
    \draw[->] (v2)--(v6);

    \draw[->] (v4)--(v14);
    \draw[->,dotted] (v4)-- node[pos=0.5,label=above:$V_x$]{} (v12);
    \draw[->] (v13)--(v2);
    \draw[->] (v10)--(v5);
     \begin{pgfonlayer}{background}
       \tikzstyle{edge} = [color=black!10,line cap=round, line join=round,
       line width=35pt]

       \draw[edge,line width=40pt,fill] (v6.center)--(v5.center)--(v4.center)--(v3.center)--(v11.center)--(v13.center)--(v2.center);

       \draw[edge,fill,color=black!20] (v4.center)--(v5.center)--(v13.center)--cycle;
     \end{pgfonlayer}
  \end{tikzpicture}
  \caption{The set~$\Sol{}$ consisting of the dotted arc is not a
    \partSet{} for the shown graph~$G$, but a partial \partSet{}
    for~$G[U_x]$. Note that $\Sol{}$~would not be a partial \partSet{} if
    it would additionally contain the bold arc, since then the
    sink~$s_1$ would be in a connected component with a vertex of~$V_x$
    but not reachable from~$V_x$. Also note that $\Sol{}$ does not
    separate~$s_2$ from~$s_3$, which are sinks
    in~$G[U_x]$ but not in~$G$.}
  \label{fig:partsol}
\end{figure}

\begin{definition}[Partial \partSet{}]\label{def:partsol}
  A \emph{partial \partSet{}}~$\Sol{}$ for~$G[U_x]$ is an arc
  set~$\Sol{}\subseteq A(G[U_x])$ such that
\begin{enumerate}[(i)]
\item\label{partsol1} no connected component of $G[U_x] \setminus
  \Sol{}$ contains two different sinks of $U_x \setminus V_x$, and
\item\label{partsol2} every sink in a connected component of
  $G[U_x]\setminus\Sol{}$ that contains a vertex of~$V_x$ can be reached
  from some vertex of~$V_x$ in $G[U_x]\setminus\Sol{}$.
\end{enumerate}
\end{definition}

\noindent Since we assumed to work on a tree decomposition with a
root~$r$ such that the bag~$V_r$ is empty, any partial \partSet{}
for~$G[U_r]=G$ will be a \partSet{} for the entire graph~$G$. Moreover,
\autoref{def:partsol}(\ref{partsol1}) does not require
partial \partSet{}s for~$G[U_x]$ to separate sinks in the
bag~$V_x$. This is because vertices in~$V_x$ that are sinks in~$G[U_x]$
might be non-sinks for a supergraph, as illustrated in
\autoref{fig:partsol}. Thus, it might be unnecessary to separate the
vertices in~$V_x$. However, due to
\autoref{treedec}(\ref{treedec2}~and~\ref{treedec3}) of a tree
decomposition, sinks in~$U_x\setminus V_x$ are sinks in all
supergraphs~$G[U_{q}]$ for $q$~being an ancestor node of~$x$.
\autoref{def:partsol}(\ref{partsol2}), by
\autoref{lem:dirundirequiv}, allows us to ensure that components containing both a
sink in~$U_x\setminus V_x$ and a vertex of~$V_x$ end up with
only one sink in some supergraph~$G[U_{q}]$.  The precise purpose of \autoref{def:partsol}(\ref{partsol2}) will be explained in more detail after the upcoming \autoref{patsat}.

To keep the notation free from clutter, note that \autoref{def:partsol}
implicitly relies on the bag~$V_x$ belonging to each set~$U_x$. Thus,
when a tree decomposition has a node~$x$ and a child node~$y$ such
that~$V_x\subsetneq V_y$ but~$U_x=U_y$, a partial \partSet{}~$\Sol{}$
for~$G[U_y]$ is not necessarily a partial \partSet{} for~$G[U_x]$,
although~$G[U_y]\setminus\Sol{}=G[U_x]\setminus\Sol{}$.

\bigskip\noindent Now, assume that we want to compute
partial \partSet{}s for~$G[U_x]$ from partial \partSet{}s for child
nodes of~$x$. These partial \partSet{}s might, for example, disagree on
which arcs to delete in the child bags or which connected components of
the child bags are meant to end up in a common connected component of
the entire graph: for a child node~$y$ of~$x$, multiple connected
components of~$G[U_y]\setminus\Sol{}$ might be one connected component
of~$G[U_x]\setminus\Sol{}$. To prevent such incompatibilities, we only
consider those partial \partSet{}s for~$G[U_x]$ that agree with
partial \partSet{}s for the child nodes of~$x$ on certain
\emph{patterns}.

On a high level, our algorithm will store for each node of the tree
decomposition a table with one row for each possible pattern.  The value
of a row will be the minimum weight of a partial \partSet{}
\emph{satisfying} this pattern.  To compute this value, our algorithm
will use the rows with \emph{corresponding} patterns in the tables of
the child nodes.  In the following, we first formalize the terms
\emph{patterns} and \emph{satisfying} partial \partSet{s}.  Then, we
present our algorithm and we specify the \emph{corresponding} patterns.
We start by formally defining patterns, see \autoref{fig:pattern} for an
illustration.

\begin{definition}[Pattern]\label{pat}
  Let $x$~be a node of a tree decomposition~$T$. A \emph{pattern
    for~$x$} is a triple~$(\bagsol{},\bagreach{},\bagpart)$ such that
  \begin{enumerate}[i)]
  \item\label{pat3} $\bagsol{}$ is a directed acyclic graph with the
    vertices~$V_x$.
  \item\label{pat1} $\bagreach{}$ is a directed acyclic graph with the
    vertices~$V_x$ and at most $|V_x|$~additional vertices such that
    each vertex in~$V(\bagreach{})\setminus V_x$ is a non-isolated sink, and
  \item\label{pat2} $\bagpart$ is a partition of the vertices
    of~$\bagreach{}$
    such that each connected
    component of $\bagreach{}$
    is within one set~$\q_i\in\bagpart$ and such that each $\q_i$
    contains at most one vertex of $V(\bagreach{})\setminus V_x$.
  \end{enumerate}
\end{definition}
We will use a pattern~$(\bagsol{},\bagreach{},\bagpart)$ for~$x$ to
capture important properties of partial \partSet{}s for~$G[U_x]$.
Intuitively, the graph~$\bagsol{}$ will describe which arcs between the
vertices in the bag~$V_x$ a partial \partSet{}~$\Sol{}$ for~$G[U_x]$
will not delete. The graph~$\bagreach{}$ will
describe %
which vertices of~$V_x$ can reach each other in~$G[U_x]\setminus\Sol{}$
and which sinks outside of~$V_x$ they can reach. The
partition~$\bagpart$ describes which vertices are meant to end up in the
same connected component of~$G\setminus\Sol{}$ for a \partSet{}~$\Sol{}$
of the entire graph. For this reason, the sets of the
partition~$\bagpart$ are allowed to contain only one vertex
of~$V(\bagreach{})\setminus V_x$ as these vertices are sinks.

We will now explain precisely what it means for a partial \partSet{} to
\emph{satisfy} a pattern. The following definition is illustrated in
\autoref{fig:pattern}.

\begin{figure}
  \centering
  \begin{tikzpicture}[shorten >= 0.5mm]
    \tikzstyle{vertex}=[circle,draw,fill=black,minimum size=3pt,inner
    sep=0pt]
    \node[vertex,label=below:$s_1$] (v3) at (-1,-1) {};
    \node[vertex] (v4) at (1,1) {};
    \node[vertex] (v5) at (3,1) {};
    \node[vertex,label=below:$s_2$] (v6) at (5,-1) {};
    \node[vertex,label=right:$s_3$] (v12) at (2,1) {};
    \node[vertex] (v13) at (2,0) {};

    \node at (2,2) {$V_x$};
    \node at (-1,0) {$P_1$};
    \node at (5,0) {$P_2$};

    \draw[->] (v13)--(v12);
    \draw[->] (v13)--(v5);
    \draw[->] (v4)--(v3);
    \draw[->] (v5)--(v6);
    \draw[->] (v13)--(v6);
     \begin{pgfonlayer}{background}
       \tikzstyle{edge} = [color=black!10,line cap=round, line join=round,
       line width=35pt]
       \draw[edge,fill,color=black!10] (v4.center)--(v5.center)--(v13.center)--cycle;

       \draw[edge,fill,color=black!20,line width=25pt] (v4.center)--(v3.center);

       \draw[edge,fill,color=black!20,line width=25pt] (v6.center)--(v5.center)--(v12.center)--(v13.center)--cycle;

     \end{pgfonlayer}
  \end{tikzpicture}
  \caption{A pattern~$(\bagsol,\bagreach,\bagpart)$, where shown are the graph~$\bagreach{}$
    and a partition~$\bagpart{}$ of its vertices into two sets~$\q_1$
    and~$\q_2$. If the graph~$\bagsol$ is the subgraph of~$\bagreach{}$ induced by the vertices~$V_x$, then, in terms of \autoref{patsat}, $\bagsol=\bagsol_x(\Sol{})$ and~$\bagreach{}=\bagreach_x(\Sol{})$ for the partial
    \partSet{}~$\Sol{}$ for~$G[U_x]$ shown in \autoref{fig:partsol}.  In
    this case, the partial \partSet{}~$\Sol{}$ satisfies the shown
    pattern. Moreover, a vertex is a sink in this figure if and only if
    it is a sink in~$G[U_x]\setminus\Sol{}$ shown in
    \autoref{fig:partsol}.}
  \label{fig:pattern}
\end{figure}

\newcommand{\interesting}{bag-reachable}

\begin{definition}[Pattern satisfaction]\label{patsat}
  Let $\Sol{}$ be a partial \partSet{} for~$G[U_x]$. A~sink~$s$ in $U_x
  \setminus V_x$ is \emph{\interesting{} in $G[U_x]\setminus\Sol{}$} if
  some vertex in~$V_x$ can reach~$s$ in~$G[U_x]\setminus\Sol{}$. We
  define a canonical
  pattern~$(\bagsol_x(\Sol{}),\bagreach_x(\Sol{}),\bagpart_x(\Sol{}))$
  at~$x$ for~$\Sol$, where
  \begin{itemize}
  \item\label{defbagsol} $\bagsol_x(\Sol{})$ is~$G[V_x]\setminus\Sol{}$,
  \item\label{defbagreach} $\bagreach{}_x(\Sol{})$ is the directed
    acyclic graph on the vertices~$V_x \cup V'$, where $V'$ is the set
    of \interesting{} sinks in~$G[U_x]\setminus\Sol{}$, and there is an
    arc~$(u,v)$ in $\bagreach{}_x(\Sol{})$ if and only if the vertex~$u$
    can reach the vertex~$v$ in~$G[U_x]\setminus\Sol{}$, and

  \item\label{defbagpart} $\bagpart_x(\Sol{})$ is the the partition of
    the vertices of~$\bagreach{}_x(\Sol{})$ such that the vertices $u$ and
    $v$ are in the same set of $\bagpart_x(\Sol{})$ if and only if they
    are in the same connected component of~$G[U_x]\setminus\Sol{}$.
  \end{itemize}

  \noindent Let $(\bagsol,\bagreach{},\bagpart)$ be a pattern
  for~$x$. We say that~$\Sol{}$ \emph{satisfies} the pattern
  $(\bagsol,\bagreach{},\bagpart)$ at~$x$ if
  \begin{enumerate}[i)]
  \item\label{patsatsol} $\bagsol=\bagsol_x(\Sol{})$,
  \item\label{patsat1} $\bagreach{}=\bagreach{}_x(\Sol{})$, and
  \item\label{patsat2} for each set~$\q \in \bagpart_x(\Sol{})$ there
    exists a set~$\q' \in \bagpart$ such that $\q \subseteq \q'$, that is,
    $\bagpart$ is a \emph{coarsening} of $\bagpart_x(\Sol{})$.
  \end{enumerate}
\end{definition}

\noindent It is easy to verify that a partial \partSet{}~$\Sol{}$
for~$G[U_x]\setminus\Sol{}$ satisfies its canonical
pattern~$(\bagsol_x(\Sol{}), \bagreach{}_x(\Sol{}),\bagpart_x(\Sol{}))$
at node~$x$: to this end, observe that
$(\bagsol_x(\Sol{}),\bagreach{}_x(\Sol{}),\bagpart_x(\Sol{}))$~is indeed
a pattern for~$x$: for the vertex set~$V_x \cup V'$
of~$\bagreach{}_x(\Sol{})$, we have $|V'|\leq |V_x|$ since each vertex
in~$V_x$ can reach at most one distinct sink in~$V'\subseteq
U_x\setminus V_x$ in~$G[U_x]\setminus\Sol{}$.

Note that, since $\bagreach{}_x(\Sol{})$~contains an arc $(u,v)$ if and
only if $u$~\emph{can reach}~$v$ instead of requiring them to be merely
connected, a vertex is a sink in~$\bagreach{}_x(\Sol{})$ if and only if
it is a sink in~$G[U_x]\setminus\Sol{}$.  Herein,
\autoref{def:partsol}(\ref{partsol2}) ensures that any sink~$s$
connected to a vertex in~$V_x$ is a vertex in~$\bagreach{}_x(\Sol{})$. %

While it might seem more natural to replace the condition (\ref{patsat2}) in \autoref{patsat} by simply $\bagpart=\bagpart_x(S)$, 
we prefer the current definition, because it allows for several connected components of~$G[U_x]\setminus\Sol{}$ becoming a part of one connected component of the entire graph. This greatly simplifies some parts of the algorithm.

\paragraph{The Algorithm} We now describe a dynamic programming
algorithm. Starting from the leaves of the tree decomposition~\(T\) and
working our way to its root, with each node~$x$ of~\(T\), we associate a table~$\tab_x$ that is indexed by all
possible patterns for~$x$. Semantically, we want that
\begin{align*}
  \tab_x(\bagsol,\bagreach{}, \bagpart) &= \text{minimum weight of a partial
    \partSet{} for~$G[U_x]$ that satisfies the pattern
    $(\bagsol,\bagreach{},\bagpart)$ at~$x$}.
\end{align*}
Since %
we have~$V_r=\emptyset$, there is exactly one
pattern~$(\bagsol,\bagreach{},\bagpart)$ for the root~$r$:
$\bagsol=\bagreach{}$~is the empty graph and~$\bagpart=\emptyset$. Thus,
$\tab_r$ has exactly one entry and it contains the minimum weight of a
partial \partSet{}~$\Sol{}$ for~$G[U_r]$, which is equivalent
to~$\Sol{}$ being a \partSet{} for~$G$. It follows that once the tables
are correctly filled, to decide the \DAGP{} instance~$(G,\wei{},k)$, it
is enough to test whether the only entry of $\tab_r$ is at most~$k$.

We now present an algorithm to fill the tables and prove its
correctness.  First, we initialize all table entries of all tables
by~$\infty$. By \emph{updating the entry
  $\tab_x(\bagsol,\bagreach{},\bagpart)$ with~$m$} we mean setting
$\tab_x(\bagsol,\bagreach{},\bagpart):= m$ if
$m < \tab_x(\bagsol,\bagreach{},\bagpart)$. For each leaf node~$x$, it
is obviously correct to set~$\tab_x(\bagsol,\bagreach{},\bagpart)=0$ for
the only pattern $(\bagsol,\bagreach{},\bagpart)$ at~$x$, which has the
empty graph as~$\bagsol$ and~$\bagreach{}$ and the empty set
as~$\bagpart$.
In the following, for each type of a node~\(x\) of a tree decomposition, that is, for forget nodes, introduce nodes, and join nodes, we independently show how to compute the table~$\tab_x$ given that we correctly computed the tables for all children of~$x$.  To show that the table \(\tab_x\) is filled correctly, we prove the following lemma for each node type.
\begin{lemma}\label{lem:ulobo}\leavevmode
  \begin{enumerate}[(i)]
  \item \label{upperbound}
    There is a partial \partSet{} for~$G[U_x]$ satisfying a
    pattern~$(\bagsol,\bagreach{},\bagpart)$ at~$x$ with weight at
    most~$\tab_x(\bagsol,\bagreach{},\bagpart)$.
  \item \label{lowerbound}
    The minimum weight of a partial \partSet{} for~$G[U_x]$ satisfying a
    pattern~$(\bagsol,\bagreach{},\bagpart)$ at~$x$ is at
    least~$\tab_x(\bagsol,\bagreach{},\bagpart)$.
  \end{enumerate}
\end{lemma}

\noindent We present the algorithm and the proof for \autoref{lem:ulobo} independently for each node type in Sections~\ref{sec:fnod}, \ref{sec:inod}, and \ref{sec:jnod}, respectively, where we assume that all tables~$\tab_y$ for child nodes~$y$ of~$x$ have been computed correctly.

\subsubsection{Forget nodes}\label{sec:fnod}

We use the following procedure to compute the table~\(\tab_x\) of a forget node~\(x\) under the assumption that the table~$\tab_y$ for the child node~$y$ of~$x$ has been computed correctly.

\begin{proc}[Forget node]\label{forgproc}\upshape
  Let $x$~be a forget node with a single child $y$. Assume that $v$~is
  the vertex being ``forgotten'', that is, $v$~is in the child bag~$V_y$
  but not in the current bag~$V_x$. From the weights of optimal partial
  \partSet{}s for~$G[U_y]$, we want to compute the weight of optimal
  partial \partSet{}s for~$G[U_x]$.

  To this end, for each pattern $(\bagsol,\bagreach{},\bagpart)$
  for~$y$, we distinguish four cases. In each case, we will construct a
  pattern~$(\bagsol',\bagreach{}',\bagpart')$ for~$x$ such that a
  partial \partSet{} for~$G[U_y]$ that
  satisfies~$(\bagsol,\bagreach{},\bagpart)$ is a partial
  \partSet{} for~$G[U_x]$ and
  satisfies~$(\bagsol',\bagreach{}',\bagpart')$. Then, we update
  $\tab_x(\bagsol',\bagreach{}',\bagpart')$ with the value
  of~$\tab_y(\bagsol,\bagreach{},\bagpart)$. 
Herein, the following case distinction is not exhaustive.  We do not
  take action for patterns~$(\bagsol,\bagreach{},\bagpart)$ that do not
  satisfy any of the following conditions (for the reasons informally
  explained in the cases).
In all cases, we set
  $\bagsol':=\bagsol-\{v\}$.
  \begin{caselist}
  \item\label{forgcase1} If $v$~is isolated in~$\bagreach{}$ and there is
    a set~$\{v\}$ in~$\bagpart$, then we let $\bagreach{}':=\bagreach{}
   -\{v\}$ and $\bagpart':=\bagpart\setminus\{\{v\}\}$ and update
    $\tab_x(\bagsol',\bagreach{}',\bagpart')$ with the value
    of~$\tab_y(\bagsol,\bagreach{},\bagpart)$: an isolated vertex that is alone in its part of~$\bagpart$ can simply be forgotten. %

  \item\label{forgcase2} If $v$ is a non-isolated sink in~$\bagreach{}$
    and $v \in \q_i \in \bagpart$ such that $\q_i \subseteq V_y$, then
    we let $\bagreach{}':=\bagreach{}$ and~$\bagpart':=\bagpart$.  We
    update $\tab_x(\bagsol',\bagreach{}',\bagpart')$ with the value
    of~$\tab_y(\bagsol,\bagreach{},\bagpart)$: in this case, the sink~$v$ ``moves'' from~\(V_y\) to %
 $V(\bagreach{}')\setminus V_x$. To ensure that $(\bagsol',\bagreach{}',\bagpart')$ is
    a pattern, the part~\(P_i\) containing~$v$ cannot contain any additional sink in 
    $V(\bagreach{})\setminus V_y$, thus we require \(P_i\subseteq V_y\).

  \item\label{forgcase3} If $v$ is not a sink in $\bagreach{}$ and there
    is no sink in $V(\bagreach{})\setminus V_y$ such that $v$~is its
    only in-neighbor, then let $\bagreach{}':=\bagreach{} - \{v\}$ and
    $\bagpart'$ be the partition of the vertices of~$\bagreach{}'$
    obtained from~$\bagpart$ by removing~$v$ from the set it is
    in. Update $\tab_x(\bagsol',\bagreach{}',\bagpart')$ with the value
    of~$\tab_y(\bagsol,\bagreach{},\bagpart)$.
        This the simplest case, where the vertex is somewhat unimportant to partial partitioning sets satisfying the pattern $(\bagsol,\bagreach{},\bagpart)$ at $y$, so we simply forget it.

  \item\label{forgcase4} If there is a sink~$u\in V(\bagreach{})\setminus V_y$ such that $v$~is its only in-neighbor
    and $\{u,v\}$ is a set of~$\bagpart$, then let
    $\bagreach{}':=\bagreach{} - \{u,v\}$ and $\bagpart'$ be the partition
    of the vertices of~$\bagreach{}'$ obtained from $\bagpart$ by removing
    the set $\{u,v\}$. Update $\tab_x(\bagsol',\bagreach{}',\bagpart')$ with the
    value of~$\tab_y(\bagsol,\bagreach{},\bagpart)$:    
    If there was a sink~$u$ in $V(\bagreach{}) \setminus V_y$ only reachable from~$v$, then it would be unreachable from $V_x$ since $v$~is forgotten. Therefore, if the part~$P_i$ of~$\bagpart$ containing~$u$ and~\(v\) contained more vertices, then we could not be sure that a partial partitioning set satisfying the pattern $(\bagsol,\bagreach{},\bagpart)$ at~$y$ is a partial partitioning set for~$G[U_x]$ at all. Namely, it may break \autoref{def:partsol}(\ref{partsol2}).
  \end{caselist}

\end{proc}

\noindent We show that \autoref{forgproc} fills the table~\(\tab_x\) associated with a forget node~\(x\) correctly.  First, we show that there is a partial \partSet{} for~$G[U_x]$ satisfying a pattern~$(\bagsol,\bagreach{},\bagpart)$ at~$x$ and having weight at most~$\tab_x(\bagsol,\bagreach{},\bagpart)$ as computed by \autoref{forgproc}.

\begin{proof}[of \autoref{lem:ulobo}\eqref{upperbound} for forget nodes]
  Let $x$~be a forget node with child node~$y$ and let $v$~be the vertex
  ``forgotten'', that is, $v$~is in the child bag~$V_y$ but not in the
  current bag~$V_x$.  For any table
  entry~$\tab_x(\bagsol',\bagreach',\bagpart')<\infty$, we show that
  there is a partial \partSet{}~$\Sol{}$ for~$G[U_x]$ satisfying
  $(\bagsol',\bagreach',\bagpart')$ and having weight at
  most~$\tab_x(\bagsol',\bagreach',\bagpart')$. To this end, observe
  that, since $\tab_x(\bagsol',\bagreach',\bagpart')<\infty$ there is a
  pattern~$(\bagsol,\bagreach,\bagpart)$ for~$y$ from which
  \autoref{forgproc} generates~$(\bagsol',\bagreach',\bagpart')$ and
  such that
  $\tab_x(\bagsol',\bagreach',\bagpart')=\tab_y(\bagsol,\bagreach,\bagpart)$.
  Since there is a \partSet{} for~$G[U_y]$ that
  satisfies~$(\bagsol,\bagreach,\bagpart)$ and has weight at
  most~$\tab_y(\bagsol,\bagreach,\bagpart)$, it is sufficient to show
  that \emph{any} partial \partSet{}~$\Sol{}$ for~$G[U_y]$ that
  satisfies the pattern~$(\bagsol,\bagreach{},\bagpart)$ at~$y$ is also
  a partial \partSet{} for~$G[U_x]$ that satisfies at~$x$ the
  pattern~$(\bagsol',\bagreach{}',\bagpart')$ generated in each of the
  cases~(\ref{forgcase1})--(\ref{forgcase4}) of
  \autoref{forgproc}.

  We first argue that~$\Sol{}$ is a partial \partSet{} for~$G[U_x]$ if
  any of the cases~(\ref{forgcase1})--(\ref{forgcase4}) of
  \autoref{forgproc} applies. We first verify
  \autoref{def:partsol}(\ref{partsol1}).  To this end, observe that by
  \autoref{patsat}, a vertex~$u \in V_y$ is a sink
  in~$\bagreach{}_y(\Sol{})$ if and only if it is a sink
  in~$G[U_y]\setminus\Sol{}$. Now, assume that there is a connected
  component of~$G[U_x]\setminus\Sol{}=G[U_y]\setminus\Sol{}$ that
  contains two different sinks~$s_1,s_2$ in~$U_x \setminus V_x$. Then,
  one of these sinks, say~$s_1$, must be~$v$. Since, then, $v\in V_y$~is
  a sink in~$G[U_x]\setminus\Sol{}=G[U_y]\setminus\Sol{}$, it is a sink
  in $\bagreach{}=\bagreach{}_y(\Sol{})$ and none of the
  cases~(\ref{forgcase3}) and~(\ref{forgcase4}) apply.  Moreover, since
  $s_2$~is connected to~$v\in V_y$ in~$G[U_y]\setminus\Sol{}$, by
  \autoref{def:partsol}(\ref{partsol2}), some vertex in~$V_y$ can
  reach~$s_2$, implying that $s_2$~is a vertex
  of~$\bagreach{}=\bagreach_y(\Sol{})$. Thus, by
  \autoref{patsat}(\ref{patsat2}), $s_2$~is in the same
  set~$\q_i\in\bagpart$ as $s_1=v$ and, hence, (\ref{forgcase1})~does
  not apply. Since $s_2\notin V_y$, also (\ref{forgcase2})~does not
  apply.

  We now verify \autoref{def:partsol}(\ref{partsol2}). It can only be
  violated if $v$~is the only vertex of~$V_y$ that can reach some
  sink~$u$ in the connected component of~$v$ in
  $G[U_y]\setminus\Sol{}$. However, then, $v$~is the only in-neighbor
  of~$u$ in~$\bagreach{}_y(\Sol{})=\bagreach{}$. Hence, only
  case~(\ref{forgcase4}) might become applicable. When this case
  applies, however, $\{u,v\}\in\bagpart{}$ implies that no vertex
  in~$V_y\supseteq V_x$ is connected to~$v$ or~$u$
  in~$G[U_y]\setminus\Sol=G[U_x]\setminus\Sol$. Thus,
  \autoref{def:partsol}(\ref{partsol2}) is satisfied.

  It remains to show that $\Sol{}$ satisfies the generated
  pattern~$(\bagsol',\bagreach{}',\bagpart')$, that is, to verify
  $\bagsol'=\bagsol_x(\Sol{})$ (\autoref{patsat}(\ref{patsatsol})),
  $\bagreach{}'=\bagreach{}_x(\Sol{})$ (\autoref{patsat}(\ref{patsat1}))
  and that~$\bagpart'$ is a coarsening of $\bagpart_x(\Sol{})$
  (\autoref{patsat}(\ref{patsat2})). Herein,
  $\bagsol'=\bagsol-\{v\}=\bagsol_y(\Sol)-\{v\}=\bagsol_x(\Sol)$ is
  trivial. To show $\bagreach{}'=\bagreach{}_x(\Sol{})$, we
  distinguish between the case of \autoref{forgproc} applied.

  \begin{caselist}
  \item[Case~\ref{forgcase1})~~] In this case, $v$ is not in~$V_x$ and,
    obviously, not a \interesting{} sink
    in~$G[U_x]\setminus\Sol{}$. Hence, $v$~is not
    in~$\bagreach_x(\Sol{})$. Moreover, $v$~is isolated
    in~$\bagreach{}=\bagreach{}_y(\Sol{})$.
    Therefore, \autoref{forgproc} sets
    $\bagreach{}':=\bagreach-\{v\}=\bagreach{}_y(\Sol{})-\{v\} =
    \bagreach{}_x(\Sol{})$.

  \item[Case~\ref{forgcase2})~~] In this case, $v$ is not in $V_x$ but
    it is a \interesting{} sink in $G[U_x]\setminus\Sol{}$, since it is
    not isolated in~$\bagreach{}$. Therefore, \autoref{forgproc} sets
    $\bagreach':=\bagreach=\bagreach{}_y(\Sol{})=\bagreach{}_x(\Sol{})$.

  \item[Case~\ref{forgcase3})~~] In this case, $v$ is not a sink in
    $\bagreach{}$ (and thus also not in $G[U_x]\setminus\Sol{}$) and,
    therefore, clearly does not appear in~$\bagreach{}_x(\Sol{})$. Moreover, any sink not in~$V_y$ is
    reachable from a vertex in~$V_y\setminus\{v\}=V_x$
    in~$\bagreach{}_y(\Sol{})$ if and only if it is reachable
    in~$\bagreach{}_y(\Sol{})-\{v\}$. Hence, \autoref{forgproc} sets
    $\bagreach':= \bagreach{} - \{v\} = \bagreach{}_y(\Sol{}) - \{v\} =  \bagreach{}_x(\Sol{})$.

  \item[Case~\ref{forgcase4})~~] In this case, $u$~was a \interesting{}
    sink in $G[U_y]\setminus\Sol{}$ but is not \interesting{} in
    $G[U_x]\setminus\Sol{}$. Moreover, since $\{v,u\}\in\bagpart$, no
    vertex of~$V_x$ is connected to~$v$ in
    $G[U_y]\setminus\Sol=G[U_x]\setminus\Sol{}$. Hence, neither~$v$
    nor~$u$ are vertices of~$\bagreach_x(\Sol{})$. Hence,
    \autoref{forgproc} sets $\bagreach{}':=\bagreach{} -
    \{u,v\}=\bagreach{}_y(\Sol{})-\{u,v\}=\bagreach{}_x(\Sol{})$.
  \end{caselist}

  Finally, we verify \autoref{patsat}(\ref{patsat2}) by showing that
  $\bagpart'$~is a coarsening of~$\bagpart_x(\Sol{})$. Assume the
  contrary. Then, there are two vertices~$u,w$ of~$\bagreach{}_x(\Sol{})$
  in the same set of~$\bagpart_x(\Sol{})$ but in different sets of
  $\bagpart'$. By construction of~$\bagpart'$ from~$\bagpart$, they are
  also in different sets of~$\bagpart=\bagpart_y(\Sol{})$. It follows
  that~$u$ and~$w$ lie in the same connected component
  of~$G[U_x]\setminus\Sol{}$ but in different connected components
  of~$G[U_y]\setminus\Sol{}$. Since these two graphs are the same, we have a contradiction.\qed
\end{proof}

\noindent We now show that the minimum weight of a partial \partSet{} for~$G[U_x]$
satisfying a pattern~$(\bagsol,\bagreach{},\bagpart)$ at~$x$ is at
least~$\tab_x(\bagsol,\bagreach{},\bagpart)$ as computed by \autoref{forgproc}.

\begin{proof}[of \autoref{lem:ulobo}\eqref{lowerbound} for forget nodes]
  Let $x$~be a forget node with child node~$y$. Let $v$~be the vertex ``forgotten'' that is, $v\in V_y$ but $v\notin V_x$. Assume that $\Sol{}$~is a partial \partSet{} for~$G[U_x]$ satisfying the pattern~$(\bagsol_x,\bagreach{}_x, \bagpart_x)$ at~$x$. It is sufficient to construct a pattern~$(\bagsol,\bagreach{},\bagpart{})$ that $\Sol{}$~satisfies at~$y$ and from which \autoref{forgproc} generates exactly the pattern~$(\bagsol_x,\bagreach{}_x,\bagpart_x)$ to update the table $\tab_x(\bagsol_x,\bagreach{}_x,\bagpart_x)$ with $\tab_y(\bagsol,\bagreach{},\bagpart)$. Then, \autoref{lem:ulobo}\eqref{lowerbound} follows for forget nodes, because we have $\tab_x(\bagsol_x,\bagreach{}_x,\bagpart_x)\leq \tab_y(\bagsol,\bagreach{},\bagpart) \leq \wei{}(\Sol{})$. Herein, the last inequality follows from the induction hypothesis.

  We first show that $\Sol{}$ is a partial \partSet{} for~$G[U_y]$, that
  is, we verify \autoref{def:partsol}.
  \autoref{def:partsol}(\ref{partsol1}) is easy to verify: since
  $\Sol{}$ is a partial \partSet{} for~$G[U_x]$,
  each connected component of
  $G[U_x]\setminus\Sol{}=G[U_y]\setminus\Sol{}$ contains at most one
  sink in~$U_x\setminus V_x\supseteq U_y \setminus V_y$. It remains to
  verify
  \autoref{def:partsol}(\ref{partsol2}). %
  Assume, for a contradiction, that there is a connected component~$C$ in~$G[U_y]\setminus\Sol{}$ that contains a vertex of~$V_y$ such that no vertex of~$V_y\supseteq V_x$ can reach some sink~$s\in C\setminus V_y$. Then, since $\Sol{}$ is a partial \partSet{} for~$G[U_x]\setminus\Sol{}$, the connected component~$C$ cannot contain vertices of~$V_x$ and, hence, $C \cap V_y=\{v\}$. However, since $G[U_y]\setminus\Sol{}$~is a directed acyclic graph, $v$~reaches some sink in~$C$. Since $v$~cannot reach~$s\in C$, it follows that $C$~contains two sinks. Since $C\cap V_x=\emptyset$, this contradicts~$\Sol{}$ being a partial \partSet{} for~$G[U_x]\setminus\Sol{}$. It follows that $\Sol{}$ is a partial
  \partSet{} for~$G[U_y]$.

  We now construct a pattern~$(\bagsol,\bagreach,\bagpart)$ that~$\Sol{}$ satisfies at~$y$.  Consider $\bagsol:=\bagsol{}_y(\Sol{})$ and $\bagreach{} := \bagreach{}_y(\Sol{})$.  %
  Note that there are at most two vertices in~$\bagreach{}$ that are not in~$\bagreach{}_x$: one of them is~$v$, %
  the possibly other vertex is a sink~$u$ in $V(\bagreach{})\setminus V_y$ only reachable from~$v$.  We define $\bagpart$~as a partition of the vertices of~$\bagreach{}$ that partitions the set $V(\bagreach{}_x)$ in the same way as~$\bagpart_x$. We add the possibly missing vertices~$v$ and~$u$ to that partition as follows: if there is a vertex~$w\in V_x$ in the same connected component of $G[U_y]\setminus\Sol{} = G[U_x]\setminus\Sol{}$ as~$v$, then we put~$v$ and~$u$ into same set as~$w$. Otherwise, we add the set~$\{v\}$ or~$\{u,v\}$, respectively, to~$\bagpart$.  By choice of~$\bagsol,\bagreach{},$ and~$\bagpart$, the partial
  \partSet{}~$\Sol{}$ clearly satisfies the
  pattern~$(\bagsol,\bagreach{},\bagpart)$ at~$y$.

  We have shown that $\Sol{}$~satisfies the pattern~$(\bagsol,\bagreach{}, \bagpart{})$ at~$y$.  Moreover, if any of the cases (\ref{forgcase1})--(\ref{forgcase4}) of \autoref{forgproc} applies to~$(\bagsol,\bagreach{},\bagpart)$, then it generates a pattern~$(\bagsol',\bagreach{}',\bagpart')$ with $\bagsol'=\bagsol_x$ and $\bagreach{}'=\bagreach{}_x$, since we showed in the proof of \autoref{lem:ulobo}\eqref{upperbound} for forget nodes that $\Sol{}$~satisfies the pattern generated by \autoref{forgproc} at~$x$.  Hence, it remains to show that indeed at least one of the cases (\ref{forgcase1})--(\ref{forgcase4}) of \autoref{forgproc} applies and that in all cases~$\bagpart'=\bagpart_x$.
  \begin{caselist}
  \item[Case \ref{forgcase1})~~] If $v$ is an isolated sink
    in~$\bagreach{}$, then no vertex in~$V_y\supseteq V_x$ can reach~$v$
    in~$G[U_y]\setminus\Sol{}=G[U_x]\setminus\Sol{}$. Hence, there is no
    vertex of~$V_x$ in the same connected component of
    $G[U_x]\setminus\Sol{}$ as~$v$, as otherwise $v$~would be a sink
    in~$U_x \setminus V_x$ not reachable from the vertices
    of~$V_x$. Hence, by construction of~$\bagpart$, we have
    $\{v\}\in\bagpart$ and case~(\ref{forgcase1}) of \autoref{forgproc}
    applies. It sets $\bagpart':=\bagpart \setminus
    \{\{v\}\}=\bagpart_x$. %

  \item[Case \ref{forgcase2})~~] If $v$~is a non-isolated sink
    in~$\bagreach{}$, then $v$~is a \interesting{} sink in
    $G[U_x]\setminus\Sol{}$. Hence, it is contained in $\bagreach{}_x$
    and we have $V(\bagreach{}_x)=V(\bagreach)$. By construction
    of~$\bagpart$, we also have $\bagpart_x=\bagpart$.  The
    set~$\q_i\in\bagpart_x$ containing~$v\notin V_x$ cannot contain any other
    vertex in~$V(\bagreach{}_x) \setminus V_x$ by
    \autoref{pat}(\ref{pat2}). Thus, $\q_i\subseteq V_y$ and
    case~(\ref{forgcase2}) of \autoref{forgproc} applies. It
    sets~$\bagpart':=\bagpart=\bagpart_x$. %

  \item[Case \ref{forgcase3})~~] If $v$~is not a sink in~$\bagreach{}$
    and there is no sink in~$V(\bagreach{}) \setminus V_y$ only
    reachable from~$v$ in~$\bagreach{}$, then case~(\ref{forgcase3}) of
    \autoref{forgproc} applies.  Since the sink~$s$ reachable from~$v$
    is also reachable from some vertex~$u\notin\{v,s\}$, and thus,
    connected to~$u$ in~$G[U_y]\setminus\Sol{}$, the set in~$\bagpart$
    containing~$v$ also contains~$u$. \autoref{forgproc}
    sets~$\bagpart'$ to be~$\bagpart$ with $v$~removed from the set it
    is in. This, by construction of~$\bagpart$, is exactly~$\bagpart_x$.

  \item[Case \ref{forgcase4})~~] Finally, if there is a sink~$u$
    in~$V(\bagreach{}) \setminus V_y$ only reachable from $v$, then the
    connected component of~$G[U_y]\setminus\Sol{}=G[U_x]\setminus\Sol{}$
    containing the vertex~$v$ does not contain any vertex of~$V_x$,
    since $u$~is not reachable from any vertex of~$V_x$.  It follows
    that $\{u,v\}$~is a set of~$\bagpart$. Case~(\ref{forgcase4}) of
    \autoref{forgproc} applies. It sets
    $\bagpart':=\bagpart\setminus\{\{v,u\}\}=\bagpart_x$.\qed
  \end{caselist}
\end{proof}

\subsubsection{Introduce nodes}\label{sec:inod}

We use the following procedure to compute the table~\(\tab_x\) of an introduce node~\(x\) under the assumption that the table~$\tab_y$ for the child node~$y$ of~$x$ has been computed correctly.

\begin{proc}[Introduce node]\label{intproc}\upshape Let $x$~be an
  introduce node with a single child~$y$. Assume that~$v$ is the node being ``introduced'', that is, $v$~is not in the child bag~$V_y$ but in the current bag~$V_x$. Moreover, let $B\subseteq A(G[U_x])$~be the set of arcs incident to~$v$. By \autoref{treedec}(\ref{treedec2} and \ref{treedec3}) of a tree decomposition, one actually has $B\subseteq A(G[V_x])$.

  We now try each possible subset~$B'\subseteq B$ and consider it \emph{not} deleted by a partial \partSet{} for the graph~$G[U_x]$. Similarly as in the case for forget nodes, we will transform each pattern~$(\bagsol,\bagreach{},\bagpart)$ for~$y$ into a pattern~$(\bagsol',\bagreach{}',\bagpart')$ for~$x$ such that if a partial
  \partSet{}~$\Sol{}$ for~$G[U_y]$
  satisfies~$(\bagsol,\bagreach{},\bagpart)$, then
  $\Sol{}\cup(B\setminus B')$ is a partial \partSet{} for~$G[U_x]$ and
  satisfies~$(\bagsol',\bagreach{}',\bagpart')$. Then, we update
  $\tab_x(\bagsol',\bagreach{}',\bagpart')$ with the value
  of~$\tab_y(\bagsol,\bagreach{},\bagpart)+\wei(B)-\wei(B')$.

  For each pattern~$(\bagsol,\bagreach{},\bagpart)$ for~$y$ such that all vertices incident to the arcs in~$B'$ (if any) except for~$v$ are contained in the same set~$\q_i\in\bagpart$, we obtain~$\bagsol'$ from~$\bagsol$ by adding~$v$ and the arcs in~$B'$ to~$\bagsol$. Similarly, we obtain $\bagreach'$ from~$\bagreach$ by adding~$v$ and the arcs in~$B'$ to~$\bagreach$. 
  Moreover, for each~$u,w \in V(\bagreach{}')$ such that $u$~can reach~$w$ in~$\bagreach{}'$ we add the arc~$(u,w)$ to~$\bagreach{}'$. For obtaining~$\bagpart'$, we distinguish two cases.

  \begin{caselist}
  \item If $B'=\emptyset$, then we try all possibilities of adding~$v$
    to a set in~$\bagpart$. That is, for every~$\q_i \in \bagpart$, we
    get a set~$\bagpart'$ from $\bagpart$ by adding~$v$ to~$\q_i$ and
    update $\tab_x(\bagsol',\bagreach{}',\bagpart')$ with
    $\tab_y(\bagsol,\bagreach{},\bagpart) + \wei{}(B)$.  Additionally,
    for $\bagpart':=\bagpart\cup\{\{v\}\}$, we update the
    entry~$\tab_x(\bagsol',\bagreach{}',\bagpart')$ with
    $\tab_y(\bagsol,\bagreach{},\bagpart) + \wei{}(B)$.

  \item If $B'\ne\emptyset$, then let $\q_i$~be the set of~$\bagpart$
    that contains all vertices incident to arcs in~$B'$ except~$v$ and
    let $\bagpart'$~be obtained from~$\bagpart$ by adding~$v$ to the
    set~$\q_i$. We update $\tab_x(\bagsol',\bagreach{}',\bagpart')$ with
    $\tab_y(\bagsol,\bagreach{},\bagpart) + \wei{}(B) - \wei{}(B')$.
  \end{caselist}
  
  Note that, since $\bagpart$ simulates the connected components of the resulting graph, 
  all arcs incident on~$v$ remaining in the graph must be  within one set of $\bagpart$, i.e., their endpoints different from $v$ must be in one set of $\bagpart$.
\end{proc}

\noindent We show that \autoref{intproc} fills the table associated
with an introduce node~\(x\) correctly.  First, we show that there is a
partial \partSet{} for~$G[U_x]$ satisfying a
pattern~$(\bagsol,\bagreach{},\bagpart)$ at~$x$ and having weight at
most~$\tab_x(\bagsol,\bagreach{},\bagpart)$ as computed by
\autoref{intproc}.

\begin{proof}[of \autoref{lem:ulobo}\eqref{upperbound} for introduce nodes]
  Let $x$~be an introduce node with child node~$y$ and let~$v$ be the vertex ``introduced'' that is, $v$~is not in the child bag~$V_y$ but in the current bag~$V_x$. Let $B\subseteq A(G[V_x])$ be the arcs incident to~$v$, $B'\subseteq B$ and, finally, $(\bagsol,\bagreach{},\bagpart)$~be some pattern for~$y$ such that all vertices incident to the arcs in~$B'$ (if any) except for~$v$ are contained in the same set~$\q_i\in\bagpart$. For any partial \partSet{} $\Sol{}$ for~$G[U_y]$ satisfying the pattern~$(\bagsol,\bagreach{},\bagpart)$ at~$y$, we show that $\Sol{}'= \Sol{} \cup (B \setminus B')$ is a partial \partSet{} for~$G[U_x]$ that satisfies the pattern~$(\bagsol',\bagreach{}',\bagpart')$ constructed by \autoref{intproc}. From this, since $\wei(\Sol{}')=\wei(\Sol{}) + \wei{}(B) - \wei{}(B')$, \autoref{lem:ulobo}\eqref{upperbound} follows (as already discussed in the beginning of the proof of \autoref{lem:ulobo}\eqref{upperbound} for forget nodes).

  We start by showing that~$\Sol{}'$ is a partial \partSet{} for~$G[U_x]$. First, we verify \autoref{def:partsol}(\ref{partsol1}).  For the sake of a contradiction, assume that there is a connected component of~$G[U_x]\setminus\Sol{}'$ that contains two distinct sinks~$s_1,s_2$ in~$U_x \setminus V_x$. Since there is no such connected component in~$G[U_y]\setminus\Sol{} = (G[U_x]\setminus\Sol{}') - \{v\}$, there are vertices~$s_1', s_2'\in V_y$ in the same connected components of $G[U_y]\setminus\Sol{}$ as~$s_1$ and $s_2$, respectively, that are incident to some arcs in~$B'$. By \autoref{def:partsol}(\ref{partsol2}), there are vertices in~$V_y$ that can reach $s_1$ and $s_2$ in~$G[U_y]\setminus\Sol{}$. Hence, $s_1$~and~$s_2$ are \interesting{} and, therefore, in $\bagreach{}=\bagreach_y(\Sol)$.  Since $s_1',s_2'\in\q_i$ and $\Sol{}$~satisfies $(\bagsol,\bagreach{},\bagpart)$, by \autoref{pat}(\ref{pat2}), we also have $s_1,s_2\in\q_i$. Then, however, $\q_i$~contains the two different vertices~$s_1$ and~$s_2$ of $V(\bagreach{}) \setminus V_y$, which contradictions \autoref{pat}(\ref{pat2}).

  To show that $\Sol{}'$ is a partial \partSet{}, it remains to verify
  \autoref{def:partsol}(\ref{partsol2}). For the sake of contradiction,
  assume that some connected component of~$G[U_x]\setminus\Sol{}'$
  contains some sink~$s\in U_x\setminus V_x$, some vertex in~$V_x$, but
  $s$~is not reachable from any vertex in~$V_x$. Then, $s$~is not
  reachable from any vertex in~$V_y\subseteq V_x$ in the
  subgraph~$G[U_y]\setminus\Sol{}$ either. Thus, the connected component
  does not contain any vertex of~$V_y$ and, therefore, not of~$V_x$,
  since the only vertex in~$V_x\setminus V_y$ is~$v$ and the added
  arcs~$B'$ connect only vertices in~$V_y$.

  We have shown that~$\Sol{}'$ is a partial \partSet{} for~$G[U_x]$. We
  now show that it satisfies the
  pattern~$(\bagsol',\bagreach{}',\bagpart{}')$ generated by
  \autoref{intproc}; we verify \autoref{patsat}.
  \autoref{patsat}(\ref{patsatsol}), that is,
  $\bagsol'=\bagsol_x(\Sol{}')$ is trivial by the construction
  of~$\bagsol'$. 

  We verify \autoref{patsat}(\ref{patsat1}), that is,
  $\bagreach{}_x(\Sol{}') = \bagreach{}'$. First, observe
  that~$V(\bagreach{}_x(\Sol{}'))=V(\bagreach_y(\Sol{}))\cup\{v\}=V(\bagreach')$. We
  have to show that %
  there is an arc~$(u,w)$ in~$\bagreach'$ %
  if and only if~$u$ can reach~$w$ in~$G[U_x]\setminus\Sol{}'$. Let
  $(u,w)$~be such an arc in~$\bagreach{}'$. If $(u,w)$ is already
  in~$\bagreach{}$, then $u$ can reach~$w$ in~$G[U_y]\setminus\Sol{} =
  (G[U_x]\setminus\Sol{}') - \{v\}$. Otherwise, $u$ reaches~$w$
  in~$\bagreach{}'$ via some arcs~$(u',v),(v,w')\in B'$, that is, $u$
  can reach~$u'$ and $w'$ can reach~$w$ in~$G[U_x]\setminus\Sol{}'$. It
  follows that $u$ can reach $w$ in $G[U_x]\setminus\Sol{}'$. Now, for
  the opposite direction, let $u,w$ be vertices of~$\bagreach'$ such
  that $u$~can reach~$w$ in~$G[U_x]\setminus\Sol{}'$. If $u$ can
  reach~$w$ in~$(G[U_x]\setminus\Sol{}')-\{v\}=G[U_y]\setminus\Sol{}$,
  then the arc $(u,w)$~is already present in~$\bagreach{}$. Otherwise,
  $u$ reaches~$w$ via some arcs~$(u',v),(v,w')\in B'$. The arcs~$(u',v)$
  and~$(v,w')$ are in~$\bagreach{}'$, since $u',w'\in V_x$. Moreover,
  $u$~reaches $u'$ and $w'$ reaches~$w$
  in~$G[U_x]\setminus\Sol{}'$. Hence, there are arcs~$(u,u')$ and
  arc~$(w',w)$ in~$\bagreach{}'$ and~$u$ reaches~$w$ in~$\bagreach{}'$
  via $u'$ and~$w'$. By construction of~$\bagreach{}'$, it follows that
  $\bagreach{}'$~contains the arc~$(u,w)$.

  Finally, we verify \autoref{patsat}(\ref{patsat2}); we show that~$\bagpart'$ is a coarsening of $\bagpart_x(\Sol{}')$. For the sake of a contradiction, assume that there are two vertices~$u,w$ that are in the same set of~$\bagpart_x(\Sol{}')$ but in different sets of~$\bagpart'$. By construction of~$\bagpart'$ from~$\bagpart$, this implies that $u$~and~$w$ are in different sets of~$\bagpart$ and, therefore, in different connected components of~$G[U_y]\setminus\Sol$. Thus, in order for $u$~and~$w$ to be connected in~$G[U_x]\setminus\Sol'$, there are vertices~$u', w'$ in the same connected components of~$G[U_y]\setminus\Sol$ as~$u$ and~$w$, respectively, that are incident to arcs in~$B'$ and, hence, $u',w'\in \q_i\in\bagpart$.  But then, also $u,w \in \q_i\in\bagpart{}$ --- a~contradiction.\qed
\end{proof}

\noindent We now show that the minimum weight of a partial \partSet{} for~$G[U_x]$
satisfying a pattern~$(\bagsol,\bagreach{},\bagpart)$ at~$x$ is at
least~$\tab_x(\bagsol,\bagreach{},\bagpart)$ as computed by by
\autoref{intproc}.

\begin{proof}[of \autoref{lem:ulobo}\eqref{lowerbound} for introduce nodes]
  Let $x$~be an introduce node with child node~$y$. Let $v$~be the vertex ``introduced'', that is, $v\notin V_y$ but $v\in V_x$. Assume that $\Sol{}$~is a minimum-weight partial \partSet{} for~$G[U_x]$ satisfying the pattern~$(\bagsol_x,\bagreach{}_x, \bagpart_x)$ at~$x$. Let $B$~be the set of arcs incident to~$v$ in~$G[U_x]$ and~$B'':=B\cap\Sol{}$. It is sufficient to construct a pattern~$(\bagsol,\bagreach{},\bagpart{})$ that~$\Sol{}\setminus B''$ satisfies at~$y$ and from which \autoref{intproc} generates exactly the pattern~$(\bagsol_x,\bagreach{}_x,\bagpart{}_x)$ to update the table~$\tab_x(\bagsol_x,\bagreach{}_x,\bagpart{}_x)$ with~$\tab_y(\bagsol,\bagreach{},\bagpart{})+\wei(B)-\wei(B')$, where~$B'=B\setminus B''$. Then, \autoref{lem:ulobo}\eqref{lowerbound} follows for introduce nodes, since $\tab_x(\bagsol_x,\bagreach{}_x,\bagpart{}_x) \leq \tab_y(\bagsol, \bagreach{},\bagpart{}) + \wei{}(B) -\wei{}(B') \leq \wei{}(\Sol{} \setminus B'') + \wei{}(B) -\wei{}(B') =\wei{}(\Sol{})$.

  It is easy to verify that~$\Sol{}\setminus B''$ is a partial \partSet{} for~$G[U_y]$ (\autoref{def:partsol}), since $G[U_y]\setminus(\Sol{}\setminus B'')=(G[U_x]\setminus\Sol{})-\{v\}$ and $\Sol{}$ is a partial \partSet{} for~$G[U_x]$; to this end, observe that, by \autoref{treedec}(\ref{treedec2} and \ref{treedec3}) of a tree decomposition, $v$~only has arcs~$B\subseteq A(G[V_x])$ incident to vertices in~$V_y$.

  We now construct a pattern. Let $\bagsol=\bagsol{}_y(\Sol{}\setminus B'')$ and $\bagreach{}=\bagreach{}_y(\Sol{}\setminus B'')$.  Let $\bagpart$~be the partition obtained from~$\bagpart_x$ by removing the vertex~$v$ from the set it is in or by removing the set $\{v\}$ if it exists in~$\bagpart{}_x$. It is easy to verify that $\Sol{} \setminus B''$ satisfies~$(\bagsol,\bagreach{},\bagpart)$ at~$y$: \autoref{patsat}(\ref{patsatsol}) and~(\ref{patsat1}) are trivially satisfied by choice of~$\bagsol$ and~$\bagreach{}$; \autoref{patsat}(\ref{patsat2}) holds by construction of~$\bagpart{}$ from~$\bagpart{}_x$, since $G[U_y]\setminus(\Sol{}\setminus B'')$ is a subgraph of~$G[U_x]\setminus\Sol{}$.

  It remains to show that \autoref{intproc} applies to the pattern~$(\bagsol,\bagpart{},\bagreach{})$ and the set~$B'$ in order to generate the pattern~$(\bagsol_x,\bagpart{}_x,\bagpart{}_x)$. Since~$\Sol{}$ satisfies~$(\bagsol_x,\bagreach{}_x,\bagpart{}_x)$ at~$x$, all vertices incident to arcs in~$B'$ (if any) are contained in the same set~$\q_i\in\bagpart_x$ and, hence, all of them except~$v$ are contained in the set~$\q_i \setminus \{v\}\in\bagpart$. Therefore, \autoref{intproc} applies to~$B'$ and the pattern~$(\bagsol,\bagreach{},\bagpart)$, produces some new pattern~$(\bagsol',\bagreach{}',\bagpart')$, and updates $\tab_x(\bagsol',\bagreach{}',\bagpart')$ with $\tab_y(\bagsol,\bagreach{},\bagpart{}) + \wei{}(B) -\wei{}(B')$.

  It remains to show that, for at least one of the generated patterns, $\bagsol'=\bagsol_x$, $\bagreach{}'=\bagreach{}_x$, and $\bagpart'=\bagpart_x$. If $B' \neq \emptyset$ then $\bagpart'=\bagpart_x$ by construction of~$\bagpart'$ from~$\bagpart{}$ in \autoref{intproc}. If $B' = \emptyset$, then $\bagpart_x$ is clearly among the partitions~$\bagpart'$ generated from~$\bagpart{}$ by \autoref{intproc}.  Moreover, we already proved in the proof of \autoref{lem:ulobo}\eqref{upperbound} that~$\Sol{}=(S\setminus B'') \cup (B \setminus B')$ satisfies the pattern generated by \autoref{intproc} at~$x$. Hence, $\bagsol'=\bagsol_x$ and $\bagreach{}'=\bagreach_x$.  \qed
\end{proof}

\subsubsection{Join nodes}\label{sec:jnod}

We use the following procedure to compute the table~\(\tab_x\) of a join node~\(x\) under the assumption that the tables~$\tab_y$ for all child nodes~$y$ of~$x$ have been computed correctly.

\begin{proc}[Join node]\label{joiproc}\upshape
  Let $x$~be a join node with children $y$ and $z$, that is, $V_x=V_y=V_z$.  %
  For each pair of patterns~$(\bagsol,\bagreach{}_y,\bagpart_y)$ for~$y$ and $(\bagsol,\bagreach{}_z,\bagpart_z)$ for~$z$ such that~$\bagpart_y$ and~$\bagpart_z$ partition the vertices of~$V_y=V_z=V_x$ in the same way, we construct a new pattern~$(\bagsol,\bagreach,\bagpart)$ as follows.

  Let $\bagreach{}'$~be the graph containing all vertices and arcs of~$\bagreach_y$ and $\bagreach_z$, and 
  for each~$u,w \in V(\bagreach{}')$ such that $u$~can reach~$w$ in~$\bagreach{}'$ add the arc~$(u,w)$ to~$\bagreach{}'$.
  Note that by
  \autoref{treedec}(\ref{treedec3}) of a tree decomposition,
  $\bagreach_y$ and~$\bagreach_z$ have only the vertices in~$V_x$ in
  common.

  Let $\bagpart'$~be the partition of~$V_x$ that partitions~$V_x$ in the
  same way as~$\bagpart_y$ and~$\bagpart_z$. We extend~$\bagpart'$ to a
  partition for the vertices of~$\bagreach{}'$: for each $u \in
  V(\bagreach{}') \setminus V_x$, add~$u$ to a set~$\q_i$ of~$\bagpart'$
  that contains a vertex~$v$ with $(v,u)$ being an arc
  of~$\bagreach{}'$.  Since there are no arcs between different sets
  of~$\bagpart$ in~$\bagreach{}_y$ or~$\bagreach{}_z$, there is exactly
  one such set~$\q_i\in\bagpart'$.

  If we created some set~$\q \in \bagpart'$ with more than one vertex
  of~$V(\bagreach{}') \setminus V_x$, then continue with a different
  pair of patterns.  Otherwise, we update
  $\tab_x(\bagsol,\bagreach{}',\bagpart')$ with
  $\tab_y(\bagsol,\bagreach{}_y,\bagpart_y) +
  \tab_z(\bagsol,\bagreach{}_z,\bagpart_z)-\wei{}(A(G[V_x]))+\wei{}(A(\bagsol))$.
\end{proc}

\noindent We show that \autoref{joiproc} fills the table associated with a join node~\(x\) correctly.  First, we show that there is a partial \partSet{} for~$G[U_x]$ satisfying a pattern~$(\bagsol,\bagreach{},\bagpart)$ at~$x$ and having weight at most~$\tab_x(\bagsol,\bagreach{},\bagpart)$ as computed by \autoref{joiproc}.

\begin{proof}[of \autoref{lem:ulobo}\eqref{upperbound} for join nodes]
  Let $x$ be a join node with child nodes~$y$ and~$z$, that is,
  $V_x=V_y=V_z$. Let $\Sol{}_y$ be a partial
  \partSet{} for~$G[U_y]$ satisfying the pattern~$(\bagsol,\bagreach{}_y,\bagpart_y)$ at~$y$ and let $\Sol{}_z$~be a partial \partSet{} for~$G[U_z]$ satisfying the pattern~$(\bagsol,\bagreach{}_z,\bagpart_z)$ at $z$.  We show that~$\Sol{} = \Sol{}_y \cup \Sol{}_z$ is a partial \partSet{} for~$G[U_x]$ that satisfies the pattern~$(\bagsol,\bagreach{}',\bagpart')$ constructed by \autoref{joiproc}. Since $\wei(\Sol{})=\wei(\Sol{}_y)+\wei(\Sol{}_z)-\wei{}(\Sol{}_y \cap \Sol{}_z)$, wherein $\Sol{}_y \cap \Sol{}_z = A(G[V_x]) \setminus A(\bagsol)$, \autoref{lem:ulobo}\eqref{upperbound} follows for join nodes.

  We show that~$\Sol{}$ is indeed a partial \partSet{} for~$G[U_x]$,
  that is, we verify \autoref{def:partsol}. We first verify
  \autoref{def:partsol}(\ref{partsol2}) and then use it to verify
  \autoref{def:partsol}(\ref{partsol1}).  Let $s \in U_x \setminus V_x$
  be a sink such that the connected component containing~$s$
  in~$G[U_x]\setminus\Sol{}$ contains a vertex of~$V_x$. Then, $s\in
  U_y\setminus V_x$ or $s\in U_z\setminus V_x$. Without loss of
  generality, let $s\in U_y\setminus V_x$.  From
  \autoref{treedec}(\ref{treedec3}) of a tree decomposition, we see
  that~$U_z\cap U_y\subseteq V_x$ and, hence, $G[U_y]\setminus\Sol{}_y=
  G[U_x]\setminus\Sol{}-(U_z\setminus V_z)$. It follows that there is
  also a connected component of~$G[U_y]\setminus\Sol{}_y=
  G[U_x]\setminus\Sol{}-(U_z\setminus V_z)$ that contains~$s$ and a
  vertex of~$V_x$ and, therefore, $s$~is reachable from some
  vertex~$v\in V_x$ in~$G[U_x]\setminus\Sol{}-(U_z\setminus V_z)$ and,
  hence, in~$G[U_x]\setminus\Sol{}$. It also follows that $(v,s)$~is an
  arc in~$\bagreach_y$ and, by construction in \autoref{joiproc},
  of~$\bagreach'$.

  To verify \autoref{def:partsol}(\ref{partsol1}), for the sake of a
  contradiction, assume that there is a connected component
  of~$G[U_x]\setminus\Sol{}$ that contains two sinks~$s_1,s_2$ in~$U_x
  \setminus V_x$. Note that, by \autoref{treedec}(\ref{treedec3}) of a
  tree decomposition, there are no arcs between~$U_y\setminus V_x$
  and~$U_z\setminus V_x$. Hence, this connected component contains a
  vertex~$v$ of $V_x$; otherwise, it would be a connected component with
  two sinks outside of~$V_x$ already in either~$G[U_y]\setminus\Sol{}_y$
  or~$G[U_z]\setminus\Sol{}_z$. Thus, as seen in the previous paragraph,
  we have arcs~$(s_1',s_1)$~and~$(s_2',s_2)$ with $s_1',s_2'\in V_x$.
  It follows by construction of~$\bagpart'$ from~$\bagreach'$ in
  \autoref{joiproc} that $s_1$~and~$s_1'$ are in a
  set~$\q_i\in\bagpart'$ and $s_2$~and~$s_2'$ are in a
  set~$\q_j\in\bagpart'$. We show~$i=j$, which contradicts the
  construction of~$\bagpart'$, since then~$\q_i=\q_j$ contains two
  vertices~$s_1\notin V_x$ and~$s_2\notin V_x$.

  Since $s_1$~and~$s_2$ are in the same connected component
  of~$G[U_x]\setminus\Sol{}$, also $s_1'$~and~$s_2'$ are, since they can
  reach~$s_1$ and~$s_2$, respectively. Hence, there is an \upath{}~$p$
  between~$s_1$ and~$s_2$ in~$G[U_x]\setminus\Sol{}$. It consists of
  consecutive path segments~$p'$ that only have their endpoints~$u,w$
  in~$V_x$ (possibly, such a path segment only consists of one arc). It
  follows that such a path segment~$p'$ is entirely contained in
  $G[U_y]\setminus\Sol{}_y$ or $G[U_z]\setminus\Sol{}_y$ and, hence, its
  endpoints $u$~and~$w$ are in the same set of~$\bagpart_y$
  or~$\bagpart_z$. Since $u,w\in V_x$, by construction of~$\bagpart'$ in
  \autoref{joiproc}, $u$~and~$w$ are in the same set of~$\bagpart'$. It
  follows that~$s_1'$ and~$s_2'$ are in the same set of~$\bagpart'$, and
  so are~$s_1$ and~$s_2$.

  It follows that $\Sol{}$ is indeed a partial \partSet{} for~$G[U_x]\setminus\Sol{}$. It remains to verify that~$\Sol{}$ satisfies the pattern~$(\bagsol,\bagreach',\bagpart')$ (\autoref{patsat}). Herein, \autoref{patsat}(\ref{patsatsol}), $\bagsol{}=\bagsol_x(\Sol{})$, is trivial.  We verify~(\ref{patsat1}), that is, $\bagreach{}_x(\Sol{})=\bagreach{}'$. Herein, $V(\bagreach_x(\Sol{}))\subseteq V(\bagreach')$ we already verified when verifying \autoref{def:partsol}(\ref{partsol2}). Now, assume that there are two vertices $u,w$ in~$\bagreach'$ such that $u$~can reach~$w$ in $G[U_x]\setminus\Sol{}$. Since, then, $u$~is not a sink, it is in~$V_x$. The \dpath{} from~$u$ to~$w$ consists of consecutive subpaths, each being entirely contained in~$G[U_y]\setminus\Sol{}_y$ or~$G[U_z]\setminus\Sol{}_z$ and thus, causing an arc in~$\bagreach_y$ or~$\bagreach_z$ and, therefore, in~$\bagreach'$. It follows that~$u$ can reach~$w$ in~$\bagreach'$, which therefore has an arc~$(u,w)$.  In the opposite direction, for every arc~$(u,w)$ in~$\bagreach'$ that is already in~$\bagreach_y$ or~$\bagreach_z$, there is an \dpath{} in either $G[U_y]\setminus\Sol{}_y$ or in~$G[U_z]\setminus\Sol{}_z$ from~$u$ to~$w$ and, thus, $u$~can reach~$w$ in $G[U_x]\setminus\Sol{}$. For an arc~$(u,w)$ in~$\bagreach'$ that is neither present in~$\bagreach_y$ nor~$\bagreach_z$, there is an \dpath{} in~$\bagreach'$ from~$u$ to~$w$ consisting only of arcs that are already present in~$\bagreach_y$ or~$\bagreach_z$. Since we have seen that for each such arc there is a corresponding \dpath{} in~$G[U_x]\setminus\Sol{}$, we have that $u$~can reach~$w$ in~$G[U_x]\setminus\Sol{}$.

  For \autoref{patsat}(\ref{patsat2}), it has been shown above that if
  two vertices of~$\bagreach'$ are in the same connected component
  of~$G[U_x]\setminus\Sol{}$, then they are in the same set
  in~$\bagpart'$.  \qed
\end{proof}

\noindent We now show that the minimum weight of a partial \partSet{} for~$G[U_x]$
satisfying a pattern~$(\bagsol,\bagreach{},\bagpart)$ at~$x$ is at
least~$\tab_x(\bagsol,\bagreach{},\bagpart)$ as computed by by
\autoref{joiproc}.

\begin{proof}[of \autoref{lem:ulobo}\eqref{lowerbound} for join nodes]
  Let $x$~be a join node with the child nodes~$y$ and~$z$, that
  is~$V_x=V_y=V_z$. Assume that $\Sol{}$ is a minimum-weight partial
  \partSet{} for~$G[U_x]$ satisfying the pattern
  $(\bagsol,\bagreach{}_x, \bagpart_x)$ at~$x$. It is sufficient to
  construct patterns~$(\bagsol,\bagreach{}_y,\bagpart{}_y)$
  and~$(\bagsol,\bagreach{}_z,\bagpart{}_z)$ that are satisfied
  by~$S_y:=\Sol{}\cap A(G[U_y])$ at~$y$ and by~$S_z:=\Sol{}\cap
  A(G[U_z])$ at~$z$, respectively, such that from these patterns
  \autoref{joiproc} generates exactly the
  pattern~$(\bagsol,\bagreach_x,\bagpart_x)$ to update
  $\tab_x(\bagsol,\bagreach_x,\bagpart_x)$ with
  \begin{align*}
    &\tab_y(\bagsol,\bagreach{}_y,\bagpart_y) +
    \tab_z(\bagsol,\bagreach{}_z,\bagpart_z)-\wei{}(A(G[V_x]))
    +\wei{}(A(\bagsol))\\
    &\leq \wei{}(\Sol_y) +\wei{}(\Sol{}_z)
    -\wei{}(\Sol_y\cap\Sol_z)\\
    &= \wei{}(\Sol{}).
  \end{align*}
  We first show that $S_y$~is a partial \partSet{}
  for~$G[U_y]$. Symmetrically, it follows that $S_z$~is a
  partial \partSet{} for~$G[U_z]$.  We first verify
  \autoref{def:partsol}(\ref{partsol1}).  Since by
  \autoref{treedec}(\ref{treedec3}), there are no arcs between vertices
  in~$U_y\setminus V_y$ and~$U_z\setminus V_z$ in~$G[U_x]$, it follows
  from $G[U_y]\setminus S_y = G[U_x]\setminus\Sol{}-(U_z\setminus V_z)$
  that no connected component of~$G[U_y]\setminus S_y$ contains two
  sinks not in~$V_y=V_x$. It remains to verify
  \autoref{def:partsol}(\ref{partsol2}). To this end, let $u$~be a sink
  in~$U_y \setminus V_y$ in a connected component of~$G[U_y]\setminus
  S_y$ containing a vertex of~$V_y$. Then, by
  \autoref{def:partsol}(\ref{partsol2}), the connected component of~$G[U_x]\setminus\Sol{}$ containing~$u$ contains a \dpath{} from some
  vertex in~$V_x$ to~$u$. The subpath of this \dpath{} that contains only
  one vertex of~$V_x$ is preserved in~$G[U_y]\setminus \Sol{}_y$. Hence,
  $u$~is reachable from some vertex of~$V_x=V_y$.

  We now construct the patterns. To this end, let
  $\bagreach{}_y:=\bagreach{}_y(\Sol{}_y)$ and
  $\bagreach{}_z:=\bagreach{}_z(\Sol{}_z)$.  Moreover, we
  choose~$\bagpart_y$ and~$\bagpart_z$ such that they partition the
  set~$V_x=V_y=V_z$ in the same way as~$\bagpart_x$ and such that the
  vertices of~$V(\bagreach{}_y) \setminus V_y$ (or $V(\bagreach{}_z)
  \setminus V_z$) are in the same set as the other vertices of their
  connected components in $G[U_y]\setminus\Sol_y$ (or
  $G[U_z]\setminus\Sol_z$).

  We show that $\Sol{}_y$~satisfies $(\bagsol,\bagreach{}_y,
  \bagpart_y)$ at~$y$. Analogously, it then follows that
  $\Sol{}_z$~satisfies~$(\bagsol,\bagreach{}_z, \bagpart_z)$ at~$z$. We
  verify \autoref{patsat}. Since
  $\bagsol=\bagsol_x(\Sol)=\bagsol_y(\Sol_y)=\bagsol_z(\Sol_z)$ and
  $\bagreach{}_y=\bagreach{}_y(S_y)$ hold by definition, it remains to
  verify \autoref{patsat}(\ref{patsat2}). To this end, observe that
  $G[U_y]\setminus\Sol{}_y=G[U_x]\setminus\Sol{}-(U_y\setminus
  V_y)$. Now, assume, for the sake of a contradiction, that there are
  two vertices~$v,w$ of~$\bagreach_y$ in different sets of~$\bagpart_y$
  but in the same connected component of~$G[U_y]\setminus\Sol{}_y$. It
  follows that $v$~and~$w$ are in the same connected component
  of~$G[U_x]\setminus\Sol{}$.  If $v,w\in V_y$, then, by construction
  of~$\bagpart_y$ from~$\bagpart_x$, the vertices~$v$ and~$w$ are in
  different sets of~$\bagpart_x$, contradicting
  $\Sol{}_x$~satisfying~$(\bagsol,\bagreach{}_x,\bagpart{}_x)$.  If
  exactly one of~$v,w$ is in~$V_y$, then $v$~and~$w$ being in different
  sets of~$\bagpart_y$ contradicts the construction of~$\bagpart_y$.  If
  both $v,w\notin V_y$, then $v$~and~$w$ are two \interesting{} sinks
  in~$G[U_y]$, which contradicts~$v$ and~$w$ being in the same connected
  component of~$G[U_y]\setminus\Sol{}_y$.

  Hence, indeed $\Sol{}_y$ satisfies $(\bagsol,\bagreach{}_y,
  \bagpart_y)$ at~$y$ and $S_z$ satisfies
  $(\bagsol,\bagreach{}_z,\bagpart_z)$ at~$z$. Moreover, since
  $\bagpart_y$ and~$\bagpart_z$ partition~$V_x$ in the same way,
  \autoref{joiproc} applies to the patterns~$(\bagsol,\bagreach{}_y,
  \bagpart_y)$ and~$(\bagsol,\bagreach{}_z,\bagpart_z)$ and produces a
  pattern $(\bagsol,\bagreach{}',\bagpart')$. If no set of~$\bagpart'$
  contains more than one vertex of $V(\bagreach{}') \setminus V_x$, it
  indeed updates $\tab_x(\bagsol,\bagreach{}',\bagpart')$.

  Hence it remains to show that $\bagreach{}'=\bagreach{}_x$ and $\bagpart'=\bagpart_x$, as no set of~$\bagpart_x$ contains two vertices of $V(\bagreach{}_x) \setminus V_x$ by \autoref{pat}(\ref{pat2}). We already showed in the proof of \autoref{lem:ulobo}\eqref{upperbound} for join nodes that $\Sol{}$~satisfies the pattern~$(\bagsol,\bagreach{}',\bagpart{}')$ generated by \autoref{joiproc}. Hence, $\bagreach{}'=\bagreach{}_x$.  Finally, by construction of~$\bagpart'$ in \autoref{joiproc}, the vertices of~$V_x$ are partitioned the same way by~$\bagpart'$ and~$\bagpart_x$. For a vertex~$v \in V(\bagreach{}') \setminus V_x$, there is a vertex~$u$ in~$V_x$ that can reach~$v$ in~$\bagreach{}'$ and, therefore, in~$G[U_x]\setminus\Sol{}$. Hence, $u$ and~$v$ must be in the same set of both $\bagpart_x$ and~$\bagpart'$ by construction of~$\bagpart'$ in \autoref{joiproc}.
  \qed\end{proof}

\subsubsection{Running time}

\noindent Having shown the correctness of the
Procedures~\ref{forgproc}--\ref{joiproc}, we can finally complete the proof of \autoref{th:fpt-tw} by analyzing the running time of the procedures.

\begin{proof}[of \autoref{th:fpt-tw}]
  \autoref{lem:ulobo} showed that the presented dynamic programming algorithm is correct, that is, it solves \DAGP{} given a tree decomposition of the input graph. It remains to analyze the running time.

  To this end, recall that each bag in a tree decomposition of width~$\tw{}$ contains at most~$\tw{}+1$~vertices. This allows us to give an upper bound on the number of possible patterns~$(\bagsol,\bagreach{},\bagpart{})$. There are at most $3^{\binom{\tw+1}{2}}$ directed acyclic graphs~$\bagsol$ on at most $\tw{}+1$ vertices: for each pair~$(v,w)$ of vertices: there is either no arc, or an arc from~$v$ to~$w$, or an arc from~$w$ to~$v$. Similarly, there are at most $3^{\binom{2\tw+2}{2}}$ directed graphs~$\bagreach$ on at most $2\tw{}+2$ vertices.  Moreover, there are at most $(2\tw+2)^{2\tw+2}$ partitions~$\bagpart{}$ of at most $2\tw+2$~vertices into at most $2\tw+2$~sets. Hence, each table has at most $3^{\binom{2\tw+2}{2}}\cdot3^{\binom{\tw+1}{2}}\cdot(2\tw+2)^{2\tw+2}=3^{\bigO(\tw^2+\tw \log \tw)}=2^{\bigO(\tw^2)}$~entries and looking up entries in the tables can be implemented to run in $\bigO(\log 2^{\tw^2})$~time, which is polynomial in~$\tw$.

  In each leaf node, we set the single table entry to~$0$ in constant
  time.

  In each forget node, \autoref{forgproc} iterates over the
  entries of the table of the child node and for each entry spends time
  polynomial in~$\tw$. Thus, it spends a total of $2^{\bigO(\tw^2)}$~time in
  each forget node.

  To analyze the running time \autoref{intproc} spends in an introduce
  node, observe that there are at most $\tw$~arcs in~$A(G[V_x])$
  incident to the introduced vertex~$v$. Hence, there are at most
  $2^\tw$ subsets of them. For each of these subsets and for each entry
  of the child node, \autoref{intproc} spends time polynomial
  in~$\tw$. This makes a total of $2^{\bigO(\tw^2)}$~time spent in each
  introduce node.

  Finally, in a join node, \autoref{joiproc} considers every pair of
  patterns of its two child nodes and for each combination spends time
  polynomial in~$\tw$. Hence, the total time spent in a join node
  is~$(2^{\bigO(\tw^2)})^2=2^{O{(\tw^2)}}$.

  Since the nice tree decomposition has $\bigO(\tw n)$~nodes, the algorithm
  runs in $2^{\bigO(\tw^2)}\cdot n$~time. \qed
\end{proof}

\section{Other parameters yield stronger NP-hardness
  results}\label{sec:classical}

In Sections~\ref{sec:smallsol} and~\ref{sec:tw}, we have seen that
\DAGP{} is solvable in linear time when fixing the weight of the
\partSet{} sought or the treewidth of the input graph. %
The question whether fixed-parameter algorithms can be obtained for
parameters that are smaller than the solution weight or the
treewidth naturally arises~\cite{Nie10,KN12,FJR13}.
One parameter of interest is the maximum vertex outdegree in the graph:
a citation network of journal articles will, for example, have a small
outdegree since journal articles seldom contain more than 50~references.
In this section, however, we will show that, among others, this parameter
being small will not help solving \DAGP{} efficiently.

\citet{AM12} already showed that \DAGP{} remains NP-hard even if the input graph has only two sinks.  We complement this negative result by showing that the problem remains NP-hard even if the diameter or the maximum vertex degree of the input graph are constant. In conclusion, parameters like the number of sinks, the graph diameter or maximum degree cannot lead to fixed-parameter algorithms unless P${}={}$NP.

\begin{theorem} \label{thm:diam2} \DAGP{} is solvable in linear time on
  graphs of diameter one, but NP-complete on graphs of diameter two even if all arcs have unit weight.
\end{theorem}

\begin{proof}
  On graphs of diameter one, the problem is trivial: a directed acyclic
  graph with diameter one is an acyclic \emph{tournament}, that is,
  there is no pair of vertices not joined by an arc. As such, it already
  contains exactly one source and one sink. Hence, we just verify in
  linear time whether the input graph is an acyclic tournament and
  answer ``yes'' or ``no'' accordingly.

  For graphs of diameter two, we show NP-hardness by means of a
  polynomial-time many-one reduction from \DAGP, which is NP-hard even
  when all arcs have weight one. Therefore, we agree on all arcs in this
  proof having weight one.

  Given an arbitrary instance~$(G,\wei,k)$ of \DAGP, we add a gadget
  to~$G$ to obtain in polynomial time an instance~$(G',\wei', k')$ such
  that~$G'$ is a graph of diameter two and such that~$(G,\wei,k)$ is a
  yes-instance if and only if~$(G',\wei',k')$ is a yes-instance.  We
  obtain a graph~$G'$ from~$G$ by adding an acyclic tournament~\(T\)
  consisting of $k+n+2$~vertices and outgoing arcs from the source~$s$
  of~\(T\) to all vertices of~$V(G)$ in~$G'$.  We set~$k':= k +
  n$. Since every vertex in~$G'$ is in distance one from~$s$, the
  constructed graph~$G'$ has diameter two.

  If \((G,\wei,k)\)~is a yes-instance, then let \(S\)~be a \partSet{}
  with \(k\)~arcs for~\(G\).  A \partSet{}~\(S'\) with \(k'\)~arcs
  for~\(G'\) is obtained by adding to~\(S\) the \(n\)~arcs from the
  source~\(s\) of~\(T\) to all vertices of~\(V(G)\) in~\(G'\).  Thus,
  \((G',\wei',k)\)~is a yes-instance.

  If \((G',\wei',k')\)~is a yes-instance, then let \(S'\)~be a \partSet{}
  with \(k'\)~arcs for~\(G'\).  By \autoref{lem:dirundirequiv}, every
  vertex in~$V(G)$ reaches at least one sink in~\(G'\setminus S'\).
  This sink cannot be the sink of~\(T\), since no vertex in~\(T\) is
  reachable from~\(V(G)\).  Thus, \(S'\)~has to disconnect the sink
  of~\(T\) from all vertices of~$V(G)$ in~$G'$, where all paths between
  the sink of~\(T\) and~\(V(G)\) are via the source~\(s\) of~\(T\).
  Since \(T\) has $k+n+2$~vertices, \(S'\)~cannot disconnect~\(s\) from
  the sink of~\(T\) and thus, has to remove from~\(G'\) the $n$~arcs
  connecting~$s$ to the vertices of~$V(G)$.  Then, the remaining
  \(k\)~arcs in~\(S'\) have to be a \partSet{} for~\(G'\) without the
  tournament~\(T\), which is precisely the original graph~$G$. Thus,
  $(G,\wei,k)$~is a yes-instance. \qed
\end{proof}

\begin{theorem}
  \label{thm:maxdeg3}
  \DAGP{} is solvable in linear time on graphs of maximum degree two,
  but NP-complete on graphs of maximum degree three even if all arcs have unit weight.
\end{theorem}
\begin{proof}
  \newcommand{\Veins}{X}
  \newcommand{\Vzwei}{Y}
  \newcommand{\Vdrei}{Z}
  Any graph of maximum degree two consists of \ucycle{}s or \upath{}s. Thus, the
  underlying graph has treewidth at most two. We have seen in
  \autoref{th:fpt-tw} that \DAGP{} is linear-time solvable when the
  treewidth of the input graph is bounded by a constant.

  We prove the NP-hardness on graphs of maximum degree three. To this
  end, we adapt the polynomial-time many-one reduction from \MC{} to
  \DAGP{} presented by \citet{LBK09}. In their reduction, we replace
  vertices of degree greater than three by equivalent structures of
  lower degree.

  \decprob{\MC}{%
    An undirected graph $G=(V,E)$, a weight function $\wei:E \to \N$,
    a set~$T \subseteq V$ of terminals, and an integer $k$.  
  }{%
    Is there a subset $\Sol{} \subseteq E$ with $\sum_{e \in \Sol{}} \wei(e)
    \leq k$ such that the removal of~$\Sol{}$ from $G$ disconnects
    each terminal from all the others?
  }
  \noindent We first recall the reduction from \MC{} to \DAGP{}. From a
  \MC{} instance~$I_1:=(G_1,\wei_1,T,k_1)$, we construct in polynomial
  time a \DAGP{} instance~$I_2:=(G_2,\wei_2,k_2)$ such that $I_1$~is a
  yes-instance if and only if~$I_2$ is. From $I_2$, we then obtain an
  instance~$I_3$ with maximum degree three. Since \MC{} remains NP-hard
  even for three terminals and unit weights~\cite{DJP+94}, we may
  assume $|T|=3$ and, similarly as in the proof of \autoref{thm:diam2},
  we agree on all arcs in this proof having weight one.  We now
  construct the \DAGP{} instance~$I_2=(G_2,\wei_2,k_2)$
  from~$I_1=(G_1,\wei_1,T,k_1)$ as follows.  The construction is
  illustrated in \autoref{fig:mc-to-dagp}.  

  \begin{figure}[t] \centering
    \begin{tikzpicture}
      \tikzstyle{vertex}=[circle,draw,fill=black,minimum size=3pt,inner sep=0pt]

      \node[vertex, label=above:$t_1$] (t1) {};
      \node[vertex, label=above:$v_1$] (v1) at (1,0) {};
      \node[vertex, label=above:$t_2$] (t2) at (2,0) {};
      \node[vertex, label=above:$v_2$] (v2) at (3,0) {};
      \node[vertex, label=above:$t_3$] (t3) at (4,0) {};

      \draw(t1) -- (v1);
      \draw[dotted] (v1) -- (t2);
      \draw (t2) -- (v2);
      \draw[dotted] (v2) -- (t3);
      \draw[dotted] (t1) to[out=-45,in=-135] (t2);
    \end{tikzpicture}

    \bigskip\noindent
    \begin{tikzpicture}[x=3cm, shorten >= 0.5mm]
      \tikzstyle{vertex}=[circle,draw,fill=black,minimum size=3pt,inner sep=0pt]

      \node[vertex, label=above:$t_1$] (t1) {};
      \node[vertex, label=above:$v_1$] (v1) at (1,0) {};
      \node[vertex, label=above:$t_2$] (t2) at (2,0) {};
      \node[vertex, label=above:$v_2$] (v2) at (3,0) {};
      \node[vertex, label=above:$t_3$] (t3) at (4,0) {};
      
      \node[vertex, label=below:$s_1$] (r1) at (0,-1) {};
      \node[vertex, label=below:$s_2$] (r2) at (2,-1) {};
      \node[vertex, label=below:$s_3$] (r3) at (4,-1) {};

      \node[vertex, label=above:$e_{\{t_1,v_1\}}$] (e1) at (0.5,1) {};
      \node[vertex, label=above:$e_{\{v_1,t_2\}}$] (e2) at (1.5,1) {};
      \node[vertex, label=above:$e_{\{t_2,v_2\}}$] (e3) at (2.5,1) {};
      \node[vertex, label=above:$e_{\{v_2,t_3\}}$] (e4) at (3.5,1) {};
      \node[vertex, label=above:$e_{\{t_1,t_2\}}$] (e5) at (1,2) {};

      \draw[->] (t1) -> (r1);
      \draw[->] (t2) -> (r2);
      \draw[->] (t3) -> (r3);

      \draw[->] (v1) -> (r1);
      \draw[->,dotted] (v1) -> (r2);
      \draw[->,dotted] (v1) -> (r3);

      \draw[->,dotted] (v2) -> (r1);
      \draw[->] (v2) -> (r2);
      \draw[->,dotted] (v2) -> (r3);

      \draw[->] (e1) -> (t1);
      \draw[->] (e1) -> (v1);
      \draw[->,dotted] (e2) -> (v1);
      \draw[->] (e2) -> (t2);
      \draw[->] (e3) -> (t2);
      \draw[->] (e3) -> (v2);
      \draw[->] (e4) -> (v2);
      \draw[->,dotted] (e4) -> (t3);

      \draw[->,shorten >= 0.35cm,dotted] (e5) to[out=180,in=90] (t1);
      \draw[->,shorten >= 0.35cm] (e5) to[out=0,in=90] (t2);

      \node[label=below right:$\Vdrei{}$] (recta) at (-0.75,2.5) {};
      \node[label=below right:$\Vzwei{}$] (rectb) at (-0.75,0.6) {};
      \node[label=below right:$\Veins{}$] (rectc) at (-0.75,-0.5) {};

      \begin{pgfonlayer}{background}
        \path[fill=black!10,rounded corners] (recta)
        rectangle (4.25,0.7);

        \path[fill=black!10,rounded corners] (rectb)
        rectangle (4.25,-0.4);
        
        \path[fill=black!10,rounded corners] (rectc)
        rectangle (4.25,-1.5);
      \end{pgfonlayer}

    \end{tikzpicture}

    \caption{Reduction from a \MC{} instance with the terminals~$t_1,t_2$,
      and~$t_3$ to \DAGP{}. The top shows an instance~$I_1$ of
      \MC{}, where the dotted edges are a \mcset{} of size~$k_1=3$. The
      bottom shows the corresponding instance~$I_2$ of \DAGP{},
      where the dotted arcs are a corresponding \partSet{} of
      size~$k_2=k_1+2(n-3)=7$ ($n$~is the number of vertices in the
      graph of the \MC{} instance). The constructed vertex
      sets~$\Veins{},\Vzwei{},$ and $\Vdrei{}$ are highlighted using a gray
      background.}
    \label{fig:mc-to-dagp}
  \end{figure}
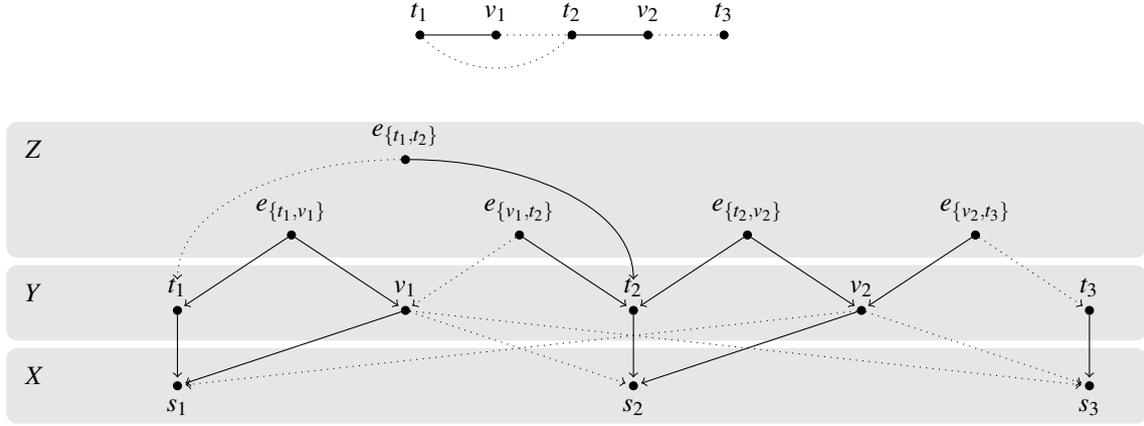

  \begin{enumerate}
  \item Add three vertices $s_1,s_2,s_3$ to~$G_2$, forming the vertex
    set~$\Veins{}$,
  \item add each vertex of~$G_1$ to~$G_2$, forming the vertex
    set~$\Vzwei{}$,
  \item for each edge~$\{u,v\}$ of~$G_1$, add a vertex $e_{\{u,v\}}$
    to~$G_2$, forming the vertex set~$\Vdrei{}$,
  \item for each terminal $t_i \in T$, add the arc~$(t_i,s_i)$ to~$G_2$,
  \item for each vertex~$v \in \Vzwei{} \setminus T$, add the
    arcs~$(v,s_i)$ for $i = 1,2,3$ to~$G_2$, and
  \item for each edge~$\{u,v\}$ of~$G_1$, add the arcs~$(e_{\{u,v\}},
    u)$ and~$(e_{\{u,v\}}, v)$ to~$G_2$.
  \end{enumerate}
  Set $k_2 = k_1 + 2(n-3)$, where $n$~is the number of vertices
  of~$G_1$. We claim that $I_1$ is a yes-instance if and only if $I_2$
  is a yes-instance.

  First, suppose that there is a \mcset{}~$\Sol{}$ of size at most~$k_1$
  for~$G_1$. Then, we obtain a \partSet{} of size at most~$k_2$ for
  $G_2$ as follows: if a vertex~$v$ belongs to the same connected
  component of $G_1\setminus\Sol{}$ as terminal~$t_i$, then remove every
  arc $(v,s_j)$ with~$j \neq i$ from~$G_2$. Furthermore, for each
  edge~$\{u,v\} \in \Sol{}$, remove \emph{either} the arc~$(e_{\{u,v\}},
  u)$ or the arc~$(e_{\{u,v\}}, v)$ from~$G_2$. One can easily check
  that we end up with a valid \partSet{} of size $k_2=k+2(n-3)$ for~$G_2$:
  we delete at most $k$~arcs from~$\Vdrei{}$ to~$\Vzwei{}$ and, for each
  of the $n-3$~vertices in~$\Vzwei{}\setminus T$, delete two arcs
  from~$\Vzwei{}$ to~$\Veins{}$. There are no arcs from~$\Veins{}$
  to~$\Vdrei{}$.

  Conversely, suppose that we are given a minimal \partSet{}~$\Sol{}$ of
  size at most~$k_2$ for~$G_2$. Note that it has to remove at least two
  of the three outgoing arcs of each vertex~$v_2 \in \Vzwei{}\setminus
  T$ but cannot remove all three of them: contrary to
  \autoref{lem:no_new_sinks}, this would create a new sink. Thus,
  $\Sol{}$~deletes $2(n-3)$~arcs from~$\Vzwei{}$ to~$\Veins{}$ and the
  remaining $k_2-2(n-3)=k_1$~arcs from~$\Vdrei{}$
  to~$\Vzwei{}$. Therefore, we can define the following \mcset{} of
  size~$k_1$ for $G_1$: remove an edge $\{u,v\}$ from~$G_1$ if and only
  if one of the arcs~$(e_{\{u,v\}}, u)$ and~$(e_{\{u,v\}}, v)$ is
  removed from~$G_2$ by~$\Sol{}$.  Again, one can easily check that we
  end up with a valid \mcset{}.

  \begin{figure}[t]
    \begin{center}
      \begin{tikzpicture}[shorten >= 0.5mm]
        \tikzstyle{vertex}=[circle,draw,fill=black,minimum size=3pt,inner sep=0pt]
        \tikzstyle{neighbor}=[circle,draw,fill=black,minimum size=3pt,inner sep=0pt]
        \tikzstyle{annot} = [text width=4em, text centered]
        
        \node[neighbor, label=below:$s_1$] (U-1) at ( 0, 2) {};
        \node[neighbor, label=below:$s_2$] (U-2) at ( 1, 2) {};
        \node[neighbor, label=below:$s_3$] (U-3) at ( 2, 2) {};
        
        \node[neighbor, label=right:$v$] (V) at (1,4) {};

        \node[neighbor, label=above:$w_1$] (B-0) at (-0.5,8) {};
        \node[neighbor, label=above:$w_2$] (B-1) at (0.5,8) {};
        \node[neighbor, label=above:$w_3$] (B-2) at (1.5,8) {};
        \node[neighbor, label=above:$w_4$] (B-3) at (2.5,8) {};

        \draw[->] (B-0) to[out=-90,in=126] (V);
        \draw[->] (B-1) to[out=-90,in=102] (V);
        \draw[->] (B-2) to[out=-90,in=78] (V);
        \draw[->] (B-3) to[out=-90,in=54] (V);

        \draw[->] (V) to[out=-126,in=90] (U-1);
        \draw[->] (V) edge (U-2);
        \draw[->] (V) to[out=-54,in=90] (U-3);
      \end{tikzpicture}\hspace{2cm}
      \begin{tikzpicture}[node distance=1cm, shorten >= 0.5mm]
        \tikzstyle{vertex}=[circle,draw,fill=black,minimum size=3pt,inner sep=0pt]
        \tikzstyle{neighbor}=[circle,draw,fill=black,minimum size=3pt,inner sep=0pt]
        \tikzstyle{annot} = [text width=4em, text centered]

        \node[neighbor, label=above:$w_1$] (C-1) {};
        \node[neighbor, label=above:$w_2$] (C-2) [right of=C-1] {};
        \node[neighbor, label=above:$w_3$] (C-3) [right of=C-2] {};
        \node[neighbor, label=above:$w_4$] (C-4) [right of=C-3] {};

        \node[neighbor, label=left:$w_2'$] (C-5) [below of=C-2] {};
        \node (C-6-1) [below of=C-3] {};
        \node[neighbor, label=left:$w_3'$] (C-6) [below of=C-6-1] {};
        \node (C-7-1) [below of=C-4] {};
        \node (C-7-2) [below of=C-7-1] {};
        \node[neighbor, label=left:$w_4'$] (C-7) [below of=C-7-2] {};

        \node[neighbor, label=left:$v$] (C-8) [below of=C-7] {};

        \node (C-12-1) [below of=C-8] {};
        \node[neighbor, label=below:$s_3$] (C-12) [below of=C-12-1] {};

        \node[neighbor, label=left:$v'$] (C-9)  [left of=C-12-1] {};

        \node[neighbor, label=below:$s_2$] (C-11) [below of=C-9] {};
        \node[neighbor, label=below:$s_1$] (C-10) [left of=C-11] {};

        \node[annot, xshift=0.5cm] (T) [right of=C-9] {$T_{v}$};

        \foreach \src / \dest in {1/5, 2/5, 3/6, 4/7, 5/6, 6/7, 7/8, 8/9, 8/12, 9/10, 9/11}
        \draw[->] (C-\src) edge (C-\dest);

        \begin{pgfonlayer}{background}

          \path[fill=black!10,rounded corners]
          (0.5,-3.5) rectangle (4,-6.5);
        \end{pgfonlayer}

      \end{tikzpicture}
    \end{center}
    \vspace{-15pt}
    \caption{Reduction of the degree of a vertex~$v$ to three. On the
      left, the original neighborhood of~$v$ is shown. The right side
      shows~$v$ after modification. The tree structure~$T_{v}$
      constructed in the proof is highlighted using a gray background.}
    \label{fig:tree-structure}
  \end{figure}
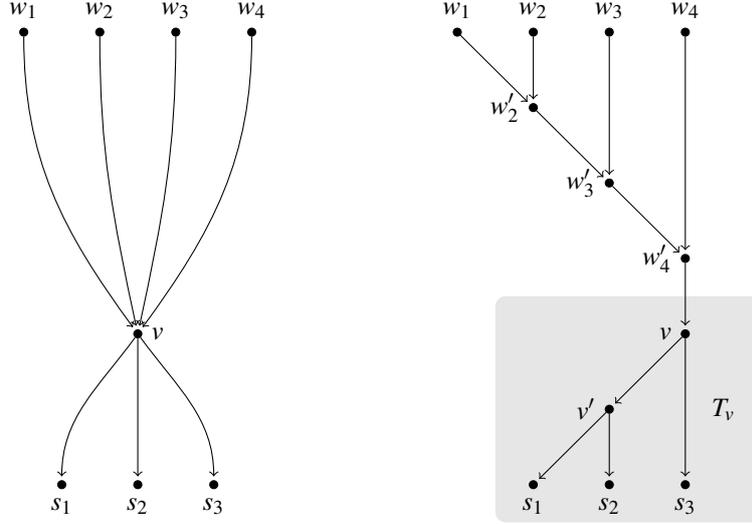
      
  It remains to modify the instance $I_2=(G_2,\wei_2,k_2)$ to get an
  instance $I_3=(G_3,\wei_3,k_3)$ of maximum degree three.  To this end,
  first we show how to reduce the outdegree of each vertex of~$G_2$ to
  two.  Thereafter, we show how to reduce the indegree of each vertex of~$G_2$
  to one by introducing gadget vertices, each having
  indegree two and outdegree one.  The construction is illustrated in
  \autoref{fig:tree-structure}.

  Note that all vertices of~$G_2$ with outdegree larger than two are
  in~$\Vzwei{}$. In order to decrease the degree of these vertices, we
  obtain a graph~$G_3'$ from~$G_2$ by carrying out the following
  modifications (see \autoref{fig:tree-structure}) to~$G_2$: for each
  vertex~$v \in \Vzwei{}$, with $N^+(v) = \{s_1,s_2,s_3\}$, remove
  $(v,s_1)$ and $(v,s_2)$ add a new vertex~$v'$, and insert the three
  arcs~$(v,v')$, $(v',s_1)$, and~$(v',s_2)$.

  \looseness=-1We show that~$(G_3',\wei_3',k_2)$ is a yes-instance if and only
  if~$(G_2,\wei_2,k_2)$~is. To this end, for~$v\in\Vzwei{}$, let $T_{v}$~be the
  induced subgraph~$G_3'[\{v,v', s_1,s_2,s_3\}]$.  In~$G_2$, a
  minimal \partSet{} removes exactly two of the outgoing arcs of~$v$,
  since~$s_1,s_2$, and~$s_3$ are sinks.  It is enough to show that a
  minimal \partSet{}~$S$ removes exactly two arcs in~$T_{v}$ from~$G_3'$
  in such a way that there remains exactly one \dpath{} from~$v$ to exactly
  one of~$s_1,s_2$, or~$s_3$. This remaining \dpath{} then one-to-one
  corresponds to the arc that a \partSet{} would leave in~$G_2$
  between~$v$ and~$s_1,s_2$, or~$s_3$.  Since $s_1,s_2$, and~$s_3$ are
  sinks, $\Sol{}$ indeed has to remove at least two arcs from~$T_{v}$:
  otherwise, two sinks will belong to the same connected component.
  However, due to \autoref{lem:no_new_sinks}, $\Sol{}$ cannot remove
  more than two arcs from~$T_{v}$.  Moreover, again exploiting
  \autoref{lem:no_new_sinks}, the two arcs removed by~$\Sol{}$ leave a
  single \dpath{} from~$v$ to exactly one of~$s_1,s_2$, or~$s_3$.

  We have seen that $(G_3',\wei_3',k_2)$ is equivalent to~$(G_2,\wei_2,k_2)$ and that
  all vertices of~$G_3'$ have outdegree two. To shrink the overall
  maximum degree to three, it remains to reduce the indegrees. 
  Note that the vertices newly introduced in the previous step already have indegree one. 
  We obtain
  graph~$G_3$ of maximum degree three from~$G_3'$ as follows.  For each
  vertex~$v$ with $|N^-(v)| = |\{w_1,\ldots,w_{d^-(v)}\}| \geq 2$, do
  the following (see \autoref{fig:tree-structure}): for $j = 2, \ldots,
  d^-(v)$, remove the arc~$(w_j,v)$ and add a vertex~$w_j'$ together
  with the arc~$(w_j,w_j')$. Moreover, add the arcs $(w_{1},w_{2}')$,
  $(w_{d^-(v)},v)$, and~$(w_j',w_{j+1}')$ for each $j\in\{2, \ldots,
  d^-(v)-1\}$.  Now, every vertex of~$V(G_3')$ in~$G_3$ has indegree
  one and outdegree two, while the newly introduced vertices have
  indegree two and outdegree one. It follows that all vertices in~$G_3$
  have degree at most three. 
  
  It remains to show that~$(G_3,\wei_3,k_2)$ is a yes-instance if and
  only if~$(G_3',\wei_3',k_2)$ is. It then follows that
  $(G_3,\wei_3,k_2)$ is a yes-instance if and only if~$(G_2,\wei_2,k_2)$
  is.  To this end, note that by \autoref{lem:no_new_sinks}, among the
  introduced arcs, only the arcs~\(w_1,w'_2\) and~$(w_j,w_j')$ can be
  removed by a minimal \partSet{}.  From this, there is a one-to-one
  correspondence between deleting the arc \(w_1,w'_2\) or $(w_j,w_j')$
  in the graph~$G_3$ and deleting the arc~$(w_j,v)$ in the
  graph~$G_3'$. \qed
\end{proof}

\section{Outlook}\label{sec:outlook}

We have presented two fixed-parameter algorithms for \DAGP, one with
respect to the weight~$k$ of the \partSet{} sought and one with respect
to the parameter treewidth~$\tw$.

We demonstrated the feasibility of applying the fixed-parameter algorithm for the
parameter~$k$ (\autoref{alg:simple-st}) to large input instances with optimal
\partSet{}s of small weight. However, we were unable to solve the
instances in the data set of~\citet{LBK09}, since the weight of
optimal \partSet{}s is too large. We found out that the heuristic
presented by \citet{LBK09} finds nearly optimal \partSet{}s on the
instances that our algorithm works on best. However, we have also seen
that one does not have to specially craft adversarial instances to make
the heuristic perform badly.  Surprisingly, our algorithm for \DAGP{} is
much simpler and more efficient than the algorithm for \MC{} by
\citet{Xia10}, although \MC{} is much easier to approximate than \DAGP{}
\citep{KKS+04, AM12}.

On the theoretical side, we improved a fixed-parameter algorithm by
\citet{AM12} such that the running time now depends on the treewidth of the input graph rather than on its pathwidth. 
However, our algorithm, as well as the algorithm
of~\citet{AM12}, are practically inapplicable.

Towards solving the instances of \citet{LBK09} exactly in reasonable
time, a challenging task would be to analyze the data used by
\citet{LBK09} in order to find parameters that are small \emph{and} make
\DAGP{} fixed-parameter tractable, that is, to take a data-driven
approach to parameterizing \DAGP{}.

\section*{Acknowledgments}
\noindent  We are thankful to Rolf Niedermeier and to the anonymous referees
of CIAC'13 and \emph{Discrete Applied Mathematics} for helpful comments.

René van Bevern acknowledges supported by the Russian Foundation for
Basic Research (RFBR), project~16-31-60007
mol\textunderscore{}a\textunderscore{}dk, while working at Novosibirsk
State University, and by the German Research Foundation (DFG), project
DAPA (NI 369/12), while working at TU Berlin. Robert Bredereck
acknowledges support by DFG project PAWS (NI~369/10). Morgan Chopin was
supported by the DAAD during a three-month visit to TU~Berlin in
summer~2012.  Falk Hüffner acknowledges support by DFG projects PABI
(NI~369/7) and ALEPH (HU~2139/1).  Ondřej Suchý acknowledges support by
DFG project AREG (NI~369/9) while at TU~Berlin and by the Czech Science
Foundation, project 14-13017P.

\section*{References}
\bibliographystyle{abbrvnat}
\bibliography{dag}
\end{document}